\documentclass[english]{article}
\usepackage[T1]{fontenc}
\usepackage[latin9]{inputenc}
\usepackage{geometry}
\usepackage{color}
\usepackage{babel}
\usepackage{float}
\usepackage{units}
\usepackage{mathrsfs}
\usepackage{url}
\usepackage{algorithm2e}
\usepackage{amsmath}
\usepackage{amsthm}
\usepackage{amssymb}
\usepackage{stmaryrd}
\usepackage{graphicx}
\usepackage[pdfusetitle,
 bookmarks=true,bookmarksnumbered=true,bookmarksopen=true,bookmarksopenlevel=3,
 breaklinks=true,pdfborder={0 0 0},pdfborderstyle={},backref=false,colorlinks=true]
 {hyperref}

\makeatletter

\providecommand{\tabularnewline}{\\}

\newcommand{\lyxaddress}[1]{
	\par {\raggedright #1
	\vspace{1.4em}
	\noindent\par}
}
\theoremstyle{plain}
\newtheorem{thm}{\protect\theoremname}
\theoremstyle{definition}
\newtheorem{example}[thm]{\protect\examplename}
\theoremstyle{plain}
\newtheorem{prop}[thm]{\protect\propositionname}
\theoremstyle{definition}
\newtheorem{defn}[thm]{\protect\definitionname}
\theoremstyle{plain}
\newtheorem{lem}[thm]{\protect\lemmaname}
\theoremstyle{plain}
\newtheorem{cor}[thm]{\protect\corollaryname}
\theoremstyle{remark}
\newtheorem{rem}[thm]{\protect\remarkname}

\LinesNumbered
\usepackage{accents}
\usepackage{tikz}
\usepackage{circuitikz}

\ifdefined\showcaptionsetup
 \PassOptionsToPackage{caption=false}{subfig}
\fi
\usepackage{subfig}
\makeatother

\providecommand{\corollaryname}{Corollary}
\providecommand{\definitionname}{Definition}
\providecommand{\examplename}{Example}
\providecommand{\lemmaname}{Lemma}
\providecommand{\propositionname}{Proposition}
\providecommand{\remarkname}{Remark}
\providecommand{\theoremname}{Theorem}

\begin{document}
\title{Monte Carlo sampling with integrator snippets}
\author{Christophe Andrieu, Mauro Camara Escudero and Chang Zhang}
\maketitle

\lyxaddress{School of Mathematics, University of Bristol}
\begin{abstract}
Assume interest is in sampling from a probability distribution $\mu$
defined on $(\mathsf{Z},\mathscr{Z})$. We develop a framework for
sampling algorithms which takes full advantage of ODE numerical integrators,
say $\psi\colon\mathsf{Z}\rightarrow\mathsf{Z}$ for one integration
step, to explore $\mu$ efficiently and robustly. The popular Hybrid
Monte Carlo (HMC) algorithm \cite{duane1987hybrid,neal2011mcmc} and
its derivatives are examples of such a use of numerical integrators.
A key idea developed here is that of sampling integrator snippets,
that is fragments of the orbit of an ODE numerical integrator $\psi$,
and the definition of an associated probability distribution $\bar{\mu}$
such that expectations with respect to $\mu$ can be estimated from
integrator snippets distributed according to $\bar{\mu}$. The integrator
snippet target distribution $\bar{\mu}$ takes the form of a mixture
of pushforward distributions which suggests numerous generalisations
beyond mappings arising from numerical integrators, e.g. normalising
flows. Very importantly this structure also suggests new principled
and robust strategies to tune the parameters of integrators, such
as the discretisation stepsize, effective integration time, or number
of integration steps, in a Leapfrog integrator. 

We focus here primarily on Sequential Monte Carlo (SMC) algorithms,
but the approach can be used in the context of Markov chain Monte
Carlo algorithms. We illustrate performance and, in particular, robustness
through numerical experiments and provide preliminary theoretical
results supporting observed performance. 

\newpage{}
\end{abstract}
\tableofcontents{}

\global\long\def\otimeswapped{\accentset{\curvearrowleft}{\otimes}}%

\newpage{}

\section{Overview and motivation: SMC sampler with HMC\protect\label{sec:Overview-and-motivation:}}

Assume interest is in sampling from a probability distribution $\mu$
on a probability space $(\mathsf{Z},\mathscr{Z})$. The main ideas
of sequential Monte Carlo (SMC) samplers to simulate from $\mu$ are
(a) to define a sequence of probability distributions $\{\mu_{n},n\in\llbracket0,P\rrbracket\}$
on $(\mathsf{Z},\mathscr{Z})$ where $\mu_{P}=\mu$, $\mu_{0}$ is
chosen by the user, simple to sample from and the sequence $\{\mu_{n},n\in\llbracket P-1\rrbracket\}$
``interpolates'' $\mu_{0}$ and $\mu_{P}$, (b) and to propagate
a cloud of samples $\{z_{n}^{(i)}\in\mathsf{Z},i\in\llbracket N\rrbracket\}$
for $n\in\llbracket0,P\rrbracket$ to represent $\{\mu_{n},n\in\llbracket0,P\rrbracket\}$
using an importance sampling/resampling mechanism \cite{del2006sequential}. 

After initialisation, for $n\in\llbracket P\rrbracket$, samples $\{z_{n-1}^{(i)},i\in\llbracket N\rrbracket\}$,
representing $\mu_{n-1}$, are propagated thanks to a user-defined
mutation Markov kernel $M_{n}\colon\mathsf{Z}\times\mathscr{Z}\rightarrow[0,1]$,
as follows. For $i\in\llbracket N\rrbracket$ sample $\tilde{z}_{n}^{(i)}\sim M_{n}(z_{n-1}^{(i)},\cdot)$
and compute the importance weights, assumed to exist for the moment,
\begin{equation}
\omega_{n}^{(i)}=\frac{{\rm d}\mu_{n}\otimeswapped L_{n-1}}{{\rm d}\mu_{n-1}\otimes M_{n}}\big(z_{n-1}^{(i)},\tilde{z}_{n}^{(i)}\big)\,,\label{eq:importance-w-SMC}
\end{equation}
where $L_{n-1}\colon\mathsf{Z}\times\mathscr{Z}\rightarrow[0,1]$
is a user-defined ``backward'' Markov kernel required to define
importance sampling on $\mathsf{Z}\times\mathsf{Z}$ and swap the
rôles of $z_{n-1}^{(i)}$ and $\tilde{z}_{n}^{(i)}$, in the sense
that for $f\colon\mathsf{Z}\rightarrow\mathbb{R}$ $\mu_{n}$-integrable,
\[
\int f(z')\frac{{\rm d}\mu_{n}\otimeswapped L_{n-1}}{{\rm d}\mu_{n-1}\otimes M_{n}}\big(z,z'\big)\mu_{n-1}({\rm d}z)M_{n}(z,{\rm d}z')=\mu_{n}(f)\,.
\]
More details are provided in Appendices~\ref{sec:app-Notation}-\ref{sec:One-measure-theoretic}
concerning the existence and definition of these weights, but can
be omitted on a first reading. 

The mutation step is followed by a selection step where for $i\in\llbracket N\rrbracket$,
$z_{n}^{(i)}=\tilde{z}_{n}^{(a_{i})}$ for $a_{i}$ the random variable
taking values in $\llbracket N\rrbracket$ with $\mathbb{P}(a_{i}=k)\propto\omega_{n}^{(k)}$.
The procedure is summarized in Alg.~\ref{alg:generic-SMC}. 

Given $\{M_{n},n\in\llbracket P\rrbracket\}$, theoretically optimal
choice of $\{L_{n-1},n\in\llbracket P\rrbracket\}$ is well understood
but tractability is typically obtained by assuming that $M_{n}$ is
$\mu_{n-1}$-invariant, or considering approximations of $\{L_{n-1},n\in\llbracket P\rrbracket\}$
and that $M_{n}$ is $\mu_{n}$-invariant. This makes Markov chain
Monte Carlo (MCMC) kernels very attractive choices for $M_{n}$.

\begin{algorithm}
\For{ $i\in\llbracket N\rrbracket$}{Sample $z_{0}^{(i)}\sim\mu_{0}(\cdot)$\;Set
$\omega_{0}^{(i)}=1$

}\For{ $n\in\llbracket P\rrbracket$}{\For{ $i\in\llbracket N\rrbracket$}{Sample
$\tilde{z}_{n}^{(i)}\sim M_{n}\big(z_{n-1}^{(i)},\cdot\big)$\;Compute
$w_{n}^{(i)}$ as in (\ref{eq:importance-w-SMC}).}\For{$i\in\llbracket N\rrbracket$}{Sample
$a_{i}\sim{\rm Cat}\big(\omega_{n}^{(1)},\ldots,\omega_{n}^{(N)}\big)$ 

Set $z_{n}^{(i)}=\tilde{z}_{n}^{(a_{i})}$}

}

\caption{Generic SMC sampler}
\label{alg:generic-SMC}
\end{algorithm}

A possible choice of MCMC kernel is the Hybrid Monte Carlo (HMC) method,
a Metropolis-Hastings (MH) update using a discretisation of Hamilton's
equations \cite{duane1987hybrid,neal2011mcmc} to update the state,
a particular instance of the use of numerical integrators of ODEs
in this context. More specifically, assume that interest is in sampling
$\pi$ defined on $(\mathsf{X},\mathscr{X})$. First the problem is
embedded into that of sampling from the joint distribution $\mu({\rm d}z):=\pi\otimes\varpi\big({\rm d}z\big)=\pi({\rm d}x)\varpi({\rm d}v)$
defined on $(\mathsf{Z},\mathscr{Z})=(\mathsf{X}\times\mathsf{V},\mathscr{X}\otimes\mathscr{V})$,
where $v$ is an auxiliary variable facilitating sampling. Following
the SMC framework we set $\mu_{n}({\rm d}z):=\pi_{n}\otimes\varpi_{n}\big({\rm d}z\big)$
for $n\in\llbracket0,P\rrbracket,$ a sequence of distributions on
$(\mathsf{Z},\mathscr{Z})$ with $\pi_{P}=\pi$, $\{\pi_{n},n\in\llbracket0,P-1\rrbracket\}$
probabilities on $\big(\mathsf{X},\mathscr{X}\big)$ and $\{\varpi_{n},n\in\llbracket0,P\rrbracket\}$
on $\big(\mathsf{V},\mathscr{V}\big)$. With $\psi\colon\mathsf{Z}\rightarrow\mathsf{Z}$
an integrator of an ODE of interest, one can use $\psi^{k}(z)$ for
some $k\in\mathbb{N}$ as a proposal in a MH update mechanism; $v$
is a source of randomness allowing ``exploration'', resampled every
now and then. Again, hereafter we let $z=:(x,v)\in\mathsf{X}\times\mathsf{V}$
be the corresponding components of $z$.
\begin{example}[Leapfrog integrator of Hamilton's equations]
\label{exa:verlet-HMC} Assume that $\mathsf{X}=\mathsf{V}=\mathbb{R}^{d}$,
that $\{\pi_{n},\varpi_{n},n\in\llbracket0,P\rrbracket\}$ have densities,
denoted $\pi_{n}(x)$ and $\varpi_{n}(v)$, with respect to the Lebesgue
measure and let $x\mapsto U_{n}(x):=-\log\pi_{n}(x)$. For $n\in\llbracket0,P\rrbracket$
and $U_{n}$ differentiable, Hamilton's equations for the potential
$H_{n}(x,v):=U_{n}(x)+\frac{1}{2}|v|^{2}$ are 
\begin{equation}
\dot{x}_{t}=v_{t},\dot{v}_{t}=-\nabla U_{n}(x_{t})\,,\label{eq:hamiltons-equations}
\end{equation}
 and possess the important property that $H_{n}(x_{t},v_{t})=H_{n}(x_{0},v_{0})$
for $t\geq0$. The corresponding leapfrog integrator is given, for
some $\epsilon>0$, by
\begin{gather}
\psi_{n}(x,v)={}_{\textsc{b}}\psi\circ{}_{\textsc{a}}\psi_{n}\circ{}_{\textsc{b}}\psi(x,v)\label{eq:leapfrog-integrator}\\
_{\textsc{b}}\psi(x,v):=\big(x,v-\tfrac{1}{2}\epsilon\nabla U_{n}(x)\big),\quad{}_{\textsc{a}}\psi_{n}(x,v)=(x+\epsilon\,v,v)\,.\nonumber 
\end{gather}
We point out that, with the exception of the first step, only one
evaluation of $\nabla U(x)$ is required per integration step since
the rightmost $_{\textsc{b}}\psi$ in (\ref{eq:leapfrog-integrator})
recycles the last computation from the last iteration.  Let $\sigma\colon\mathsf{Z}\rightarrow\mathsf{Z}$
be such that for any $f\colon\mathsf{Z}\rightarrow\mathsf{Z}$, $f\circ\sigma(x,v)=f(x,-v)$,
then in its most basic form the HMC update leaving $\mu_{n}$ invariant
proceeds as follows, for $(z,A)\in\mathsf{Z}\times\mathscr{Z}$
\begin{equation}
M_{n+1}(z,A)=\int\varpi_{n}({\rm d}v')\big[\alpha_{n}(x,v';T)\mathbf{1}\{\psi_{n}^{T}(x,v')\in A\}+\bar{\alpha}_{n}(x,v';T)\mathbf{1}\{\sigma(x,v)\in A\}\big]\,,\label{eq:HMC-kernel-deterministic}
\end{equation}
with, for some user defined $T\in\mathbb{N},$
\begin{equation}
\alpha_{n}(z;T):=1\wedge\frac{\mu_{n}\circ\psi_{n}^{T}(z)}{\mu_{n}(z)}\,,\label{eq:def-alpha-T}
\end{equation}
and $\bar{\alpha}_{n}(z;T)=1-\alpha_{n}(z;T)$.
\end{example}

Other ODEs, capturing other properties of the target density, or other
types of integrators are possible. However a common feature of integrator
based updates is the need to compute recursively an integrator snippet
$\mathsf{z}:=\big(z,\psi(z),\psi^{2}(z),\ldots,\psi^{T}(z)\big)$,
for a given mapping $\psi\colon\mathsf{Z}\rightarrow\mathsf{Z}$,
of which only the endpoint $\psi^{T}(z)$ is used. This raises the
question of recycling intermediate states, all the more so that computation
of the snippet often involves quantities shared with the evaluation
of $U_{n}$. In Example~\ref{exa:verlet-HMC}, for instance, expressions
for $\nabla U_{n}(x)$ and $U_{n}(x)$ often involve the same computationally
costly quantities and evaluation of the density $\mu_{n}(x)$ where
$\nabla U_{n}(x)$ has already been evaluated is therefore often virtually
free; consider for example $U(x)=x^{\top}\Sigma^{-1}x$ for a covariance
matrix $\Sigma$, then $\nabla U(x)=2\Sigma^{-1}x$.

In turn these quantities offer the promise of being able to exploit
points used to generate the snippet while preserving accuracy of the
estimators of interest, through importance sampling or reweighting.
For example an algorithm exploiting the snippet fully could use the
following mixture of $\mu_{n}$-invariant Markov chain transition
kernels \cite{mackenze1989improved,neal2011mcmc,hoffman2022tuning},
\[
\mathfrak{M}_{n+1}(z,A)=\frac{1}{T+1}\sum_{k=0}^{T}\int\varpi_{n}({\rm d}v')P_{n,k}(x,v';A)\,,
\]
where for $(z,A)\in\mathsf{Z}\times\mathscr{Z}$ and $k\in\llbracket P-1\rrbracket$
\begin{equation}
P_{n,k}(z,A):=\alpha_{n}(z;k)\mathbf{1}\{\psi_{n}^{k}(z)\in A\}+\bar{\alpha}_{n}(z;k)\mathbf{1}\{\sigma(z)\in A\}\,.\label{eq:def-Pnk-intro}
\end{equation}
As we shall see our work shares the same objective but we adopt a
strategy more closely related to \cite{neal1994improved} (see Appendix~\ref{sec:MCMC-with-integrator}
for a more detailed discussion) which however leads to a comparison
between states along a snippet. Further, an advantage of SMC samplers
is that they provide population wide information about tuning parameters,
in particular of numerical integrators, therefore suggesting robust
self-tuning procedures. 

The manuscript is organised as follows. In Section~\ref{sec:A-simple-example:}
we first provide a high level description of particular instances
of the class of algorithms considered and provide a justification
through reinterpretation as standard SMC algorithms in Subsection~\ref{subsec:Outline-justification}.
In Subsection~\ref{subsec:Straightforward-generalizations} we discuss
direct extensions of our algorithms, some of which we explore in the
present manuscript. This work has links to related recent attempts
in the literature \cite{rotskoff2019dynamical,dau2020waste,thin2021neo};
these links are discussed and contrasted with our work in Subsection~\ref{subsec:Computational-considerations}
where some of the motivations behind these algorithms are also discussed.
Early exploratory simulations demonstrating the interest of our approach
are provided in Appendices~\ref{app:early-experimental-HMC}-\ref{sec:Early-experimental-results-dilamentary}.
Our main focus is however on the development of adaptive algorithms
in Section~\ref{sec:Adaptation-with-Integrator} where, focussing
on the Leapfrog integrator, we develop novel criteria and updating
strategies to adapt the discretisation stepsize and integration time
by taking advantage of the population-wide information available;
in simulation the resulting algorithms display remarkable robustness
on the examples we have considered We draw some conclusions in Section~\ref{sec:Discussion}.

The manuscript contains a number of appendices, containing supporting
material and additional developments. Notation, definitions and basic
mathematical background can be found in Appendices~\ref{sec:app-Notation}-\ref{sec:One-measure-theoretic}.
Proofs for Section~\ref{sec:A-simple-example:} are provided in Appendix~\ref{app:Proofs-for-Section}.
In Appendix~\ref{sec:Sampling-a-mixture:} we introduce the more
general framework of Markov snippets Monte Carlo and associated formal
justifications; the algorithms used in Appendix~\ref{subsec:Numerical-illustration:-orthant}
rely on some of these extensions. Appendix~\ref{sec:Sampling-HMC-trajectories}
details the link with the scenario considered in Section~\ref{sec:A-simple-example:}.
In Appendix~\ref{subsec:Sampling-randomized-integrator} we provide
general results facilitating the practical calculation of some of
the Radon-Nikodym involved, highlighting why some of the usual constraints
on mutation and backward kernels in SMC can be lifted here. Appendix~\ref{app-sec:details-imp-adapt-epsilon}
contains ancillary supporting details for Section~\ref{sec:Adaptation-with-Integrator}.
In Appendix~\ref{sec:Preliminary-theoretical-characte} we provide
elements of a theoretical analysis explaining expected properties
of the algorithms proposed in this manuscript, although a fully rigorous
theoretical analysis is beyond the present methodological contribution.
In Appendix~\ref{sec:MCMC-with-integrator} we explore the use of
some of the ideas developed here in the context of MCMC algorithms
and establish links with earlier suggestions, such as ``windows of
states'' techniques proposed in the context of HMC \cite{neal1994improved,neal2011mcmc}. 

A Python implementation of the algorithms developed in this paper
is available at \url{https://github.com/MauroCE/IntegratorSnippets}.

\section{An introductory example \protect\label{sec:A-simple-example:}}

We still aim to sample from a probability distribution $\mu$ on $(\mathsf{Z},\mathscr{Z})$
as described above Example~\ref{exa:verlet-HMC} using an SMC sampler
relying on the leapfrog integrator of Hamilton's equations. As in
the previous section we introduce an interpolating sequence of distributions
$\{\mu_{n},n\in\llbracket0,P\rrbracket\}$ on $(\mathsf{Z},\mathscr{Z})\}$
and assume for now the existence of densities for $\{\mu_{n},n\in\llbracket0,P\rrbracket\}$
with respect to a common measure $\upsilon$, say the Lebesgue measure
on $\mathbb{R}^{2d}$, denoted $\mu_{n}(z):={\rm d}\mu_{n}/{\rm d}\upsilon(z)$
for $z\in\mathsf{Z}$ and $n\in\llbracket0,P\rrbracket$ and denote
$\psi_{n}$ the corresponding integrator, which again is measure preserving
in this setup.

\subsection{An SMC-like algorithm}

Primary interest in this paper is in algorithms of the type given
in Alg.~\ref{alg:Unfolded-PDMP-SMC-1} and variations thereof; throughout
$T\in\mathbb{N}\setminus\{0\}$. 

\begin{algorithm}
Sample $z_{0}^{(i)}\overset{{\rm iid}}{\sim}\mu_{0}$ for $i\in\llbracket N\rrbracket$.

\For{$n\in\llbracket P\rrbracket$}{

\For{$i\in\llbracket N\rrbracket$}{

\For{$k\in\llbracket0,T\rrbracket$}{

Compute $z_{n-1,k}^{(i)}:=\psi_{n}^{k}(z_{n-1}^{(i)})$ and

\[
\bar{w}_{n,k}\big(z_{n-1}^{(i)}\big):=\frac{\mu_{n}\big(z_{n-1,k}^{(i)}\big)}{\mu_{n-1}\big(z_{n-1}^{(i)}\big)}=\frac{\mu_{n}\circ\psi_{n}^{k}\big(z_{n-1}^{(i)}\big)}{\mu_{n-1}\big(z_{n-1}^{(i)}\big)}\,,
\]

}

}

\For{$j\in\llbracket N\rrbracket$}{

Sample $\llbracket N\rrbracket\times\llbracket0,T\rrbracket\ni(b_{j},a_{j})\sim{\rm Cat}\big(\{\bar{w}_{n,k}(z_{n-1}^{(i)}),(i,k)\in\llbracket N\rrbracket\times\llbracket0,T\rrbracket\}\big)$

Set $\bar{z}_{n}^{(j)}:=(\bar{x}_{n-1}^{(j)},\bar{v}_{n-1}^{(j)})=z_{n-1,a_{j}}^{(b_{j})}$

Rejuvenate the velocities $z_{n}^{(j)}=(\bar{x}_{n-1}^{(j)},v_{n}^{(j)})$
with $v_{n}^{(j)}\sim\varpi_{n}$.

}

}

\caption{Unfolded Hamiltonian Snippet SMC algorithm \protect\label{alg:Unfolded-PDMP-SMC-1}}
\end{algorithm}

The SMC sampler-like algorithm in Alg.~\ref{alg:Unfolded-PDMP-SMC-1}
therefore involves propagating $N$ ``seed'' particles $\{z_{n-1}^{(i)},i\in\llbracket N\rrbracket\}$,
with a mutation mechanism consisting of the generation of $N$ integrator
snippets $\mathsf{z}:=\big(z,\psi_{n}(z),\psi_{n}^{2}(z),\ldots,\psi_{n}^{T}(z)\big)$
started at every seed particle $z\in\{z_{n-1}^{(i)},i\in\llbracket N\rrbracket\}$,
resulting in $N\times(T+1)$ particles which are then whittled down
to a set of $N$ seed particles using a standard resampling scheme;
after rejuvenation of velocities this yields the next generation of
seed particles--this is illustrated in Fig.~\ref{fig:snippet-illustration};
a more realistic depiction of Alg.~\ref{alg:Unfolded-PDMP-SMC-1}
in the context of Example~\ref{exa:verlet-HMC} with $\pi$ a two
dimensional distribution is presented in Fig.~\ref{fig:illustration-HMC-snippet}.
This algorithm should be contrasted with standard implementations
of SMC samplers where, after resampling, a seed particle normally
gives rise to a single particle in the mutation step, in Fig.~~\ref{fig:snippet-illustration}
the last or first state on the snippet when using an MH as mutation
kernel. Intuitively validity of the algorithm follows from the fact
that if $\big\{(z_{n-1}^{(i)},1),i\in\llbracket N\rrbracket\big\}$
represent $\mu_{n-1}$, then $\big\{\big(z_{n-1,k}^{(i)},\bar{w}_{n,k}(z_{n-1}^{(i)})\big)),(i,k)\in\llbracket N\rrbracket\times\llbracket0,T\rrbracket\big\}$
represents $\mu_{n}$ in the sense that for $f\colon\mathsf{Z}\rightarrow\mathbb{R}$
summable, one can use the approximation
\begin{equation}
\mu_{n}(f)\approx\sum_{i=1}^{N}\sum_{k=0}^{T}f\circ\psi_{n}^{k}\big(z_{n-1}^{(i)}\big)\frac{\mu_{n}\circ\psi_{n}^{k}\big(z_{n-1}^{(i)}\big)/\mu_{n-1}\big(z_{n-1}^{(i)}\big)}{\sum_{j=1}^{N}\sum_{l=0}^{T}\mu_{n}\circ\psi_{n}^{l}\big(z_{n-1}^{(j)}\big)/\mu_{n-1}\big(z_{n-1}^{(j)}\big)}\,,\label{eq:snippet-estimator-I}
\end{equation}
where the self-normalization of the weights is only required in situations
where the ratio $\mu_{n}(z)/\mu_{n-1}(z)$ is only known up to a constant. 

Simple numerical experiments demonstrating robustness and efficiency
of our approach are provided in Appendices~\ref{app:early-experimental-HMC}-\ref{sec:Early-experimental-results-dilamentary}.
However our primary objective, beyond establishing correctness, is
to show that our proposed framework naturally lends itself to novel
adaptation strategies, for example of the stepsize $\epsilon$ and
integration time $T$; see Section~\ref{sec:Adaptation-with-Integrator}.
To that purpose we first provide justification for the correctness
of Alg.~\ref{alg:Unfolded-PDMP-SMC-1} and the estimator (\ref{eq:snippet-estimator-I})
by recasting the procedure as a standard SMC sampler targetting a
particular sequence of distributions in Section~\ref{subsec:Outline-justification}
and using properties of mixtures. Direct generalizations are provided
in Section~\ref{subsec:Straightforward-generalizations} and the
adaptation strategies of Section~\ref{sec:Adaptation-with-Integrator}
then follow.

\begin{figure}[H]
\begin{centering}
\includegraphics[width=0.75\textwidth]{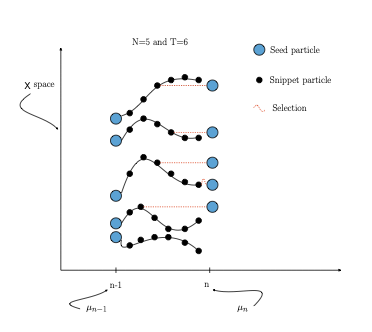}
\par\end{centering}
\caption{Illustration of the transition from $\mu_{n-1}$ to $\mu_{n}$ with
integrator snippet SMC. A snippet grows (black dots) out of each seed
particle (blue). The middle snippet gives rises through selection
(dashed red) to two seed particles while the bottom snippet does not
produce any seed particle.}
\label{fig:snippet-illustration}

\end{figure}

\begin{figure}[H]
\includegraphics[width=1\textwidth]{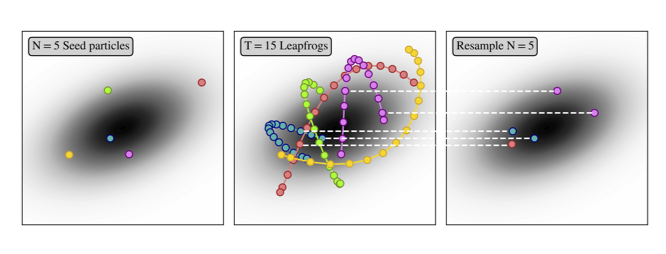}\caption{Evolution of the $x$ components of $5$ seed particles and their
integrator snippets ($T=15$) targetting a normal distribution $\pi$
defined on $\mathbb{R}^{2}$, using the leapfrog integrator of Example~\ref{exa:verlet-HMC}.}
\label{fig:illustration-HMC-snippet}

\end{figure}

\subsection{Justification \protect\label{subsec:Outline-justification}}

We now outline the main ideas underpinning the theoretical justification
of Alg.~\ref{alg:Unfolded-PDMP-SMC-1}. Key to this is establishing
that Alg.~\ref{alg:Unfolded-PDMP-SMC-1} is a standard SMC sampler
targetting a particular sequence of probability distributions $\big\{\bar{\mu}_{n},n\in\llbracket0,P\rrbracket\big\}$
from which samples can be processed to approximate expectations with
respect to $\big\{\mu_{n},n\in\llbracket0,P\rrbracket\big\}$. This
has the advantage that no fundamentally new theory is required and
that standard methodological ideas can be re-used in the present scenario,
while the particular structure of $\big\{\bar{\mu}_{n},n\in\llbracket0,P\rrbracket\big\}$
can be exploited for new developments. This section focuses on identifying
$\big\{\bar{\mu}_{n},n\in\llbracket0,P\rrbracket\big\}$ and establishing
some of their important properties. Similar ideas are briefly touched
upon in \cite{dau2020waste}, but we will provide full details and
show how these ideas can be pushed further, in interesting directions.

First for $(n,k)\in\llbracket0,P\rrbracket\times\llbracket0,T\rrbracket$
let $\psi_{n,k}\colon\mathsf{Z}\rightarrow\mathsf{Z}$ be measurable
and invertible mappings, define $\mu_{n,k}({\rm d}z):=\mu_{n}^{\psi_{n,k}^{-1}}({\rm d}z)$,
i.e. the distribution of $\psi_{n,k}^{-1}(z)$ when $z\sim\mu_{n}$.
It is worth pointing out that invertibility of these mappings is not
required theoretically, but important practially and facilitates interpretation
throughout. Earlier we have focused on the scenario where $\psi_{n,k}=\psi_{n}^{k}$
for an integrator $\psi_{n}\colon\mathsf{Z}\rightarrow\mathsf{Z}$,
but this turns out not to be a requirement, although it is our main
motivation. Useful applications of this general perspective can be
found in Subsection~\ref{subsec:Straightforward-generalizations}.
Introduce the probability distributions on $\big(\mathbb{\llbracket}0,T\rrbracket\times\mathsf{Z},\mathscr{P}(\mathbb{\llbracket}0,T\rrbracket)\otimes\mathscr{Z}\big)$
\[
\bar{\mu}_{n}(k,{\rm d}z)=\frac{1}{T+1}\mu_{n,k}({\rm d}z)\,,
\]
for $n\in\llbracket0,P\rrbracket$. We will show that Alg.~\ref{alg:Unfolded-PDMP-SMC-1}
can be interpreted as an SMC sampler targeting the sequence of marginal
distributions on $\big(\mathsf{Z},\mathscr{Z}\big)$
\begin{equation}
\bar{\mu}_{n}({\rm d}z)=\frac{1}{T+1}\sum_{k=0}^{T}\mu_{n,k}({\rm d}z)\,,\label{eq:bar-mu-n-marginal}
\end{equation}
which we may refer to as a mixture. Note that the conditional distribution
is
\begin{align}
\bar{\mu}_{n}(k & \mid z)=\frac{1}{T+1}\frac{{\rm d}\mu_{n,k}}{{\rm d}\bar{\mu}_{n}}(z)\,,\label{eq:mu-bar-k-given-z}
\end{align}
which can be computed in most scenarios of interest. 
\begin{example}
\label{exa:densities-mu-n-k} For $n\in\llbracket0,P\rrbracket$ assume
the existence of a density $\mu_{n}(z)$ w.r.t. the Lebesgue measure.
Assuming that $\psi_{n,k}$ is ``volume preserving'', then for $k\in\llbracket0,T\rrbracket$,
\[
\mu_{n,k}(z):=\frac{{\rm d}\mu_{n,k}}{{\rm d}\upsilon}(z)=\mu_{n}\circ\psi_{n,k}(z)\text{ and }\bar{\mu}_{n}(z):=\frac{{\rm d}\bar{\mu}_{n}}{{\rm d}\upsilon}(z)=\frac{1}{T+1}\sum_{k=0}^{T}\mu_{n}\circ\psi_{n,k}(z)\,,
\]
and therefore 
\[
w_{n,k}(z):=\frac{{\rm d}\mu_{n,k}}{{\rm d}\big(\sum_{l=0}^{T}\mu_{n,l}\big)}(z)=\frac{\mu_{n}\circ\psi_{n,k}(z)}{\sum_{l=0}^{T}\mu_{n}\circ\psi_{n,l}(z)}\,.
\]
When $\{\psi_{n,k},k\in\llbracket P\rrbracket\}$ are not volume preserving,
additional multiplicative Jacobian determinant-like terms may be required
(see Lemma~\ref{thm:measure-theoretic-transform} and the additional
requirement that $\{\psi_{n,k},k\in\llbracket P\rrbracket\}$ be differentiable). 
\end{example}

A central point throughout this paper is how samples from $\mu_{n}$
can be used to obtain samples from the marginal $\bar{\mu}_{n}$ and
vice-versa, thanks to the mixture structure relating the two distributions.
Naturally, given $z\sim\mu_{n}$, sampling $k\sim\mathcal{U}\big(\llbracket0,T\rrbracket\big)$
and returning $\psi_{n,k}^{-1}(z)$ yields a sample from the marginal
$\bar{\mu}_{n}$. More importantly, we have
\begin{prop}[Knitting-Tinking'']
\label{prop:knitting-tinking} For $n\in\llbracket0,P\rrbracket$,
\begin{enumerate}
\item Assume $z\sim\bar{\mu}_{n}$ and sample $k\sim\bar{\mu}_{n}(\cdot\mid z)$.
Then $z\mid k\sim\mu_{n}^{\psi_{n,k}^{-1}}({\rm d}z)$ and $\psi_{n,k}(z)\sim\mu_{n}$. 
\item For $f\colon\mathsf{Z}\rightarrow\mathbb{R}$ $\mu_{n}$-integrable
and $k\in\llbracket0,T\rrbracket$
\begin{align}
\int f(z)\mu_{n}({\rm d}z) & =\int\Big\{\sum_{k=0}^{T}f\circ\psi_{n,k}(z)\bar{\mu}_{n}(k\mid z)\Big\}\bar{\mu}_{n}({\rm d}z)\,,\label{eq:unfolding-HMC-setup}
\end{align}
\end{enumerate}
\end{prop}

\begin{proof}
For $f\colon\mathsf{Z}\rightarrow\mathbb{R}$ $\mu_{n}$-integrable
and $k\in\llbracket0,T\rrbracket$, a change of variable (see Theorem~\ref{thm:change-of-variables})
yields
\begin{align*}
\int f\circ\psi_{n,k}(z)\mu_{n,k}({\rm d}z)=\int f\circ\psi_{n,k}(z)\mu_{n}^{\psi_{n,k}^{-1}}({\rm d}z) & =\int f(z)\mu_{n}({\rm d}z)\,,
\end{align*}
Therefore

\begin{align*}
\int f(z)\mu_{n}({\rm d}z) & =\frac{1}{T+1}\sum_{k=0}^{T}\int f\circ\psi_{n,k}(z)\mu_{n,k}({\rm d}z)\\
 & =\sum_{k=0}^{T}\int f\circ\psi_{n,k}(z)\bar{\mu}_{n}(k,{\rm d}z)\\
 & =\int\Big\{\sum_{k=0}^{T}f\circ\psi_{n,k}(z)\bar{\mu}_{n}(k\mid z)\Big\}\bar{\mu}_{n}({\rm d}z)\,,
\end{align*}
Using $f=\mathbf{1}_{A}$ for $A\in\mathscr{Z}$ formally establishes
the earlier claim that if $(k,z)\sim\bar{\mu}_{n}$ then $\psi_{n,k}(z)\sim\mu_{n}$. 
\end{proof}
Note that, as suggested by Example~\ref{exa:densities-mu-n-k}, construction
of the estimator corresponding to (\ref{eq:unfolding-HMC-setup})
will only require evaluations of the density $\mu_{n}$ and function
$f$ at $z,\psi_{n,1}(z),\psi_{n,2}(z),\ldots,\psi_{n,T}(z)$. It
should be clear that the results of Proposition~\ref{prop:knitting-tinking}
extend well beyond the leapfrog integrator setup, a fact we discuss
and exploit in the remainder of the manuscript.

We now turn to the description of an SMC algorithm targeting $\big\{\bar{\mu}_{n},n\in\llbracket0,P\llbracket\big\}$,
Alg.~\ref{alg:Folded-PDMP-SMC-1}, and then establish that it is
probabilistically equivalent to Alg.~\ref{alg:Unfolded-PDMP-SMC-1},
in a sense made precise below. With $z=(x,v)$ for $n\in\llbracket P\rrbracket$
we introduce the following mutation kernel
\begin{equation}
\bar{M}_{n}(z,{\rm d}z'):=\sum_{k=0}^{T}\bar{\mu}_{n-1}(k\mid z)\,R_{n}(\psi_{n-1,k}(z),{\rm d}z'),\quad R_{n}(z,{\rm d}z'):=(\delta_{x}\otimes\varpi_{n-1})({\rm d}z')\,,\label{eq:def-bar-M}
\end{equation}
where we note that the refreshment kernel has the property that $\mu_{n-1}R_{n}=\mu_{n-1}$.
With the assumption $\mu_{n-1}\gg\bar{\mu}_{n}$, from Lemma~\ref{lem:bar-weight-derivation}
one can identify the corresponding near optimal kernel $\bar{L}_{n-1}\colon\mathsf{Z}\times\mathscr{Z}\rightarrow[0,1]$,
which leads to the SMC sampler importance weights, at step $n\in\llbracket P\rrbracket$,
\begin{equation}
\bar{w}_{n}(z,z'):=\frac{{\rm d}\bar{\mu}_{n}\otimeswapped\bar{L}_{n-1}}{{\rm d}\bar{\mu}_{n-1}\otimes\bar{M}_{n}}(z,z')=\frac{{\rm d}\bar{\mu}_{n}}{{\rm d}\mu_{n-1}}(z')\,.\label{eq:weight-folded-HMC}
\end{equation}
that is for $f\colon\mathsf{Z}\rightarrow\mathbb{R}$ such that $\bar{\mu}_{n}(|f|)<\infty$,
\[
\int f(z')\frac{{\rm d}\bar{\mu}_{n}}{{\rm d}\mu_{n-1}}(z')\bar{\mu}_{n-1}\otimes\bar{M}_{n}\big({\rm d}(z,z')\big)=\bar{\mu}_{n}(f)\,.
\]
The corresponding standard SMC sampler is given in Alg.~\ref{alg:Folded-PDMP-SMC-1}
where,
\begin{itemize}
\item the weighted particles $\{(\check{z}_{n}^{(i)},1),i\in\llbracket N\rrbracket\}$
represent $\bar{\mu}_{n}$ and $\{(\check{z}_{n,k}^{(i)},w_{n,k}(\check{z}_{n}^{(i)})),(i,k)\in\llbracket N\rrbracket\times\llbracket0,T\rrbracket\}$
represent $\mu_{n}$ from (\ref{eq:unfolding-HMC-setup}),
\item steps \ref{alg:item:folded-barM-kernel-begin}-\ref{alg:item:folded-barM-kernel-end}
correspond to sampling from the mutation kernel $\bar{M}_{n+1}$ in
(\ref{eq:def-bar-M}),
\item $\{(z_{n}^{(i)},1),i\in\llbracket N\rrbracket\}$ represent $\mu_{n}$, 
\item $\{\big(z_{n}^{(i)},\bar{w}_{n+1}(z_{n}^{(i)})),i\in\llbracket N\rrbracket\}$
represent $\bar{\mu}_{n+1}$, and so do $\{(\check{z}_{n+1}^{(i)},1),i\in\llbracket N\rrbracket\}$.
\end{itemize}
\begin{algorithm}[h]
sample $\check{z}_{0}^{(i)}\overset{{\rm iid}}{\sim}\bar{\mu}_{0}=\mu_{0}$
for $i\in\llbracket N\rrbracket$.

\For{$n=0,\ldots,P-1$}{

\For{$i\in\llbracket N\rrbracket$}{

\label{alg:item:folded-barM-kernel-begin}\For{$k\in\llbracket0,T\rrbracket$}{

\label{alg:item:folded:compute-state-weights}compute $\check{z}_{n,k}^{(i)}=(\check{x}_{n,k}^{(i)},\check{v}_{n,k}^{(i)}):=\psi_{n,k}(\check{z}_{n}^{(i)}),w_{n,k}\big(\check{z}_{n}^{(i)}\big)$

}

sample $a_{i}\sim{\rm Cat}\left(w_{n,0}\big(\check{z}_{n}^{(i)}\big),w_{n,1}\big(\check{z}_{n}^{(i)}\big),w_{n,2}\big(\check{z}_{n}^{(i)}\big),\ldots,w_{n,T}\big(\check{z}_{n}^{(i)}\big)\right)$

set $z_{n}^{(i)}=(\check{x}_{n,a_{i}}^{(i)},v_{n}^{(i)})$ with $v_{n}^{(i)}\sim\varpi_{n}$

\label{alg:item:folded-barM-kernel-end}

\[
\bar{w}_{n+1}\big(z_{n}^{(i)}\big)=\frac{{\rm d}\bar{\mu}_{n+1}}{{\rm d}\mu_{n}}\big(z_{n}^{(i)}\big)\,,
\]
\label{alg:item:folded:compute-bar-w}}

\For{$j\in\llbracket N\rrbracket$}{

sample $b_{j}\sim{\rm Cat}\left(\bar{w}_{n+1}\big(z_{n}^{(1)}\big),\ldots,\bar{w}_{n+1}\big(z_{n}^{(N)}\big)\right)$

set $\check{z}_{n+1}^{(j)}=z_{n}^{(b_{j})}$

}

}

\caption{Folded Hamiltonian Snippet SMC algorithm \protect\label{alg:Folded-PDMP-SMC-1}}
\end{algorithm}

Notice that we assume here $\psi_{0}={\rm Id}$, hence that $\bar{\mu}_{0}=\mu_{0}$,
and that the weights $w_{n,k}$ appear as being computed twice in
step \ref{alg:item:folded:compute-state-weights} and step \ref{alg:item:folded:compute-bar-w}
when evaluating the resampling weights at the previous iteration,
for the only reason that it facilitates exposition. The identities
(\ref{eq:unfolding-HMC-setup}) and (\ref{eq:mu-bar-k-given-z}) suggest,
for any $n\in\llbracket P\rrbracket$, the estimator of $\mu_{n}(f)$
for $f\colon\mathsf{Z}\rightarrow\mathbb{R}$ $\mu_{n}$-integrable,
\begin{align}
\check{\mu}_{n}(f) & =\frac{1}{N}\sum_{i=1}^{N}\sum_{k=0}^{T}\frac{1}{T+1}\frac{{\rm d}\mu_{n,k}}{{\rm d}\bar{\mu}_{n}}(\check{z}_{n}^{(i)})f\circ\psi_{n,k}\big(\check{z}_{n}^{(i)}\big).\nonumber \\
 & =\frac{1}{N}\sum_{i=1}^{N}\sum_{k=0}^{T}\frac{\mu_{n}\circ\psi_{n,k}}{\sum_{l=0}^{T}\mu_{n}\circ\psi_{n,l}}(\check{z}_{n}^{(i)})f\circ\psi_{n,k}\big(\check{z}_{n}^{(i)}\big)\,,\label{eq:mu-check}
\end{align}
where the second line is correct under the assumptions of Example~\ref{exa:densities-mu-n-k}.
Further, when $\mu_{n-1}\gg\mu_{n,k}$ for $k\in\llbracket0,T\rrbracket$
one can also write (\ref{eq:unfolding-HMC-setup}) for $f\colon\mathsf{Z}\rightarrow\mathbb{R}$
summable as
\begin{align*}
\mu_{n}(f) & =\int\Big\{\frac{1}{T+1}\sum_{k=0}^{T}f\circ\psi_{n,k}(z)\frac{{\rm d}\mu_{n,k}}{{\rm d}\mu_{n-1}}(z)\Big\}\mu_{n-1}({\rm d}z)\,,
\end{align*}
which can be estimated, using self-renormalization when required,
with 
\begin{equation}
\hat{\mu}_{n}(f)=\sum_{i=1}^{N}\sum_{k=0}^{T}\frac{\frac{{\rm d}\mu_{n,k}}{{\rm d}\mu_{n-1}}(z_{n-1}^{(i)})}{\sum_{j=1}^{N}\sum_{l=0}^{T}\frac{{\rm d}\mu_{n,l}}{{\rm d}\mu_{n-1}}(z_{n-1}^{(j)})}f\circ\psi_{n,k}\big(z_{n-1}^{(i)}\big)\,,\label{eq:mu-hat}
\end{equation}
therefore justifying the estimator suggested in (\ref{eq:snippet-estimator-I})
in the particular case where the conditions of Example~\ref{exa:densities-mu-n-k}
are satisfied, once we establish the equivalence of Alg.~\ref{alg:Folded-PDMP-SMC-1}
and Alg.~\ref{alg:Unfolded-PDMP-SMC-1} in Proposition~\ref{prop:unfolded-folded-equivalent}.
In fact it can be shown (Proposition~\ref{prop:unfolded-folded-equivalent})
that, with $\mathbb{E}_{i}$ referring to the expectation of the probability
distribution underpinning Alg.~$i$,
\[
\mathbb{E}_{3}\bigl(\check{\mu}_{n}(f)\mid z_{n-1}^{(j)},j\in\llbracket N\rrbracket\bigr)=\hat{\mu}_{n}(f)\,,
\]
a form of Rao-Blackwellization implying lower variance for $\hat{\mu}_{n}(f)$
while the two estimators share the same bias. Interestingly a result
we establish later in the paper, Proposition~\ref{prop:variance-reduc-barmu},
suggests that the variance $\check{\mu}_{n}(f)$ is smaller than that
of the standard Monte Carlo estimator that assumes $N$ samples $z_{n}^{(i)}\overset{{\rm iid}}{\sim}\mu_{n}$,
due to the control variate nature of integrator snippets estimators.
Another point is that computation of the weight $\mu_{n}\circ\psi_{n,k}/\bar{\mu}_{n}$
only requires knowledge of $\mu_{n}$ up to a normalizing constant,
that is the estimator is unbiased if $\check{z}_{n}^{(i)}\sim\bar{\mu}_{n}$
for $i\in\llbracket N\rrbracket$ even if $\mu_{n}$ is not completely
known, while the estimator (\ref{eq:snippet-estimator-I}) will most
often require self-normalisation, hence inducing a bias. 

We now provide the probabilistic argument justifying Alg.~\ref{alg:Unfolded-PDMP-SMC-1}
and the shared notation $\{z_{n}^{(i)},i\in\llbracket N\rrbracket\}$
in Alg.~\ref{alg:Unfolded-PDMP-SMC-1} and Alg.~\ref{alg:Folded-PDMP-SMC-1}.
The result is illustrated in Fig.~\ref{fig:unfolded-folded-equivalent}.
\begin{prop}
\label{prop:unfolded-folded-equivalent} Alg.~\ref{alg:Unfolded-PDMP-SMC-1}
and Alg.~\ref{alg:Folded-PDMP-SMC-1} are probabilistically equivalent.
More precisely, letting $\mathbb{P}_{i}$ refer to the probability
of Algorithm $i$ for $i\in\{2,3\}$,
\begin{enumerate}
\item for $n\in\llbracket P-1\rrbracket$ the distributions of $\{z_{n}^{(i)},i\in\llbracket N\rrbracket\}$
conditional upon $\big\{ z_{n-1}^{(i)},i\in\llbracket N\rrbracket\big\}$
in Alg.~\ref{alg:Folded-PDMP-SMC-1} and Alg.~\ref{alg:Unfolded-PDMP-SMC-1}
are the same, 
\item the joint distributions of $\{z_{n}^{(i)},i\in\llbracket N\rrbracket,n\in\llbracket0,P-1]\}$
are the same under $\mathbb{P}_{2}$ and $\mathbb{P}_{3}$,
\item for $n\in\llbracket P\rrbracket$, any $f\colon\mathsf{Z}\rightarrow\mathbb{R}$
$\mu_{n}$-integrable with $\check{\mu}_{n}(f)$ and $\hat{\mu}_{n}(f)$
as in (\ref{eq:mu-check}) and (\ref{eq:mu-hat}),
\[
\mathbb{E}_{3}\bigl(\check{\mu}_{n}(f)\mid z_{n-1}^{(j)},j\in\llbracket N\rrbracket\bigr)=\hat{\mu}_{n}(f)\,.
\]
\end{enumerate}
\end{prop}

Since the justification of the latter interpretation of the algorithm
is straightforward, as a standard SMC sampler targeting instrumental
distributions $\big\{\bar{\mu}_{n},n\in\llbracket0,P\llbracket\big\}$,
and allows for further easy generalisations we will adopt this perspective
in the remainder of the manuscript for simplicity. 

\begin{figure}
\centering
\includegraphics[width=0.85\textwidth]{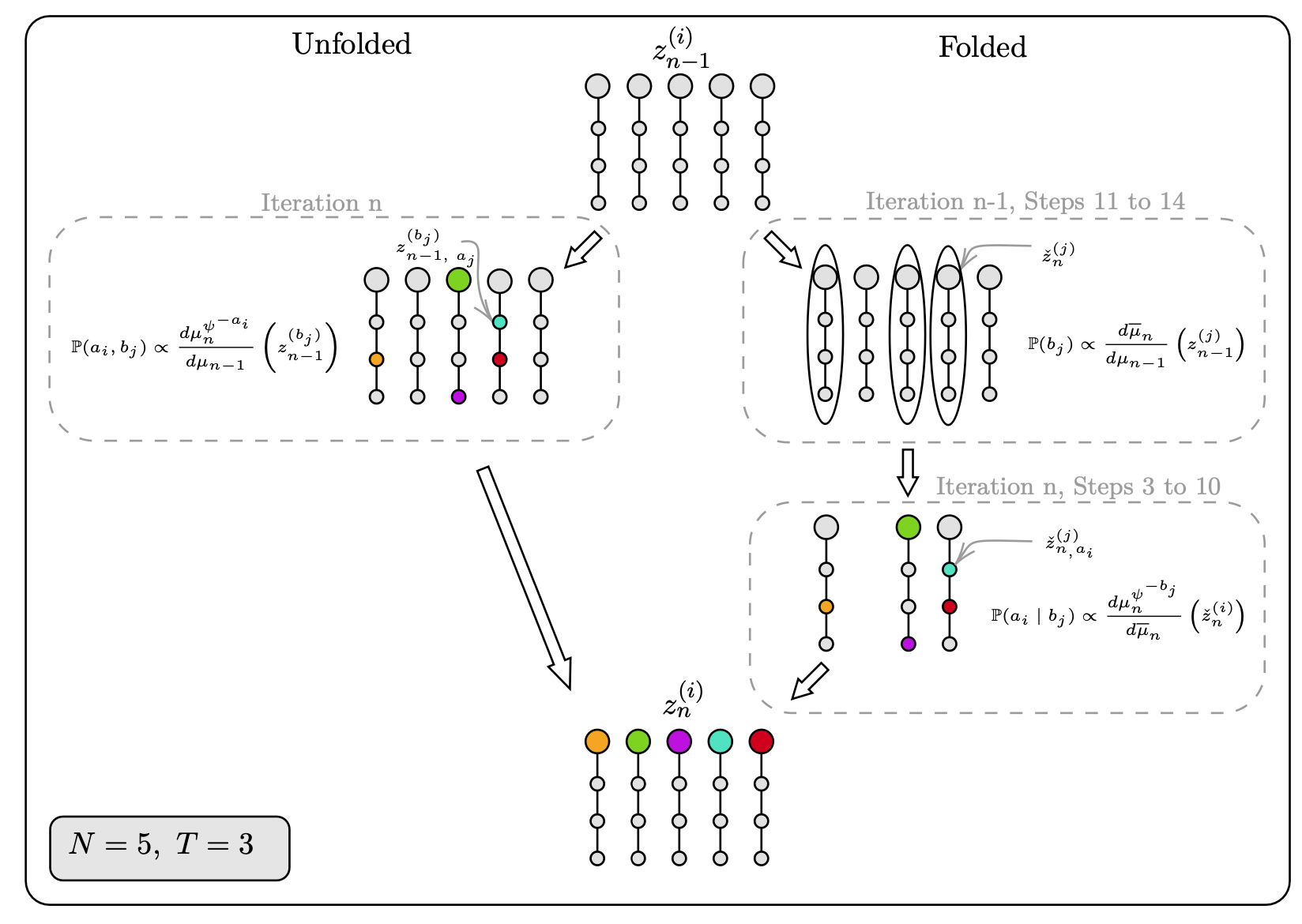}

\caption{Alg.~\ref{alg:Unfolded-PDMP-SMC-1} and Alg.~\ref{alg:Folded-PDMP-SMC-1}
are equivalent.}
\label{fig:unfolded-folded-equivalent}

\end{figure}

This reinterpretation also allows the use of known facts about SMC
sampler algorithms. For example it is well known that the output of
SMC samplers can be used to estimate unbiasedly unknown normalising
constants by virtue of the fact that, in the present scenario,
\[
\prod_{n=0}^{P-1}\Big[\frac{1}{N}\sum_{i=1}^{N}\frac{{\rm d}\bar{\mu}_{n+1}}{{\rm d}\mu_{n}}\big(z_{n}^{(i)}\big)\Big]
\]
has expectation $1$ under $\mathbb{P}_{3}$. Now assume that the
densities of $\{\mu_{n},n\in\llbracket0,P\rrbracket\}$ are known
up to a constant only, say ${\rm d}\mu_{n}/{\rm d}\upsilon(z)=\tilde{\mu}_{n}(z)/Z_{n}$
and $\upsilon\gg\upsilon^{\psi_{n}^{-k}}$ for $k\in\llbracket0,T\rrbracket$,
then
\[
Z_{n}=Z_{n}\int\mu_{n}^{\psi_{n,k}^{-1}}({\rm d}z)=\int\tilde{\mu}_{n}\circ\psi_{n}^{k}(z)\tfrac{{\rm d}\upsilon^{\psi_{n,k}^{-1}}}{{\rm d}\upsilon}(z)\,\upsilon({\rm d}z)\,,
\]
and the measure $Z_{n}\bar{\mu}_{n}({\rm d}z)$ shares the same normalising
constant as $\tilde{\mu}_{n}(z){\rm \upsilon}({\rm d}z)$,
\[
Z_{n}=\int Z_{n}\bar{\mu}_{n}({\rm d}z)=\sum_{k=1}^{T}\frac{1}{T}\int\tilde{\mu}_{n}\circ\psi_{n}^{k}(z)\tfrac{{\rm d}\upsilon^{\psi_{n,k}^{-1}}}{{\rm d}\upsilon}(z)\,\upsilon({\rm d}z)\,.
\]
Consequently 
\[
\prod_{n=0}^{P-1}\Big[\frac{1}{N(T+1)}\sum_{i=1}^{N}\sum_{k=0}^{N}\frac{\tilde{\mu}_{n+1}\circ\psi_{n+1,k}}{\tilde{\mu}_{n}}\big(z_{n}^{(i)}\big)\frac{{\rm d}\upsilon^{\psi_{n+1,k}^{-1}}}{{\rm d}\upsilon}\big(z_{n}^{(i)}\big)\Big]\,,
\]
is an unbiased estimator of $Z_{P}/Z_{0}$.

\subsection{Integrator snippets and variance reduction}

The aim in this section is to establish that integrator snippets naturally
come with variance reduction properties. Variance reduction is only
one of the benefits of our methodology but the result below provides
us with a crucial principled criterion to determine the parameters
of the integrator used in our algorithms. We use the following simplified
notation throughout. For $\{\psi_{k}\colon\mathsf{Z}\rightarrow\mathbb{R},k\in\llbracket0,T\rrbracket\}$,
invertible mappings, define
\[
\bar{\mu}(k,{\rm d}z):=\frac{1}{T+1}\mu_{k}({\rm d}z)\,,
\]
with $\mu_{k}:=\mu^{\psi_{k}^{-1}}$. 
\begin{prop}
\label{prop:variance-reduc-barmu} For any $f\colon\mathsf{Z}\rightarrow\mathbb{R}$
$\mu$-integrable, with $\bar{f}(k,z):=f\circ\psi_{k}(z)$, we have
\begin{enumerate}
\item unbiasedness
\[
\mathbb{E}_{\bar{\mu}}\left(\mathbb{E}_{\bar{\mu}}(\bar{f}(\mathrm{T},\check{Z})\mid\check{Z})\right)=\mathbb{E}_{\bar{\mu}}\left(\bar{f}(\mathrm{T},\check{Z})\right)=\mathbb{E}_{\mu}\left(f(Z)\right)\,,
\]
\item variance reduction
\[
{\rm var}_{\bar{\mu}}\left(\mathbb{E}_{\bar{\mu}}\big(\bar{f}(\mathrm{T},\check{Z})\mid\check{Z}\big)\right)={\rm var}_{\mu}\left(f(Z)\right)-\mathbb{E}_{\bar{\mu}}\big({\rm var}_{\bar{\mu}}(\bar{f}(\mathrm{T},\check{Z})\mid\check{Z})\big)\,.
\]
\end{enumerate}
\end{prop}

\begin{proof}
The first statement follows directly from (\ref{sec:Discussion})
by adaptation of the notation. The variance decomposition identity
yields
\[
{\rm var}_{\bar{\mu}}\left(\bar{f}(\mathrm{T},\check{Z})\right)={\rm var}_{\bar{\mu}}\left(\mathbb{E}_{\bar{\mu}}\big(\bar{f}(\mathrm{T},\check{Z})\mid\check{Z}\big)\right)+\mathbb{E}_{\bar{\mu}}\big({\rm var}_{\bar{\mu}}(\bar{f}(\mathrm{T},\check{Z})\mid\check{Z})\big)\,,
\]
but from the first statement,
\begin{align*}
{\rm var}_{\bar{\mu}}\left(\bar{f}(\mathrm{T},\check{Z})\right) & =\mathbb{E}_{\bar{\mu}}\left(\bar{f}(\mathrm{T},\check{Z})^{2}\right)-\mathbb{E}_{\bar{\mu}}\left(\bar{f}(\mathrm{T},\check{Z})\right)^{2}\\
 & =\mathbb{E}_{\mu}\left(f(Z)^{2}\right)-\mathbb{E}_{\mu}\left(f(Z)\right)^{2}\\
 & ={\rm var}_{\mu}\left(f(Z)\right).
\end{align*}
\end{proof}
The result is clearly relevant when using the ``Rao-Blackwellized
estimator,''
\[
\check{\mu}(f)=N^{-1}\sum_{i=1}^{N}\mathbb{E}_{\bar{\mu}}\big(\bar{f}(\mathrm{T},\check{Z}^{(i)})\mid\check{Z}^{(i)}\big)\,,
\]
which corresponds to the estimator suggested by (\ref{eq:unfolding-HMC-setup}).
For $\check{Z}^{(1)},\ldots,\check{Z}^{(N)}\overset{{\rm iid}}{\sim}\bar{\mu}$
the variance of $\check{\mu}(f)$ is naturally given by ${\rm var}_{\bar{\mu}}\left(\mathbb{E}_{\bar{\mu}}\big(\bar{f}(\mathrm{T},\check{Z})\mid\check{Z}\big)\right)/N$
and the second statement of Proposition~\ref{prop:variance-reduc-barmu}
tells us that using $T>1$ is in general better than setting $T=1$.
The term $\mathbb{E}_{\bar{\mu}}\big(\bar{f}(\mathrm{T},\check{Z})\mid\check{Z}\big)-f(K,\check{Z})$
for $K\sim\bar{\mu}(k\mid\check{Z})$ can be thought of as a control
variate for an estimator relying on copies of $f(K,\check{Z})$.

The distributional assumption on $\check{Z}^{(1)},\ldots,\check{Z}^{(N)}$
is naturally not satisfied in the context of algorithm Alg.~\ref{alg:Folded-PDMP-SMC-1}.
However the variance terms considered above are precisely the terms
appearing in the analysis of SMC samplers \cite[Chapter 9]{del2004feynman}.
Indeed the variance of estimators obtained from an SMC can be decomposed
as the sum of local variances at each iterations, which involve precisely
the terms above, as $N\rightarrow\infty$.

An important implication of the variance identity is that for given
$\mu$ and $f$ an (ODE, integrator) pair, or more generally transformations,
should be chosen or tuned to maximize $\mathbb{E}_{\bar{\mu}}\big({\rm var}_{\bar{\mu}}(\bar{f}(\mathrm{T},\check{Z})\mid\check{Z})\big)$,
that is deterministic integration should be performed along the roughest
possible directions $k\mapsto f\circ\psi_{k}(\check{Z})$ in order
to absorb as much Monte Carlo variability as possible. 

It is worth noting that these ideas are not specific to the SMC context
and could be used, for example, in the context of PISA (Pushforward
Importance SAmpling) discussed in Appendix~\ref{subsec:PISA}.

\subsection{Direct extensions \protect\label{subsec:Straightforward-generalizations}}

It should be clear that the algorithm we have described lends itself
to numerous generalizations, which we briefly discuss below. This
can be skipped on a first reading.

There is no reason to limit the number of snippets arising from a
seed particle to one, which could be of interest on parallel machines.
For example the velocity of a given seed particle can be refreshed
multiple times, resulting in partial copies of the seed particle from
which integrator snippets can be grown.

The main scenario motivating this work, corresponds to the choice,
for $n\in\llbracket0,P\rrbracket$, of$\{\psi_{n,k}=\psi_{n}^{k},k\in\llbracket0,T\rrbracket\}$
for a given $\psi_{n}\colon\mathsf{Z}\rightarrow\mathsf{Z}$. As should
be apparent from the theoretical justification this can be replaced
by a general family of invertible mappings $\{\psi_{n,k}\colon\mathsf{Z}\rightarrow\mathsf{Z},k\in\llbracket0,T\rrbracket\}$
where the $\psi_{n,k}$'s are now not necessarily required to be measure
preserving in general, in which case the expression for $w_{n,k}(z)$
may involve additional terms of the ``Jacobian'' type. These mappings
may correspond to integrators other than those of Hamilton's equations
but may be more general deterministic mappings; $T$ may not have
any temporal meaning anymore and only represent the number of deterministic
transformations used in the algorithm. More specifically, asssuming
that $\upsilon\gg\mu_{n}$ and $\upsilon\gg\upsilon^{\psi_{n,k}^{-1}}$
for some $\sigma$-finite dominating measure $\upsilon$ and letting
$\mu_{n}(z):={\rm d}\mu_{n}/{\rm d}\upsilon(z)$ the required weights
are now of the form (see Lemmas~\ref{thm:measure-theoretic-transform}-\ref{lem:Billingsley-problem324})
\[
\bar{w}_{n,k}(z):=\frac{1}{T+1}\frac{\mu_{n}\circ\psi_{n,k}(z)}{\mu_{n-1}(z)}\frac{{\rm d}\upsilon^{\psi_{n,k}^{-1}}}{{\rm d}\upsilon}(z)\,.
\]
Again when $\upsilon$ is the Lebesgue measure and $\psi_{n,k}$ a
diffeomorphism, then ${\rm d}\upsilon^{\psi_{n,k}^{-1}}/{\rm d}\upsilon$
is the absolute value of the determinant of the Jacobian of $\psi_{n,k}$.
Non uniform weights may be ascribed to each of these transformations
in the definition of $\bar{\mu}$ (\ref{eq:nu-detailed-balance}).
A more useful application of this generality is in the situation where
it is believed that using multiple integrators, specialised in capturing
different geometric features of the target density, could be beneficial.
Hereafter we simplify notation by setting $\nu\leftarrow\mu_{n}$
and $\mu\leftarrow\mu_{n+1}$. For the purpose of illustration consider
two distinct integrators $\psi_{i}$, $i\in\llbracket2\rrbracket$
(again $n$ disappears from the notation) we wish to use, each for
$T_{i}\in\mathbb{N}$ time steps with proportions $\gamma_{i}\geq0$
such that $\gamma_{1}+\gamma_{2}=1$. Again with $\mu=\pi\otimes\varpi$
define the mixture
\[
\bar{\mu}(i,k,{\rm d}z):=\frac{\gamma_{i}}{T_{i}+1}\mu^{\psi_{i,k}^{-1}}({\rm d}z)\mathbf{1}\{k\in\llbracket0,T_{i}\rrbracket\}\,,
\]
which still possesses the fundamental property that if $z\sim\bar{\mu}$
(resp. $(i,z)\sim\bar{\mu}$), then with $(i,k)\sim\bar{\mu}(k,i\mid z)$
(resp. $k\sim\bar{\mu}(k\mid i,z)$) we have $\psi_{i,k}(z)\sim\mu$.
 It is possible to aim to sample from $\bar{\mu}({\rm d}z)$, in
which case the pair $(i,k)$ plays the rôle played by $k$ in the
earlier sections. However, in order to introduce the sought persistency,
that is use either $\psi_{1}$ or $\psi_{2}$ when constructing a
snippet, we focus on the scenario where the pair $(i,z)$ plays the
rôle played by $z$ earlier. In other words the target distribution
is now $\bar{\mu}(i,{\rm d}z)$, which is to $\mu(i,{\rm d}z):=\gamma_{i}\cdot\mu({\rm d}z)$
what $\bar{\mu}_{n}({\rm d}z)$ in (\ref{eq:nu-detailed-balance})
was to $\mu_{n}({\rm d}z)$. The mutation kernel corresponding to
(\ref{eq:nu-detailed-balance}) is given by
\[
\bar{M}_{\nu,\mu}(i,z;j,{\rm d}z'):=\sum_{k=0}^{T_{i}}\bar{\nu}(k\mid i,z)R(i,\psi_{i,k}(z);j,{\rm d}z')\,,
\]
with now the requirement that $\bar{\nu}R(i,{\rm d}z)=\nu(i,{\rm d}z)=\gamma_{i}\cdot\nu({\rm d}z)$.
A straightforward choice is $R(i,z;j,{\rm d}z')=\gamma_{j}\cdot R_{0}(z,{\rm d}z')$
with $\nu R_{0}({\rm d}z')=\nu({\rm d}z')$. With these choices and
using the near-optimal backward kernel simple substitutions $(i,z)\leftarrow z$
and $(j,z')\leftarrow z'$ yields
\[
\frac{{\rm d}\bar{\mu}\otimeswapped\bar{L}_{\nu,\mu}}{{\rm d}\bar{\nu}\otimes\bar{M}_{\nu,\mu}}(i,z;j,z')=\frac{{\rm d}\bar{\mu}}{{\rm d}\nu}(j,z')\,,
\]
and justifies our abstract presentation. This idea will be exploited
later on in order to design adaptive algorithms.

Another direct extension consists of generalizing the definition of
$\bar{\mu}_{n}$ by assigning non-uniform and possibly state-dependent
weights to the integrator snippet particles as follows
\[
\bar{\mu}_{n}(k,{\rm d}z)=\omega_{n,k}\circ\psi_{n,k}(z)\,\mu_{n,k}({\rm d}z)\,,
\]
with $\omega_{n,k}\colon\mathsf{Z}\rightarrow\mathbb{R}_{+}$ for
$k\in\llbracket0,T\rrbracket$ and such that $\sum_{k=0}^{T}\omega_{n,k}(z)=1$
for any $z\in\mathsf{Z}$; this should be contrasted with (\ref{eq:bar-mu-n-marginal}).
 In Appendix~\ref{subsec:Optimal-weights} we show how such weighting
can be optimised to reduce variance of expectation estimators, with
potentially negative weights.

We leave exploration of some of these generalisations for future work.

\subsection{Rational and computational considerations \protect\label{subsec:Computational-considerations}}

Our initial motivation for this work was that computation of $\{z_{n,k},k\in\llbracket T\rrbracket\}$
and $\{w_{n,k},k\in\llbracket T\rrbracket\}$ most often involve common
quantities, leading to negligible computational overhead, and reweighting
offers the possibility to use all the states of a snippet in a simulation
algorithm rather than the endpoint only. There are other reasons for
which this approach may be beneficial. 

The benefit of using all states along integrator snippets in sampling
algorithms has been noted in the literature. For example, in the context
of Hamiltonian integrators, with $H_{n}(z):=-\log\mu_{n}(z)$ and
$\psi_{n,k}=\psi_{n}^{k}$, it is known that for $z\in\mathsf{Z}$
the mapping $k\mapsto H_{n}\circ\psi_{n}^{k}(z)$ is typically oscillatory,
motivating for example the x-tra chance algorithm of \cite{campos2015extra}.
The ``windows of state'' strategy of \cite{neal1994improved,neal2011mcmc}
effectively makes use of the mixture $\bar{\mu}$ as an instrumental
distribution and improved performance is noted, with performance seemingly
improving with dimension on particular problems; see Appendix~\ref{sec:MCMC-with-integrator}
for a more detailed discussion; see also \cite{betancourt2017conceptual}.
Further averaging is well known to address scenarios where components
of $x_{t}$ evolve on different scales and no choice of a unique integration
time $\tau:=T\times\epsilon$ can accommodate all scales \cite[section 3.2]{neal2011mcmc,hoffman2022tuning};
averaging addresses this issue effectively, see Example~\ref{exa:hamilton-different-scales}.
We also note that, keeping $T\times\epsilon$ constant, in the limit
as $\epsilon\rightarrow0$ the average in (\ref{eq:weight-folded-HMC})
corresponds to the numerical integration, Riemann like, along the
contour of $H_{n}(z)$, effectively leading to some form of Rao-Blackwellization
of this contour; this is discussed in detail in Appendix~\ref{sec:Preliminary-theoretical-characte}.

Another benefit we have observed with the SMC context is the following.
Use of a WF-SMC strategy involves comparing particles within each
Markov snippet arising from a single seed particle, while our strategy
involves comparing all particles across snippets, which proves highly
beneficial in practice. This seems to bring particular robustness
to the choice of the integrator parameters $\epsilon$ and $T$ and
can be combined with highly robust adaptation schemes taking advantage
of the population of samples, in the spirit of \cite{hoffman2021adaptive,hoffman2022tuning};
see Subsections~\ref{subsec:first-example-simulations} and \ref{subsec:example-filamentary-distributions}.

At a computational level we note that, in contrast with standard SMC
or WF-SMC implementations relying on an MH mechanism, the construction
of integrator snippets does not involve an accept reject mechanism,
therefore removing control flow operations and enabling lockstep computations
on GPUs -- actual implementation of our algorithms on such architectures
is, however, left to future work.

Finally, when using integrators of Hamilton's equations we naturally
expect similar benefits to those enjoyed by HMC. Let $d\in\mathbb{N}$
and let $\mathsf{X=\mathbb{R}}^{d}$. We know that in certain scenarios
\cite{beskos2013optimal,calvo2021hmc}, the distributions $\{\mu_{n,d},d\in\mathbb{N},n\in\llbracket T(d)\rrbracket\}$
are such that for $n\in\mathbb{N}$, $\log\big(\nicefrac{\mu_{n,d}\circ\psi_{n,d}^{k}(z)}{\mu_{n,d}(z)}\big)\rightarrow_{d\rightarrow\infty}\mathcal{N}(\mu_{n},\sigma_{n}^{2})$,
that is the weights do not degenerate to zero or one: in the context
of SMC this means that the part of the importance weight (\ref{eq:weight-folded-HMC})
arising from the mutation mechanism does not degenerate. Further,
with an appropriate choice of schedule, i.e. sequence $\{\mu_{n,d},n\in\llbracket T(d)\rrbracket\}$
for $d\in\mathbb{N}$, ensures that the contribution $\nicefrac{\mu_{n,d}(z)}{\mu_{n-1,d}(z)}$
to the importance weights (\ref{eq:weight-folded-HMC}) is also stable
as $d\rightarrow\infty$. As shown in \cite{beskos2013optimal,andrieu2018sampling,calvo2021hmc},
while direct important sampling may require an exponential number
of samples as $d$ grows, the use of such a schedule may reduce complexity
to a polynomial order.

\subsection{Links to the literature}

Alg.~\ref{alg:Unfolded-PDMP-SMC-1}, and its reinterpretation Alg.~\ref{alg:Folded-PDMP-SMC-1},
are reminiscent of various earlier contributions and we discuss here
parallels and differences. This work was initiated in \cite{chang-2022}
and pursued and extended in \cite{mauro-phd}.

Readers familiar with the ``waste-free SMC'' algorithm (WF-SMC)
of \cite{dau2020waste} may have noticed a connection. Similarly to
Alg.~\ref{alg:Unfolded-PDMP-SMC-1}, seed particles are extended
with an MCMC kernel leaving $\mu_{n-1}$ invariant at iteration $n$,
yielding $N\times(T+1)$ particles subsequently whittled down to $N$
new seed particles. A first difference is that while generation of
an integrator snippet can be interpreted as applying a sequence of
deterministic Markov kernels (see Appendix~\ref{sec:Sampling-HMC-trajectories}
for a detailed discussion) what we show is that the mutation kernels
involved do not have to leave $\mu_{n-1}$ invariant; in fact we show
in Appendix~\ref{sec:Sampling-a-mixture:} that this indeed is not
a requirement, therefore offering more freedom. Further, it is instructive
to compare our procedure with the following two implementations of
WF-SMC using an HMC kernel (\ref{eq:HMC-kernel-deterministic}) for
the mutation. A first possibility for the mutation stage is to run
$T$ steps of an HMC update in sequence, where each update uses one
integrator step. Assuming no velocity refreshment along the trajectory
this would lead to the exploration of a random number of states of
our integrator snippet due to the accept/reject mechanism involved;
incorporating refreshment or partial refreshment would similarly lead
to a random number of useful samples. Alternatively one could consider
a mutation mechanism consisting of an HMC build around $T$ integration
steps and where the endpoint of the trajectory would be the mutated
particle; this would obviously mean discarding $T-1$ potentially
useful candidates. To avoid possible reader confusion, we note apparent
typos in \cite[Proposition 1]{dau2020waste}, provided without proof,
where the intermediate target distribution of the SMC algorithms,
$\bar{\mu}_{n}$ in our notation (see (\ref{eq:bar-mu-markov-snippet})),
seems improperly stated. The statement is however not used further
in the paper.

Our work shares, at first sight, similarities with \cite{rotskoff2019dynamical,thin2021neo},
but differs in several respects. In the discrete time setup considered
in \cite{thin2021neo} a mixture of distributions similar to ours
is also introduced. Specifically for a sequence $\{\omega_{k}\geq0,k\in\mathbb{Z}\}$
such that $\#\{\omega_{k}\neq0,k\in\mathbb{Z}\}<\infty$ and $\sum_{l\in\mathbb{Z}}\omega_{k}=1$,
the following mixture is considered (in the notation of \cite{thin2021neo}
their transformation is $T=\psi^{-1}$ and we stick to our notation
for the sake of comparison and avoid confusion with our $T$),
\[
\bar{\mu}({\rm d}z)=\sum_{k\in\mathbb{Z}}\omega_{k}\mu_{k}({\rm d}z)\,.
\]
The intention is to use the distribution $\bar{\mu}$ as an importance
sampling proposal to estimate expectations with respect to $\mu$,
\begin{align*}
\int f(z)\frac{{\rm d}\mu}{{\rm d}\bar{\mu}}(z)\bar{\mu}({\rm d}z) & =\sum_{k\in\mathbb{Z}}\omega_{k}\int f(z)\frac{{\rm d}\mu}{{\rm d}\bar{\mu}}(z)\mu^{\psi^{-k}}({\rm d}z)\\
 & =\sum_{k\in\mathbb{Z}}\omega_{k}\int f\circ\psi^{-k}(z)\frac{{\rm d}\mu}{{\rm d}\bar{\mu}}\circ\psi^{-k}(z)\mu({\rm d}z)\\
 & =\sum_{k\in\mathbb{Z}}f\circ\psi^{-k}(z)\omega_{k}\frac{\mu\circ\psi^{-k}(z)}{\sum_{l\in\mathbb{Z}}\omega_{l}\,\mu\circ\psi^{l-k}(z)}\mu({\rm d}z)\,,
\end{align*}
where the last line holds when $\upsilon\gg\mu$, $\upsilon{}^{\psi}=\upsilon$
and $\mu(z)={\rm d}\mu/{\rm d}\upsilon(z)$. The rearrangement on
the second and third line simply capture the fact noted earlier in
the present paper that with $z\sim\mu$ and $k\sim{\rm Cat}\big(\{\omega_{l}:l\in\mathbb{Z}\}\big)$
then $\psi^{-k}(z)\sim\bar{\mu}$. As such NEO relies on generating
samples from $\mu$ first, typically exact samples, which then undergo
a transformation and are then properly reweighted. In contrast we
aim to sample from a mixture of the type $\bar{\mu}$ directly, typically
using an iterative method, and then exploit the mixture structure
to estimate expectations with respect to $\mu$. Note also that some
of the $\psi^{l-k}(z)$ terms involved in the renormalization may
require additional computations beyond terms for which $\omega_{l-k}\neq0$.
Our approach relies on a different important sampling identity
\[
\int f\circ\psi^{k}(z)\frac{{\rm d}\mu^{\psi^{-k}}}{{\rm d}\bar{\mu}}(z)\bar{\mu}({\rm d}z)=\mu(f)\,.
\]
We note however that the conformal Hamiltonian integrator used in
\cite{rotskoff2019dynamical,thin2021neo} could be used in our framework,
which we leave for future investigations. Their NEO-MCMC targets a
``posterior distribution'' $\mu'({\rm d}z)=Z^{-1}\mu({\rm d}z)L(z)$
and inspired by the identity
\begin{align*}
\int f(z)L(z)\frac{{\rm d}\mu}{{\rm d}\bar{\mu}}(z)\bar{\mu}({\rm d}z) & =\sum_{k\in\mathbb{Z}}\omega_{k}\int f\circ\psi^{-k}(z)L\circ\psi_{k}^{-1}(z)\frac{{\rm d}\mu}{{\rm d}\bar{\mu}}\circ\psi^{-k}(z)\mu({\rm d}z)\,,
\end{align*}
which is an expection of $\bar{f}(k,z)=f\circ\psi^{-k}(z)$ with respect
to
\[
\bar{\mu}'(k,{\rm d}z)\propto\mu({\rm d}z)\cdot\omega_{k}\frac{{\rm d}\mu}{{\rm d}\bar{\mu}}\circ\psi^{-k}(z)\cdot L\circ\psi^{-k}(z)
\]
and they consider algorithms targetting $\bar{\mu}'({\rm d}z)$ which
relying on exact samples from $\mu$ to construct proposals, resulting
in a strategy which shares the weaknesses of an independent MH algorithm.

Much closer in spirit to our work is the ``window of states'' idea
proposed in \cite{neal1994improved,neal2011mcmc}, developed in \cite{betancourt2017conceptual},
in the context of HMC algorithms. Although the MCMC algorithms of
\cite{neal1994improved,neal2011mcmc} ultimately target $\mu$ and
are not reversible, they involve a reversible MH update with respect
to $\bar{\mu}$ and additional transitions permitting transitions
from $\mu$ to $\bar{\mu}$ and $\bar{\mu}$ to $\mu$; see Section~\ref{sec:MCMC-with-integrator}
for full details. Note that using these MCMC updates as mutation kernels
in an SMC would not lead to integrator snippet SMC as they target
$\mu$, not $\bar{\mu}$.

A link we have not explored in the present manuscript is that to normalizing
flows \cite{tabak2010density,Marzouk2016} and related literature.
In particular the ability to tune the parameter of the leapfrog integrator
suggests that our methodology lends itself to learning normalising
flows. We finally note an analogy of such approaches with umbrella
sampling \cite{thiede2016eigenvector,TORRIE1977187}, although the
targetted mixture is not of the same form, and the ``Warp Bridge
Sampling'' idea of \cite{c61adb08-fb3c-3ba4-9867-a07d61c50217}.

\section{Adaptation with Integrator Snippet SMC \protect\label{sec:Adaptation-with-Integrator}}

\subsection{Need for adaptation and review of classical criteria}

Monte Carlo algorithms relying on discretization of Hamilton's equation
for the exploration of the support of a distribution of interest offer
the promise of very efficient algorithms, as supported by practical
\cite{neal2011mcmc,carpenter2017stan} and theoretical evidence \cite{beskos2013optimal,calvo2021hmc}
in the MCMC scenario. However such procedures are known for their
brittleness e.g. \cite{livingstone2022barker} in the particular case
of the Langevin algorithm, and require precise calibration of their
parameters in order to achieve their potential. In their simplest
implementation these procedures rely on three, related, parameters:
the discretisation stepsize $\epsilon>0$, $T\in\mathbb{N}$ the number
of integration steps and $\tau=\epsilon T\in\mathbb{R}_{+}$ the effective
integration time, that is the time horizon we would have liked to
use if Hamilton's equations were tractable. Qualitatively the effect
of a poor choice of these parameters is as follows. 

For fixed $\tau>0$, too large a stepsize $\epsilon$ will result
in departure of the integrator from the Hamiltonian flow, therefore
defeating the purpose of the method and empirically affecting performance,
while too small an $\epsilon>0$ requires $T$ to be large, therefore
increasing computational complexity to reach a vicinity of the Hamilton
flow at time $\tau$. This issue has motivated a substantial body
of work e.g. \cite{beskos2013optimal,calvo2021hmc}, which essentially
provide guidelines on the choice of $\epsilon$ based on statistical
properties of the Markov chain at stationarity; in particular monitoring
the expected acceptance ratio is shown to be a suitable criterion.
In practice the algorithm is run for a value of $\epsilon$, the empirical
expected acceptance probability estimated in preliminary runs or online,
and $\epsilon$ adjusted to reach a nominal value suggested by theory.
More specifically with the acceptance ratio $r(z;\epsilon,T)=\mu\circ\psi^{T}(z)/\mu(z)$
and acceptance ratio $\alpha(z;\epsilon,T)=1\wedge r(z;\epsilon,T)$
theory suggests finding $\epsilon$ such that
\[
\chi(\epsilon,T):=\mathbb{E}_{\mu}\big(\alpha(z;\epsilon,T)\big)\approx0.67,\,\text{subject to}\,\tau=T\epsilon\,.
\]

Poor choice of the effective integration time $\tau$ is also important.
Too large a $\tau$ may result in the end point of the integrator
snippet being close to the initial state, which should be clear in
the two-dimensional setting, therefore wasting computations. Too small
a $\tau$ may not take full advantage of such updates. However choosing
$\tau$ optimally in general is a largely open problem with, to the
best of our knowledge, very few theoretical results available. An
exception is \cite[Theorem 1.3]{chen2022optimal} where, assuming
the log-concavity ot the target probability density and tractability
of Hamilton's equations, an optimal value for $\tau$ in terms of
a coupling argument, which does not seem to have been exploited in
practice; we discuss this further in the context of IS2MC later on.
A popular and pragmatic approach was suggested in \cite{hoffman2014no}
to determine a good value of $\tau$. The aim set in \cite{hoffman2014no}
is to identify the time $\tau_{{\rm u}}$ (or $k_{\mathrm{u}}$ the
corresponding number of steps for a given $\epsilon$) where $t\mapsto\|\psi_{t,x}(z)-x\|^{2}$,
with again $\psi_{t,x}$ indicates the first component of $\psi_{t}$,
reaches its first maximum. While not a foolproof criterion, the no-U-turn
criterion has been shown to be useful in practice \cite{carpenter2017stan}.

SMC algorithms require additional calibration. Indeed, in the context
of SMC we are further required to choose the interpolating sequence
$\{\pi_{n},n\in\llbracket0,P\rrbracket\}$ which we hereafter assume
to arise as the discretization of a family of probability distributions
$\{\pi(\cdot;\gamma),\gamma\in[0,1]\}$ on $(\mathsf{X},\mathscr{X})$,
or ``probability path'', such that $\pi(\cdot;1)=\pi(\cdot)$. For
example in a Bayesian context one may consider
\begin{equation}
\pi({\rm d}x;\gamma)\propto L^{\gamma}(x)\eta({\rm d}x)\,,\label{eq:def-pi-smc-adaptive}
\end{equation}
where $\eta$ is a probability distribution, the prior, and $L(\cdot)$
is the likelihood function, which is combined with the prior in a
gradual way. Defining $\mu({\rm d}z;\gamma):=\pi({\rm d}x;\gamma)\varpi({\rm d}v)$
we therefore need to choose $\Gamma:=(P,\{\gamma_{i}\in[0,1],i\in\llbracket0,P\rrbracket\})\in\cup_{k=1}^{\infty}\{k\}\times\mathbb{R}_{+}^{k+1}$
to determine $\{\mu_{n},n\in\llbracket0,P\rrbracket\}$, and subsequently
$\{\bar{\mu}_{n},n\in\llbracket0,P\rrbracket\}$ in the integrator
snippet framework. For an SMC running along a probability path $\{\mu(\cdot;\gamma),\gamma\in[0,1]\}$
and using an optimal backward kernel, $\Gamma$ is usually constructed
in a sequential manner as follows. For $n\geq0$ define
\begin{equation}
\chi_{\gamma,n+1}(\gamma):=\mathbb{E}_{\mu_{n}}\left\{ \left(\frac{{\rm d}\mu(\cdot;\gamma)}{{\rm d}\mu_{n}}(z)\right)^{2}\right\} \,,\label{eq:chi-gamma-only}
\end{equation}
then we have $\chi_{\gamma,n+1}(\gamma_{n})=1$, $\chi_{n}(\gamma)\geq1$
from Jensen's inequality and in typical scenarios one expects $\gamma\mapsto\chi_{\gamma,n+1}(\gamma)$
to be increasing for $\gamma\geq\gamma_{n}$ Note that the indexed
$\gamma$ should not be thought of as a value of the corresponding
parameter, but as its name only; we use $\chi_{\epsilon},\chi_{\theta},\ldots$
for other parameters and criteria below. This leads to the recursive
definition
\begin{equation}
\gamma_{n+1}=\sup\{\gamma\colon\gamma_{n}\leq\gamma\leq1,\chi_{\gamma,n+1}(\gamma)\leq\bar{\chi}_{\gamma,n+1}\}\wedge1\,,\label{eq:update-gamma}
\end{equation}
for some user defined $\bar{\chi}_{\gamma,n+1}>1$, a user defined
tolerance, hence defining $\{\mu_{n},n\in\llbracket0,P\rrbracket\}$,
and subsequently $\{\bar{\mu}_{n},n\in\llbracket0,P\rrbracket\}$.
This criterion is justified by results on self-normalized importance
sampling bounds similar to those developed in Theorem~(\ref{thm:PISA-ESS}).
In practice this criterion is approximated with samples from $\mu_{n}$
obtained at the previous iteration of the SMC algorithm, defining
the function estimate $\gamma\mapsto\hat{\chi}_{n+1,\gamma}(\gamma)$.
Traditionally the above is formulated in terms of the so-called effective
sample size (ESS) $N/\hat{\chi}_{n+1,\gamma}(\gamma)$ and the user
defines a tolerance on the ESS scale.

In what follows we re-use and adapt some of these ideas in the context
of ${\rm IS^{2}MC}$ in order to define $(\gamma_{n},T_{n})$ and
the parameter of a distribution distribution on $\epsilon_{n}$.

Integrator snippet SMC naturally provide extensive population level
experimental information on the statistical properties arising from
the use of an integrator in an SMC step, as illustrated in Fig.~\ref{fig:snippets-population-information}.
This seems very promising to calibrate tuning parameters robustly.
However, crucially, we show below that the mixture structure of the
distributions targetted by integrator snippets, $\{\bar{\mu}_{n},n\in\llbracket0,P\rrbracket\}$,
presents another opportunity to develop novel robust and computationally
efficient adaptive SMC strategies to tune $(\tau,\epsilon,T)$. 

\begin{figure}
\centering{}\includegraphics[width=0.9\textwidth]{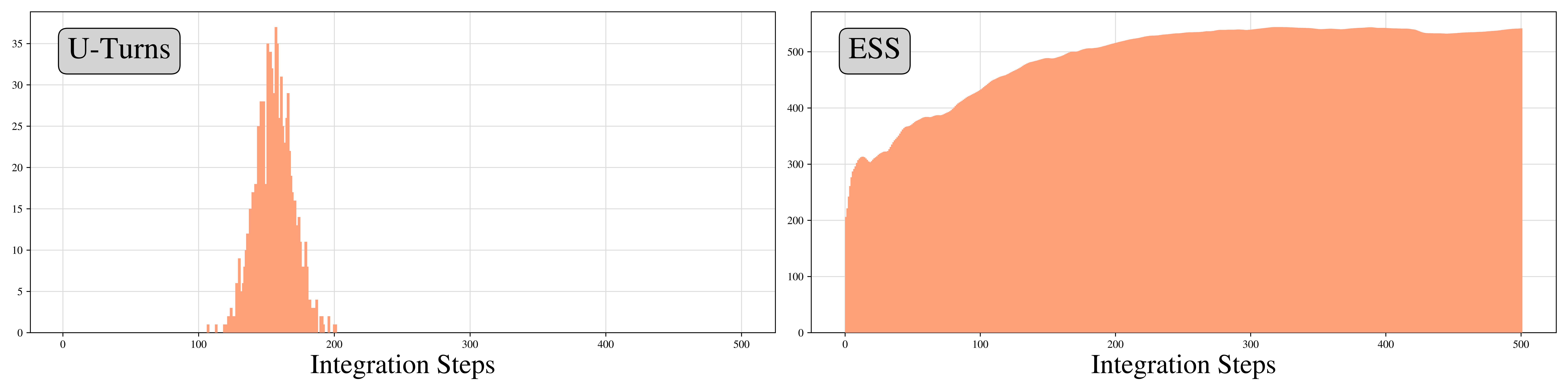}
\caption{Illustration of information available about a the Leapfrog integrator
in the particle population. For an SMC iteration, Left: histogram
of the U-turn positions ($k_{u}$) across snippets, Right: average
ESS along snippets for varying $k\in\llbracket T\rrbracket$.}
\label{fig:snippets-population-information}
\end{figure}

\subsection{Adapting $\gamma$ only \protect\label{subsec:Adapting-gamma-only}}

\paragraph{Logistic regression for the Sonar dataset \protect\label{par:logistic_regression_sonar_no_adaptation}}

We consider sampling from the posterior distribution of a logistic
regression model with a focus on estimating the normalizing constant,
or evidence. Following \cite{dau2020waste} we consider the Sonar
dataset, previously used as testbed in \cite{chopin2015leavepimaindiansalone},
which with covariates $\{\xi_{i}\in\mathbb{R}^{61},i\in\llbracket208\rrbracket\}$
(containing intercept terms) and responses $\big\{ y_{i}\in\{-1,1\},i\in\llbracket208\rrbracket\big\}$
has likelihood function
\begin{align}
\mathbb{R}^{61}\ni x\mapsto L(x)=\prod_{i=1}^{208}\left(1+\exp(-y_{i}\cdot\xi_{i}^{\top}x)\right)^{-1}.\label{eqn:likelihood_sonar}
\end{align}
The prior distribution $\eta(x)$ is chosen to be a product of independent
normal distributions with mean zero and standard deviation equal to
$20$ for the intercept and $5$ for the other parameters. The probability
path $\gamma\mapsto\pi(\cdot;\gamma)$ is then defined as in (\ref{eq:def-pi-smc-adaptive})
and $\Gamma$ determined sequentially using the standard approach
described earlier.

We compare an Integrator Snippet and the implementation of Waste-Free
SMC of \cite{dau2020waste}, which uses a random walk Metropolis-Hastings
kernel with covariance adaptively computed as $2.38^{2}\hat{\Sigma}/d$,
where $\hat{\Sigma}$ is the empirical covariance matrix obtained
from the particles in the previous SMC step. For the Integrator Snippet
we choose $\psi_{n}$ to be the one-step Leapfrog integrator and $\varpi_{n}=\mathcal{N}(0,{\rm Id})$. 

HS selects the next tempering parameter using a commonly used ESS-based
update \cite{chopin2020introduction}, closely related to (\ref{eq:chi-gamma-only})-(\ref{eq:update-gamma}),
and for both algorithms we target an ESS of $0.8N$ at each iteration.
HS uses a stepsize of $\epsilon=0.1$, selected using graduate descent
to give good results. We will show in Section~\ref{sec:Adapting-the-stepsize}
that this value is automatically recovered by our adaptation strategy.

\paragraph{Fixed Budget Comparison}

We investigate the ability of the two algorithms to estimate the log-normalizing
constant for a fixed total budget $N(T+1)=10,000$, where for HS (resp.
WF) $N$ the number of seed particles (resp. number of Markov chain
snippets) and $T$ is the number of Leapfrog steps (resp. length of
the Markov chain snippets). An important point is that both algorithms
have comparable computational costs. Indeed the gradient $\nabla_{x}\log L(x)$
required by the Leapfrog integrator and $\log L(x)$ required for
the evaluation of the weights $\bar{w}_{n,k}$ share the same costly
terms as
\begin{align}
\nabla_{x}\log L(x)=-\sum_{i=1}^{208}\frac{\exp(-y_{i}\cdot\xi_{i}^{\top}x)}{1+\exp(-y_{i}\cdot\xi_{i}^{\top}x)}y_{i}\xi_{i}\,,\label{eq:grad-log-L}
\end{align}
while the Metropolis-Hastings ratio of WF-SMC also requires the evaluation
of $\log L(x)$. 

Fig.~\ref{fig:sonar-normalizing-constants-boxplot} displays box-plots
illustrating the variability and bias of the log-normalizing constant
estimates across $100$ independent runs of the two algorithms. We
observe that for this fixed computational budget, what we interpret
to be bias seems to increase together with variability for WF as the
Markov snippets get shorter, whereas HS produces broadly consistent
estimates for all $(N,T)$ pairs considered. A similar bias in the
WF estimates was observed with three independent code implementations:
two separate implementations by the authors MCE and CZ and the original
implementation available in the Particles package \cite{chopin_github_repo}. 

\begin{figure}
\begin{centering}
\includegraphics[width=0.9\textwidth]{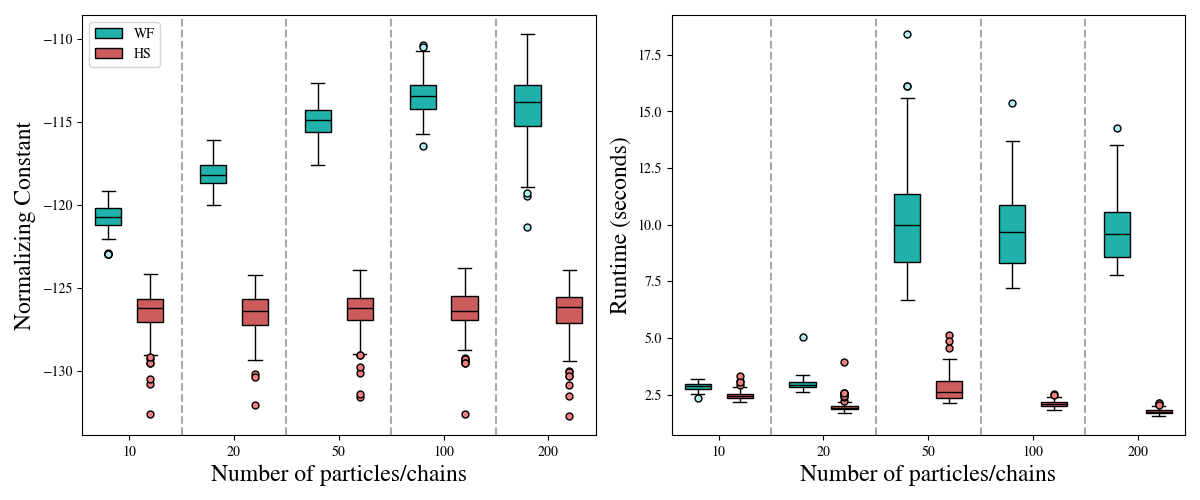}
\par\end{centering}
\centering{}\caption{Left: Variability of the estimates of the log-normalizing constant
for Waste-Free SMC (WF) and Hamiltonian Snippet (HS) for a fixed budget
$N(T+1)=10000$. Right: Variability of runtime in seconds for the
two algorithms. The lower runtime for HS is due to two factors: the
lack of if-else statements to perform the MH accept-reject step, and
no covariance estimation.}
\label{fig:sonar-normalizing-constants-boxplot}
\end{figure}

\subsection{Adapting $\gamma$ and $\epsilon$ for fixed $T$ \protect\label{sec:Adapting-the-stepsize}}

We now focus on the scenario where $T$, i.e. the computational budget
per iteration is fixed. Full adaptation of $(\epsilon,\gamma,T)$
is addressed in Section~\ref{sec:Adaptation-of-tau} and uses the
update of this section.

\subsubsection{Methodology}

\paragraph{A mixture of competing integrator snippets}

In order to define and adaptive SMC able to select suitable stepsizes
we first introduce the following instrumental ``mother'' probability
distribution on which the target sequence $\{\bar{\mu}_{n}=\bar{\mu}({\rm \cdot};\gamma_{n},\theta_{n}),n\geq0\}$
builds on. For a user defined family of probability distributions
for $\epsilon$, $\{\nu_{\theta},\theta\in\Theta\}$ defined on $\big(\mathbb{R}_{+},\mathcal{B}(\mathbb{R}_{+})\big)$,
we consider $\bar{\mu}$ on $\big(\mathsf{Z}\times\mathbb{R}_{+},\mathscr{Z}\otimes\mathcal{B}(\mathbb{R}_{+})\big)$
such that for $(\theta,\gamma)\in\Theta\times[0,1]$ 
\begin{equation}
\bar{\mu}\big(k,{\rm d}(\epsilon,z);\gamma,\theta\big):=\frac{1}{T+1}\mu({\rm d}z;\epsilon,\gamma,k)\nu_{\theta}({\rm d}\epsilon)\,,\label{eq:target-mother-adaptive}
\end{equation}
where $\mu({\rm d}z;\epsilon,\gamma,k):=\mu^{\psi_{\epsilon}^{-k}}({\rm d}z;\gamma)$,
with $\mu(\cdot;\gamma)$ as defined below (\ref{eq:update-gamma}),
and $\psi_{\epsilon}$ is the integrator using $\epsilon>0$ as stepsize.
As we shall see the idea is to let a population of stepsizes $\epsilon$
of distribution $\nu_{\theta}$ compete, for a fitness criterion to
be determined, and adjust this distribution, i.e. $\theta\in\Theta$,
to focus on the most efficient values. As we shall see the conditional
independence structure of $\bar{\mu}$ plays a central rôle at several
stages in our algorithm; additional details are provided in Appendix~\ref{appsubsec:structure-mother-dist}
if needed. First, note that since for any $(k,\epsilon,\gamma)\in\mathbb{N}\times\mathbb{R}_{+}\times[0,1]$
\[
\int f\circ\psi_{\epsilon}^{k}(z)\mu^{\psi_{\epsilon}^{-k}}({\rm d}z;\gamma)=\mu(f;\gamma)\,,
\]
we still retain the property that $\mathbb{E}_{\bar{\mu}(\cdot;\gamma,\theta)}\big(f\circ\psi_{\epsilon}^{K}(Z)\big)=\mu(f;\gamma)=\mathbb{E}_{\bar{\mu}(\cdot;\gamma,\theta)}\{\mathbb{E}_{\bar{\mu}(\cdot;\gamma,\theta)}\big(f\circ\psi_{\epsilon}^{K}(Z)\mid\epsilon,Z\big)\}$,
or more explicitely,
\[
\int\biggl\{\sum_{k=0}^{T}\int f\circ\psi_{\epsilon}^{k}(z)\cdot\frac{{\rm d}\mu^{\psi_{\epsilon}^{-k}}}{{\rm d}\bar{\mu}}(z;\gamma,\epsilon)\bar{\mu}\big({\rm d}z\mid\epsilon;\gamma\big)\biggr\}\nu_{\theta}({\rm d}\epsilon)=\mu(f;\gamma)\,.
\]
Given this property, our goal is therefore to construct the sequence
of target distributions $\{\bar{\mu}_{n}=\bar{\mu}({\rm \cdot};\gamma_{n},\theta_{n}),n\geq0\}$
for some sequence $\{(\gamma_{n},\theta_{n})\in[0,1]\times\Theta,n\geq0\}$
determined sequentially to satisfy specific local criteria. 

To optimize $\gamma$ we use the update in (\ref{eq:update-gamma}),
which departs from usual practice since one would normally consider
the ESS for weighted particles representing $\{\bar{\mu}_{n}=\bar{\mu}({\rm \cdot};\gamma_{n},\theta_{n}),n\geq0\}$,
not $\{\mu_{n},n\geq0\}$. This would however incur a substantial
computational cost since full snippets would have to be estimated
for each value of $\gamma$ used in the iterative algorithm used to
achieve (\ref{eq:update-gamma}).

\paragraph{Performance measure for $\epsilon$}

To optimise $\theta$ we use the straightforward adaptation of Proposition~\ref{prop:variance-reduc-barmu}
to the multivariate scenario and aim to maximise the expected conditional
variance of a user defined function $f\colon\mathsf{Z}\rightarrow\mathbb{R}^{m}$,
for some $m\in\mathbb{N}$, along snippets, that is with $\mathbb{E}_{\theta,\gamma}(\cdot):=\mathbb{E}_{\bar{\mu}(\cdot;\gamma,\theta)}(\cdot)$
we aim to maximise 
\[
\theta\mapsto\chi_{\theta}(\theta;\gamma):=\mathrm{Tr}\left\{ \mathbb{E}_{\gamma,\theta}\left({\rm var}_{\gamma,\theta}\big(\bar{f}(\epsilon,K,Z)\mid\epsilon,Z\big)\right)\right\} \,,
\]
with, now, $\bar{f}(k,\epsilon,z):=f\circ\psi_{\epsilon}^{k}(z)$
-- hereafter we focus on $f(z)=x$ for its intuitive interpretation,
but other choices may be more suitable in particular applications.
This quantity can be thought of as the averaged variance absorbed
by Rao-Blackwellization, interpretable here as Riemannian-like integration
along snippets. With this criterion in hand we can now describe a
new approach to tracking the optimal values of $\theta$. Let $v_{\gamma}(\epsilon,z):=\mathrm{Tr}\big[{\rm var}_{\gamma}\big(\bar{f}\mid\epsilon,z\big)\big]$,
which is indeed independent of $\theta$, leaving $f$ and $T$ implicit
for notational simplicity. Using the structure of $\bar{\mu}$ in
detailed in Appendix~\ref{appsubsec:structure-mother-dist} one obtains
\begin{align}
\chi_{\theta}(\theta;\gamma) & =\int v_{\gamma}(\epsilon,z)\,\bar{\mu}\big({\rm d}(\epsilon,z);\gamma,\theta\big)\nonumber \\
 & =\int v_{\gamma}(\epsilon,z)\,\bar{\mu}\big({\rm d}z\mid\epsilon;\gamma\big)\nu_{\theta}({\rm d}\epsilon)\,.\label{eq:def-chi-theta}
\end{align}
We briefly illustrate the relevance of this measure, that is its ability
to discriminate between good and badly performing stepsizes. To that
purpose we consider the following simple experiment, where we run
Hamiltonian Snippet on the Sonar logistic regression problem with
$N=2500$, $T=30$ and $\nu_{\theta}$ is the uniform distribution
over $18$ log-linearly spaced step sizes between $0.001$ and $10$;
$\theta$ has no use here and is therefore not adapted, but kept for
pure notational consistency. At each iteration of the SMC algorithm,
we estimate $\epsilon\mapsto\bar{v}_{\gamma}(\epsilon):=\int v_{\gamma}(\epsilon,z)\,\bar{\mu}\big({\rm d}z\mid\epsilon;\gamma\big)$
using (\ref{eq:explicit-express-v_epsilon}). These quantities are
displayed in Fig.~\ref{fig:barupsilon_function_of_epsilon} for four
different values of the tempering parameter $\gamma$, as the SMC
sampler progresses. Clearly some values of $\epsilon$ lead to a significantly
larger variance reduction than others and the set of effective values
evolves with $\gamma$, that is as the likelihood function is incorporated
in the sampling. Notice in addition the asymmetry of the criterion
around the mode: too large a stepsize leads to a sudden drop in performance
while for smallerl stepsizes deterioration is smoother. In the former
scenario the integrator diverges from the underlying dynamic while
in the latter the integrator is more precise but does not explore
much of the space around the seed particle. Fig.~\ref{fig:barupsilon_function_iteration}
displays the same criterion but now as a function of $\gamma$ for
all the values of $\epsilon$ considered. We see clearly here that
large values of $\epsilon$ are initially preferred for small values
of $\gamma$ as the target distribution is close to the prior distribution,
a normal density, leading to a highly faithful integrator even for
large stepsize values. Performance then drops sharply as the likelihood
is incorporated and more moderate values of $\epsilon$ perform better.
Interestingly the criterion is also able to discriminate between smaller
values of $\epsilon$ for which the integrator is faithful across
all temperatures $\gamma$, but among which some are more able to
explore the space than others.

We note that \cite{10.1214/13-BA814,buchholz2020adaptivetuninghamiltonianmonte}
have suggested an adaptive SMC relying on HMC which also relies on
a population of stepsizes. However, their procedures focus on tuning
parameters of a standard HMC-MH mutation kernel, not integrator snippets,
and depart from the standard SMC framework by adaptively updating
a population of stepsizes whose distribution is unspecified; we are
not aware of any result justifying this approach.

\begin{figure}
\begin{centering}
\includegraphics[width=0.8\textwidth]{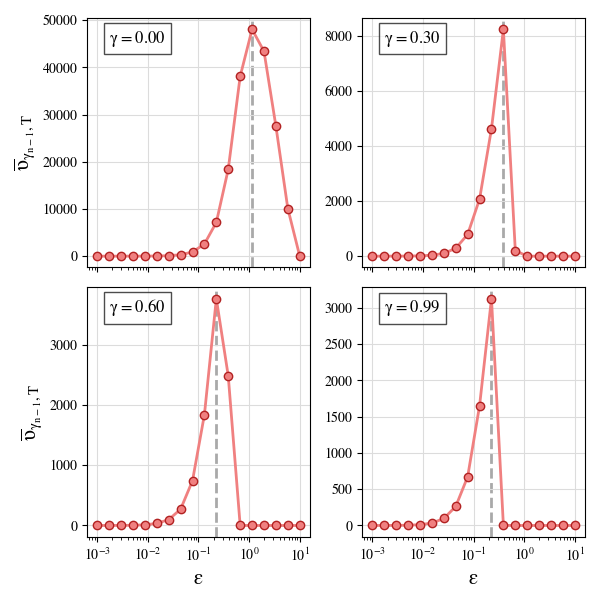}
\par\end{centering}
\caption{Estimates of $\epsilon\protect\mapsto\bar{v}_{\gamma}(\epsilon)=\protect\int v_{\gamma}(\epsilon,z)\,\bar{\mu}\big({\rm d}z\mid\epsilon;\gamma\big)$
for $\gamma=0.0,0.3,0.6,0.99$.}
\label{fig:barupsilon_function_of_epsilon}
\end{figure}

\begin{figure}
\begin{centering}
\includegraphics[width=0.8\textwidth]{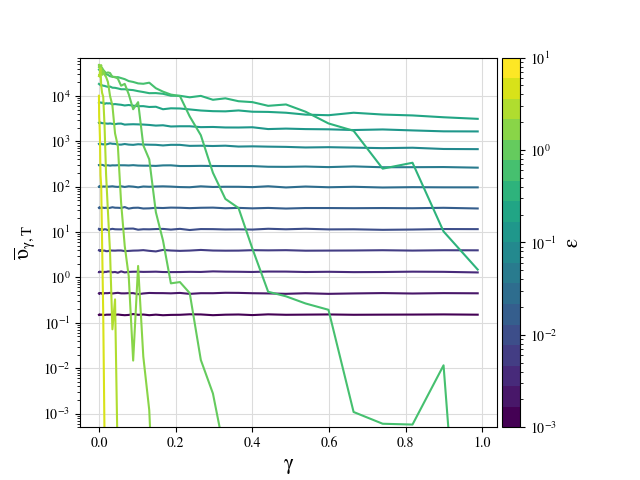}
\par\end{centering}
\caption{$\gamma\protect\mapsto\bar{v}_{\gamma}(\epsilon)$  for a range of
values of $\epsilon$.}
\label{fig:barupsilon_function_iteration}
\end{figure}

\paragraph{Optimising with Bayes' rule}

An outline of the proposed procedure to update $\nu_{\theta}$ is
provided in Alg.~\ref{alg:adapting-gamma-epsilon-outline} where
$\tilde{\mu}\big(\cdot;\gamma,\theta\big)$ is the instrumental skewed
probability distribution given by
\begin{equation}
\tilde{\mu}\big({\rm d}(\epsilon,z);\gamma,\theta\big)\propto v_{\gamma}(\epsilon,z)\nu_{\theta}({\rm d}\epsilon)\bar{\mu}\big({\rm d}z\mid\epsilon;\gamma\big)\,.\label{eq:def-mu-tilde-adaptation}
\end{equation}
In words, at iteration $n\geq1$ one projects $\tilde{\mu}(\cdot;\gamma_{n-1},\theta_{n-1})$,
a distribution designed to give more weights to the better performing
stepsizes, back onto the mother distribution $\{\bar{\mu}\big(\cdot;\gamma_{n-1},\theta\big),\theta\in\Theta\}$
by minimizing the KL divergence $\theta\mapsto{\rm KL}\big(\tilde{\mu}(\cdot;\gamma_{n-1},\theta_{n-1}),\bar{\mu}(\cdot;\gamma_{n-1},\theta)\big)$
to obtain $\theta_{n}$; again this formulation is possible thanks
to the conditional independence structure of $\bar{\mu}$. The intuition
behind this choice is as follows. Given our current best guess of
a distribution $\nu_{\theta_{n-1}}$ leading to a high average performance
function $\theta\mapsto\chi_{\theta}(\theta;\gamma_{n-1})$ one can
consider (and compute in practice) $\bar{\mu}_{n-1}(\cdot)=\bar{\mu}(\cdot;\gamma_{n-1},\theta_{n-1})$
and find a new value $\theta_{n}$ which, a posteriori, would have
led to a better average performance measure. This is achieved by adjusting
$\theta_{n}$ in order to allocate more mass to better performing
values of $\epsilon$. This corresponds to the gradient-free algorithm
recently studied in \cite{andrieu2024gradientfreeoptimizationintegration}
where convergence is established in the scenario where the fitness
function is time invariant and the variance of $\nu_{\theta}$ made
to vanish to ensure convergence to a point mass. In the present scenario
one can only hope to track a sequence of high performing values. 

We illustrate the optimisation strategy on Fig.~\ref{subfig:smoothing_v_with_nu_skewness1}
and \ref{subfig:smoothing_v_with_nu_skewness3} for $\nu_{\theta}$
an inverse Gaussian distribution of skewness $s=1$ and $s=3$ respectively
(see Appendix~\ref{app:inverse-gaussian-calculations} for details).
The blue pixelated plot is a colour-coded histogram of the number
of particles falling in cells of the $(\epsilon,\log\big(v_{\gamma_{n-1}}(\epsilon,z)\big))$
plane, superimposed with $\nu_{\theta_{n-1}}$ (yellow) and $\nu_{\theta_{n}}$
(red), obtained by fitting $\nu_{\theta}$ to $\nu_{\theta_{n-1}}$
skewed with $v_{\gamma_{n-1}}$ as in (\ref{eq:def-mu-tilde-adaptation}).
The parameter $s$ can be thought of as playing a rôle similar to
that of the magnitude of a stepsize in a learning algorithm: a higher
value can lead to faster exploration, but too high a value can also
cause instability. 

\begin{figure}
\centering
\hfill{}\subfloat[Inverse Gaussian with skewness $s=1$]{\includegraphics[width=0.45\textwidth]{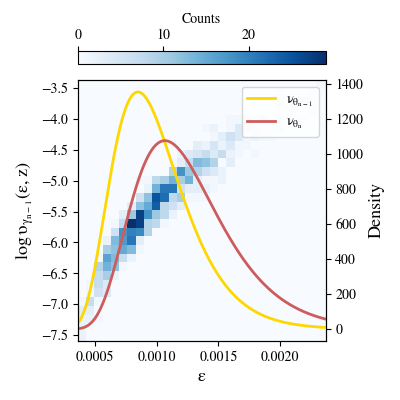}\label{subfig:smoothing_v_with_nu_skewness1}}\hfill{}\subfloat[Inverse Gaussian with skewness $s=3$]{\includegraphics[width=0.45\textwidth]{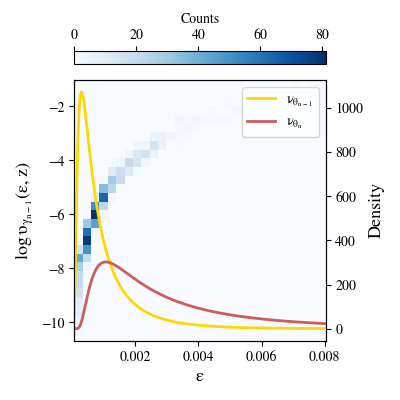}\label{subfig:smoothing_v_with_nu_skewness3}}\hfill{}

\caption{Illustration of the optimisation strategy.}
\label{fig:smoothign_v_with_nu}

\end{figure}

\LinesNumbered

\begin{algorithm}
For $n\geq0$, given the distribution $\bar{\mu}_{n}(\cdot)=\bar{\mu}(\cdot;\gamma_{n},\theta_{n})$
with $\bar{\mu}$ in (\ref{eq:target-mother-adaptive}).

Set $\gamma_{n+1}=\sup\{\gamma\colon\gamma_{n}\leq\gamma\leq1,\chi_{\gamma,n+1}(\gamma)\leq\bar{\chi}_{\gamma,n}\}\wedge1\,,$
with $\chi_{\gamma,n+1}$ in (\ref{eq:chi-gamma-only}).

Set 
\[
\theta_{n+1}={\rm \arg\min_{\theta}{\rm KL}}\big(\tilde{\mu}(\cdot;\gamma_{n},\theta_{n}),\bar{\mu}(\cdot;\gamma_{n},\theta)\big)\,.
\]

Set $\bar{\mu}_{n+1}(\cdot)=\bar{\mu}(\cdot;\gamma_{n+1},\theta_{n+1})$

\caption{Adapting $\epsilon$ and $\gamma$: sequential definition of $\{\bar{\mu}_{n}=\bar{\mu}\big({\rm d}(\epsilon,z){\rm \cdot};\gamma_{n},\theta_{n}\big),n\protect\geq0\}$ }
\label{alg:adapting-gamma-epsilon-outline}
\end{algorithm}
Implementational details are provided in Appendix~\ref{app-sec:details-imp-adapt-epsilon}.

\subsubsection{Simulations}

\paragraph{Logistic regression for the Sonar dataset}

We return to the Logistic Regression problem considered in Section~\ref{subsec:Adapting-gamma-only}.
In order to test performance of our adaptive algorithm, we set up
the following experiment. We fix $T=30$, $N=500$ and select $9$
values of stepsizes log-linearly spaced between $0.001$ and $10.0$
used to determine the initial stepsizes provided by the user to the
algorithm. For each of these values, we run $20$ independent runs
of two versions of our integrator snippet. One version, non-adaptive,
uses one of the $9$ values as unique stepsize for all particles and
keeps it fixed throughout the entire run. The second version, adaptive,
uses the same value as the initial mean of $\nu_{\theta}$, again
taken to be an inverse Gaussian with fixed skewness (see Appendix~\ref{app:inverse-gaussian-calculations})
whose parameter is then adapted throughout the algorithm; for the
inverse Gaussian this means that $\theta_{0}$ is set to those predefined
values. 

The results of this experiment are shown in Fig.~\ref{fig:boxplot_adaptive_vs_fixed_sonar_std_prop_mean}.
In the top pane we compare the box plots of the estimates of the log-normalising
constant for both algorithms, across all values of $\epsilon$ on
the x-axis. Remarkably, our adaptive algorithm is able to maintain
stable estimates over a range of initial values of $\epsilon$ that
spans five orders of magnitude. In contrast, the non-adaptive algorithm
presents precise estimates for $\epsilon=0.1,0.3$ but rapidly accumulates
bias and variance. The bottom pane displays the boxplots for the adaptive
algorithm only in order for the stability of the results to be better
appreciated. We notice that compared to Fig.~\ref{fig:sonar-normalizing-constants-boxplot},
we have accumulated a small bias but have gained remarkable robustness
with respect to our initial guess of suitable stepsize magnitudes.

\begin{figure}
\centering
\includegraphics[width=0.9\textwidth]{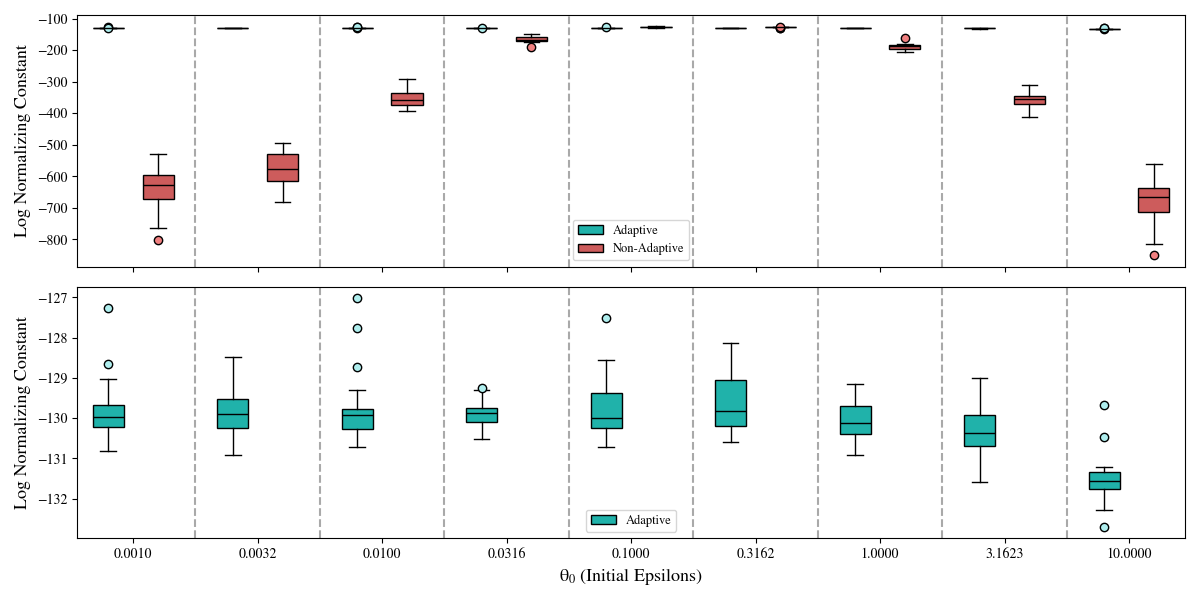}

\caption{Boxplots of the estimates of the log normalizing constant as a function
of different initial $\epsilon$ values provided by the user. Green
boxplots show the results for an adaptive integrator snippet, red
boxplots show estimates for the non-adaptive version.}
\label{fig:boxplot_adaptive_vs_fixed_sonar_std_prop_mean}
\end{figure}

Fig.~\ref{fig:final_eps_mean_sonar_adaptive_only} displays boxplots
of the final mean of $\nu_{\theta}$ for the adaptive algorithm. Despite
the wide range of initial stepsizes considered, the final means found
by the algorithm are exceptionally stable and settle around $0.175$.
Fig.~\ref{eq:def-mu-tilde-adaptation} displays boxplots representing
$\nu_{\theta}$ for different initializations of $\theta_{0}$, its
mean parameter. We emphasize here that the y-axis uses a log scale
and that our adaptive algorithm is able to recover the optimal step
size value across $9$ orders of magnitude. We note however, that
for the most extreme initialisation values, the normalising constant
estimates, whose computations relies on quantities computed at every
SMC iteration, were observed to be heavily biased and variable due
to bad performance in the early stages; these results are not reported
here.

\begin{figure}
\centering
\includegraphics[width=0.9\textwidth]{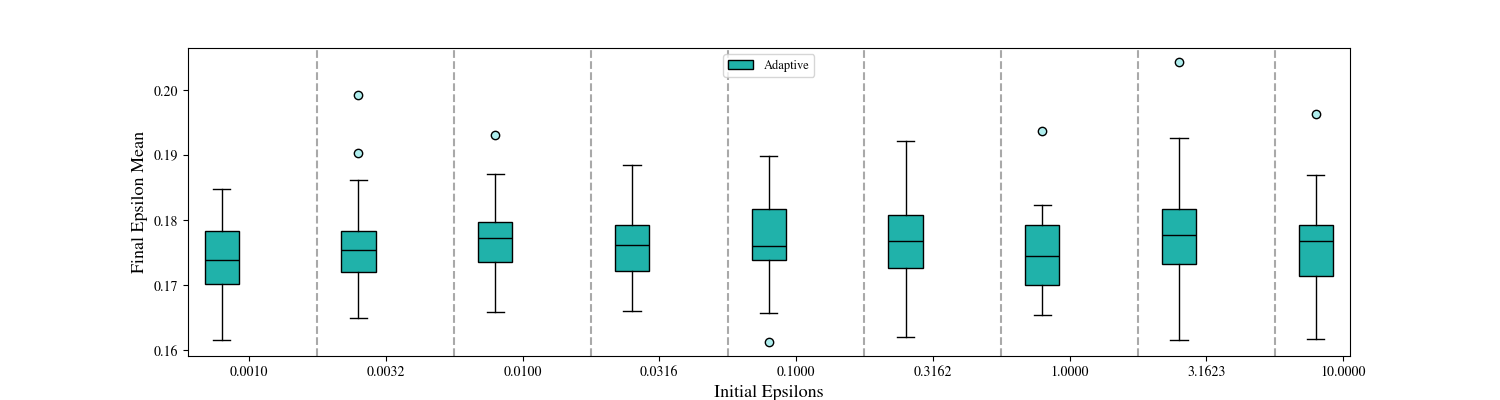}

\caption{Boxplots of the mean of the distribution $\nu_{\theta}$ on the final
iteration of the adaptive integrator snippet.}
\label{fig:final_eps_mean_sonar_adaptive_only}

\end{figure}

\begin{figure}
\centering
\includegraphics[width=0.9\textwidth]{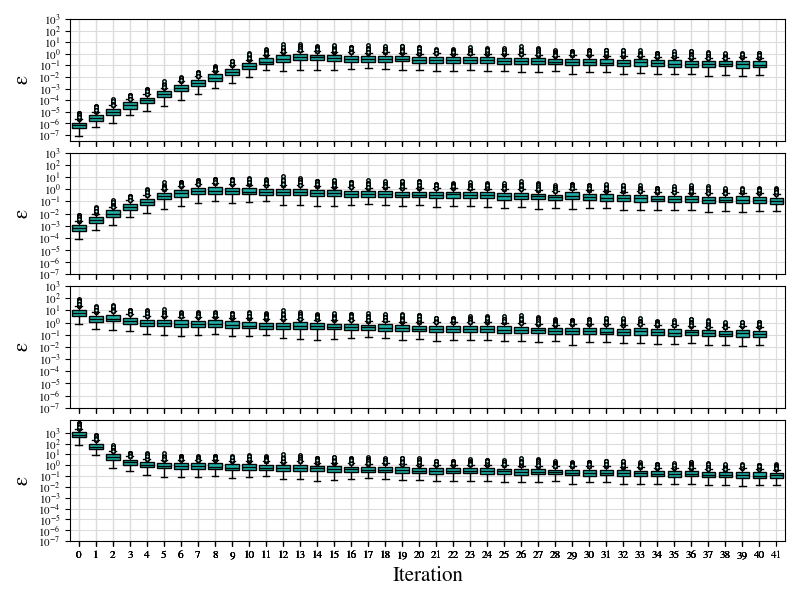}

\caption{Evolution of $\nu_{\theta}$ for different initializations of the
initial epsilon mean $\theta_{0}=10^{-6},10^{-3},10,10^{3}$}
\label{fig:distribution_epsilons_sonar}
\end{figure}

\paragraph{Log cox model for the Finnish Pines dataset}

In this paragraph we repeat the experiments above on a Log-Gaussian
Cox model for the the Finnish Pines dataset considered in \cite{girolami2011},
\cite{buchholz2020adaptivetuninghamiltonianmonte} and \cite{scaling_limits_cox_model}.
In the Leapfrog integrator, we use the following metric tensor
\begin{align*}
M_{\gamma}=\gamma\Lambda+\Sigma^{-1},\qquad\Lambda_{ii}=(1/d_{\text{grid}}^{2})\exp(\mu+\Sigma_{ii})
\end{align*}
where $\mu$ is the (scalar) mean of the Gaussian process, $\Sigma$
is the prior covariance matrix and $d_{\text{grid}}$ is the dimensionality
of our grid discretization. Fig.~\ref{fig:boxplot_adaptive_vs_fixed_logcox}
shows the estimates of the log normalizing constants for our adaptive
algorithm, using $N=500$, $T=30$ and a fixed skewness of $s=3$
(so that the standard deviation of $\nu_{\theta}$ is equal to its
mean parameter). We can see that estimates for our adaptive algorithm
are relatively stable for a wide range of initial mean parameters,
with performance decreasing only for $\theta_{0}=3.1623,10$. As for
the Sonar problem, Fig.~\ref{fig:boxplots_final_eps_mean_new_logcox}
shows that we consistently reach a value around $0.19$ for the final
mean parameter $\theta_{P}$, which is very close to the value of
$0.17$ set manually in \cite{heng2018unbiasedhamiltonianmontecarlo},
see \cite{unbiasedhmc_github_repo}. Furthermore, we briefly remark
that we have found the ESS criterion (of $\bar{\mu}_{n}$) to be a
highly unreliable metric for quantifying the performance of the algorithms,
as shown in Fig.~\ref{fig:boxplots_final_ess_logcox}, as according
to that metric $\theta_{0}=0.0032,0.01,0.0316$ would be the best
step sizes for our non-adaptive Hamiltonian Snippet, but this is in
stark contrast with Fig.~\ref{fig:boxplot_adaptive_vs_fixed_logcox}
where we see that the corresponding log normalizing constant estimates
are way off. Finally, Fig.\ref{fig:distribution_epsilons_over_iterations_1em6_1em3_10_100_new_logcox}
displays the distribution $\nu_{\theta_{n}}$ as a function $n\in\llbracket P\rrbracket$
for four values of $\theta_{0}$ spanning $10$ orders of magnitude.
All four plots present a clear progression towards the optimal value
and no oscillatory behavior, highlighting remarkable robustness of
this adaptation approach.

\begin{figure}
\centering
\includegraphics[width=0.8\textwidth]{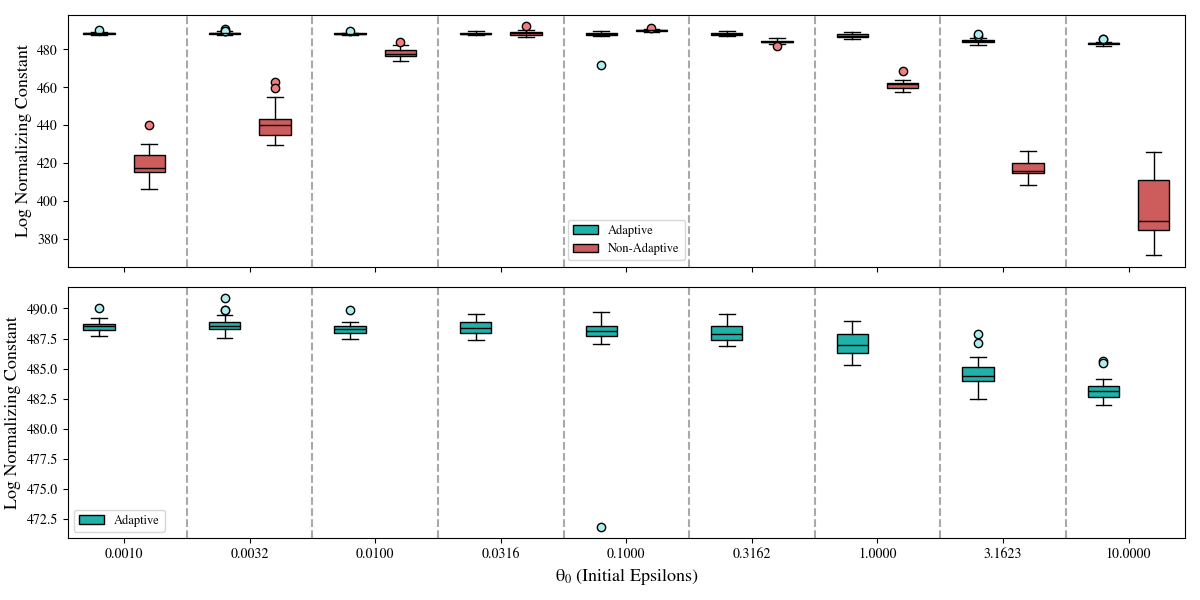}\caption{Boxplots of the estimates of the log normalizing constant for the
Log-Gaussian Cox problem as a function of different initial epsilons
provided by the user. Teal boxplots show the results for an adaptive
Hamiltonian snippet (with $\nu_{\theta}$ an Inverse Gaussian distribution
with fixed skewness $s=3$), red boxplots show estimates for the non-adaptive
version.}

\label{fig:boxplot_adaptive_vs_fixed_logcox}
\end{figure}

\begin{figure}
\centering
\includegraphics[width=0.8\textwidth]{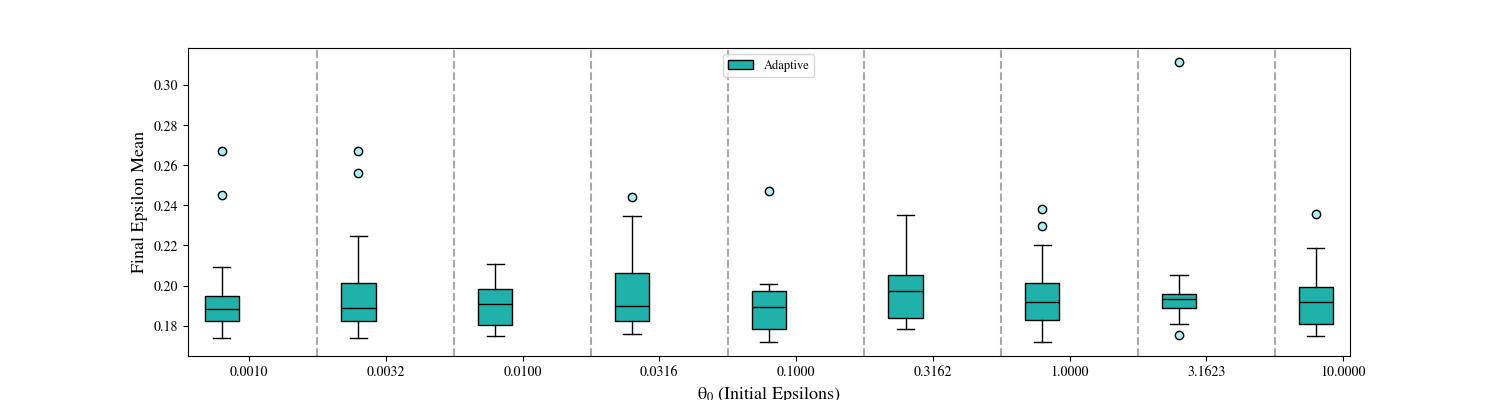}\caption{Boxplots of $\theta_{P}$, the mean parameter of the distribution
$\nu_{\theta}$ on the final iteration of the adaptive integrator
snippet, for the Log-Gaussian Cox problem.}

\label{fig:boxplots_final_eps_mean_new_logcox}
\end{figure}

\begin{figure}
\centering
\includegraphics[width=0.8\textwidth]{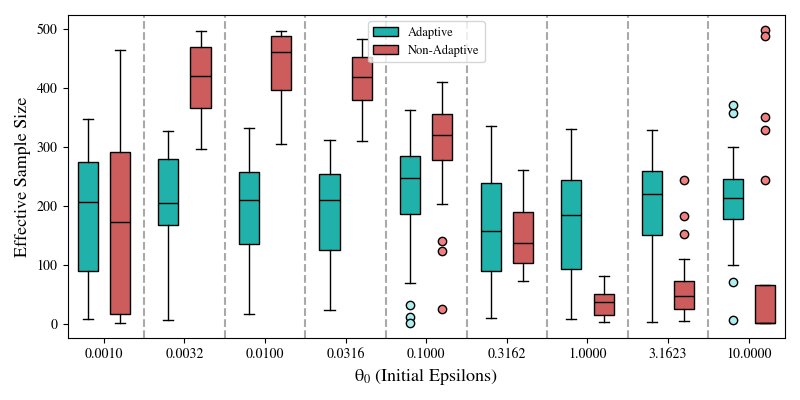}\caption{Boxplots of the final ESS for an adaptive and non-adaptive Hamiltonian
Snippet.}
\label{fig:boxplots_final_ess_logcox}
\end{figure}

\begin{figure}
\centering
\includegraphics[width=0.8\textwidth]{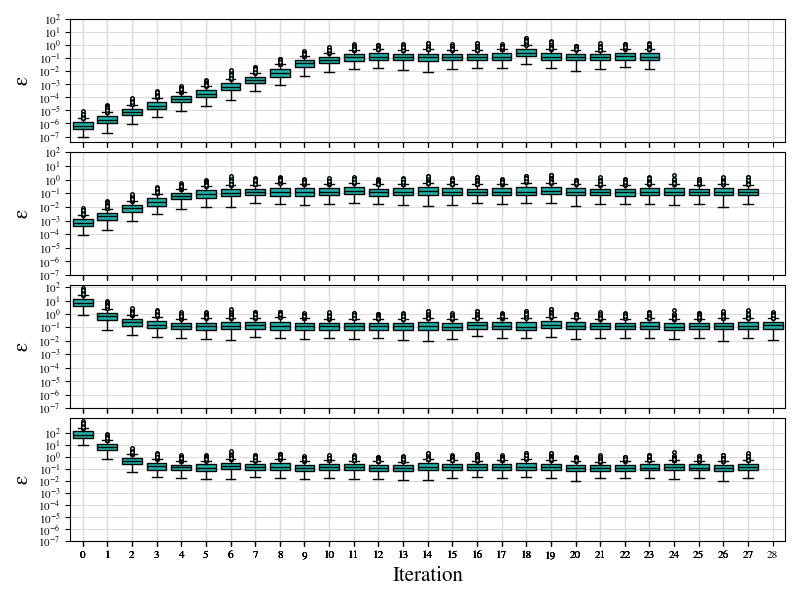}\caption{Boxplots of the estimates of the log normalizing constant for the
Log-Gaussian Cox problem as a function of different initial epsilons
provided by the user. Teal boxplots show the results for an adaptive
Hamiltonian snippet (with $\nu_{\theta}$ an Inverse Gaussian distribution
with fixed skewness $s=3$), red boxplots show estimates for the non-adaptive
version.}
\label{fig:distribution_epsilons_over_iterations_1em6_1em3_10_100_new_logcox}
\end{figure}

\subsection{Full adaptation \protect\label{sec:Adaptation-of-tau}}

Adapting $T$ offers the promise of reduced computations on an intrinsically
sequential step of the procedure, and hence not parallelisable. Conceptually
our focus is on determining an efficient effective integration time
$\tau$ associated with Hamilton's flow since it is independent of
$(\epsilon,T)$, and can therefore be estimated using any such available
pair.

\subsubsection{Methodology}

Our proposed tuning of $\epsilon$ for fixed $T$ relies on a variance
control argument. To choose $\tau$ we aim to optimise ergodicity
of the mutation kernel $\bar{M}_{n}$ in (\ref{eq:def-bar-M}), that
is its ability to forget initialisation. To that purpose our strategy
relies on a coupling argument which can be motivated by the theoretical
notion of Ricci curvature, which we review briefly. 

\paragraph{Motivation}

Consider the Wasserstein distance between probability distributions
$\upsilon_{1}$ and $\upsilon_{2}$ on $\mathsf{X}$,
\begin{equation}
W_{1}(\upsilon_{1},\upsilon_{2})=\min_{\upsilon\in\mathsf{C}(\upsilon_{1},\upsilon_{2})}\int\|y^{(1)}-y^{(2)}\|\upsilon\big({\rm d}(y^{(1)},y^{(2)})\big)\,,\label{eq:def-wasserstein}
\end{equation}
where $\mathsf{C}(\upsilon_{1},\upsilon_{2})$ is the set of couplings
of $(\upsilon_{1},\upsilon_{2})$, that is the set of distributions
$\upsilon$ of marginals $\upsilon_{1}$ and $\upsilon_{2}$. Given
a Markov kernel $Y\sim P(x,\cdot)$ its Ricci curvature is defined
as $\mathrm{Ric}(P)=1-\kappa$ where
\[
\kappa:=\sup_{x^{(1)}\neq x^{(2)}}\frac{W_{1}\big(P(x^{(1)},\cdot),P(x^{(2)},\cdot)\big)}{\|x^{(1)}-x^{(2)}\|}\,,
\]
see \cite{OLLIVIER2009810,ollivier2010survey} and also earlier ideas
\cite{Sunyach1975}, related thanks to Kantorovich-Rubinstein duality.
This can be thought of as a uniform measure of overlap between the
distributions $P(x^{(1)},\cdot)$ and $P(x^{(2)},\cdot)$ or a measure
of independence on $x^{(1)}$ and $x^{(2)}$ after one step of the
Markov chain. It is particularly well suited to studying Markov chains
defined via random mappings $Y^{(i)}=\varphi(x^{(i)},U^{(i)})$ for
random variables $U^{(i)}$, as suggested by the definition (\ref{eq:def-wasserstein}).
These ideas have been used to study HMC-MCMC \cite{chen2022optimal}
using coupling techniques; see \cite{benveniste2012adaptive} for
applications using an alternative formulation.

\paragraph{Coupling SMC particles}

Our criterion to estimate the effective integration time at SMC iteration
$n\geq1$ (see Alg.~\ref{alg:Folded-PDMP-SMC-1}) from $T_{n}$,
shared by all particles, and $\{\big(\epsilon{}_{n}^{(i)},z_{n}^{(i)}\big),i\in\llbracket N\rrbracket\}$
is motivated by the mapping contractivity ideas above and takes advantage
of the population of samples provided by SMC samplers. Full details
of the procedure are given in Alg.~\ref{alg:particle-coupling-T}
and we provide additional comments here. We describe a procedure based
on $\{\big(\epsilon{}_{n}^{(i)},z_{n}^{(i)}\big),i\in\llbracket N\rrbracket\}$
rather than $\{\big(\check{\epsilon}{}_{n}^{(i)},\check{z}_{n}^{(i)}\big),i\in\llbracket N\rrbracket\}$
as the distribution of these points turns out to be of limited importance
when interest is in a proxy for $\kappa$. We create $M\in\mathbb{N}$
pairs of distinct particles chosen at random and couple them by assigning
each pair with either of their velocities and stepsizes, say $\{(v_{j},\epsilon_{j}),j\in\llbracket M\rrbracket\}$,
as their common corresponding quantities; more sophisticated couplings
are possible but we have found this simple choice to work well for
our purpose. We then compute the local contraction coefficients $\{\kappa_{j,m},(j,m)\in\llbracket M\rrbracket\times\llbracket0,T_{n}\rrbracket\}$
as in (\ref{eq:local-contraction-coefficients}) in Alg.~\ref{alg:particle-coupling-T}
and the corresponding effective integration times $\{\tau_{j,m}=m\times\epsilon_{j},(j,m)\in\llbracket M\rrbracket\times\llbracket0,T_{n}\rrbracket\}$.
We note the following points. We here ignore the weights involved
in the definition of $\bar{M}_{n+1}$, and implicitly replace them
with uniform weights instead, for simplicity. This has nonetheless
led to satisfactory results and we have not tested a weighted version
of our work. In addition we conjecture that our choice has the additional
advantage that it further decouples the choice of $\tau$ from the
quality of the integrator, which is handled by the separate adaptation
of $\epsilon$. Naturally half of the snippets involved in the computation
of $\{\kappa_{j,m},(j,m)\in\llbracket M\rrbracket\times\llbracket0,T_{n}\rrbracket\}$
can be recycled and used in the next SMC iteration by ignoring either
of the coupled particles.

A direct approach to upperbounding $\kappa$ could consist of using,
for $k\in\llbracket0,T_{n}\rrbracket$ and $\Psi_{\epsilon,x}(z,\cdot):=\delta_{\psi_{\epsilon,x}(z)}(\cdot)$,
a Nadaraya-Watson estimator (or any nonlinear regression technique)
of 
\[
(\epsilon,k,x^{(1)},x^{(2)})\mapsto\kappa\big(\epsilon,k,x^{(1)},x^{(2)}\big):=\frac{\int\|\Psi_{\epsilon,x}^{k}(x^{(1)},v;\cdot),\Psi_{\epsilon,x}^{k}(x^{(2)},v;\cdot)\|\varpi({\rm d}v)}{\|x^{(1)}-x^{(2)}\|}\,,
\]
and aim to minimise, $(\epsilon,k)\mapsto\sup_{x^{(1)},x^{(2)}}\kappa\big(\epsilon,k,x^{(1)},x^{(2)}\big)$
on the set of available values of $k,\epsilon_{n}^{(i)}$ and $x_{n}^{(i_{j,1})},x_{n}^{(i_{j,2})}$.
This appears involved and potentially unreliable. Further, for a given
$\pi_{n}$, there is no reason for the leapfrog integrator to define
a uniformly contracting mapping for any $(k,\epsilon)$ combination,
although this is the case in the strongly convex, $L-$smooth scenario
\cite{chen2022optimal}.

\begin{algorithm}
Given $T_{n}$ and $\{\big(\epsilon_{n}^{(i)},z_{n}^{(i)}\big),i\in\llbracket N\rrbracket\}$
as in Alg.~\ref{alg:Folded-PDMP-SMC-1}.

\For{$j\in\llbracket M\rrbracket$}{

Draw $(i_{j,1},i_{j,2})\sim\mathcal{U}\big(\llbracket N\rrbracket^{2}\big)$
subject to $x_{n}^{(i_{j,1})}\neq x_{n}^{(i_{j,2})}$.

Define the instrumental particle pairs, for $v_{j}=v_{n}^{(i_{j,1})}$,
$\epsilon_{j}=\epsilon_{n}^{(i_{j,1})}$
\begin{align*}
\Big(\mathfrak{z}_{1,j} & =\big(\epsilon_{j},x_{n}^{(i_{j,1})},v_{j}\big),\mathfrak{z}_{2,j}=\big(\epsilon_{j},x_{n}^{(i_{j,2})},v_{j}\big)\Big)\,.
\end{align*}

\For{$m\in\llbracket0,T_{n}\rrbracket$}{

Compute $\tau_{j,m}=m\times\epsilon_{j}$ and the local contraction
coefficients
\begin{equation}
\kappa_{j,m}=m^{-1}\sum_{k=0}^{m}\|\psi_{n+1,x}^{k}(\mathfrak{z}_{1,j})-\psi_{n+1,x}^{k}(\mathfrak{z}_{2,j})\|/\|x_{n}^{(i_{j,1})}-x_{n}^{(i_{j,2})}\|\label{eq:local-contraction-coefficients}
\end{equation}
 } }

Define a grid of bin centres $\check{\tau}_{0}=-\check{\tau}_{1},\check{\tau}_{1}<\cdots<\check{\tau}_{b}<\check{\tau}_{b+1}=\check{\tau}_{b}>\max\{\tau_{j,m}\}$.

\For{$i\in\llbracket b\rrbracket$}{ Compute
\begin{equation}
\bar{\kappa}_{i}:=\frac{\sum_{(j,m)\in\llbracket M\rrbracket\times\llbracket0,T_{n}\rrbracket}\kappa_{j,m}\mathbf{1}\left\{ (\check{\tau}_{i-1}+\check{\tau}_{i})/2\leq\tau_{j,m}<(\check{\tau}_{i}+\check{\tau}_{i+1})/2\right\} }{\sum_{(j,m)\in\llbracket M\rrbracket\times\llbracket0,T_{n}\rrbracket}\mathbf{1}\left\{ (\check{\tau}_{i-1}+\check{\tau}_{i})/2\leq\tau_{j,m}<(\check{\tau}_{i}+\check{\tau}_{i+1})/2\right\} }\,\label{eq:bin-average-contractions}
\end{equation}

}

Set
\begin{align}
\tau_{n} & =\min\{\check{\tau}_{i}\colon i\in\arg\min\{\bar{\kappa}_{j}\colon j\in\llbracket b\rrbracket\}\,,\nonumber \\
T_{n+1} & =T_{{\rm max}}\wedge\lceil\tau_{n}/{\rm median}\big(\{\epsilon_{n}^{(j)},j\in\llbracket N\rrbracket\}\big)\rceil\,.\label{eq:tau-and-T-empirical}
\end{align}

\caption{Estimating $\tau_{n}$ and choosing $T_{n+1}$ }
\label{alg:particle-coupling-T}

\end{algorithm}

\paragraph{Binned Average of the Integration Times}

Our approach is as follows. Given the pairs $\{(\tau_{j,m},\kappa_{j,m}),(j,m)\in\llbracket M\rrbracket\times\llbracket0,T_{n}\rrbracket\}$
we would like to solve the regression problem $\tau\mapsto\kappa(\tau)$
and set 
\begin{equation}
\tau_{n}=\min\{\tau>0\colon\tau\in\arg\min\{\kappa(\tau')\colon\tau'>0\}\}\,,\label{eq:tau-n-ensemble}
\end{equation}
that is set $\tau_{n}$ to the smallest effective integration time
minimising $\kappa(\tau)$. As illustrated in Fig.~\ref{fig:contractivity_epsinit0dot001_third_iteration}-\ref{fig:contractivity_epsinit0dot1_third_iteration},
where curves $m\mapsto(\tau_{j,m},\kappa_{j,m})$ are displayed for
an example and four SMC iterations, the data is unbalanced in that
pairs with larger effective integration times $\tau_{j,m}$ are less
frequent. There is a multitude of estimation techniques for this type
of situations. Here we have focussed on a simple minded approach,
which has proved sufficiently robust for our purpose. Define a grid
of $b$ bin centres $\check{\tau}_{0}=-\check{\tau}_{1},\check{\tau}_{1}<\cdots<\check{\tau}_{b}<\check{\tau}_{b+1}=\check{\tau}_{b}>\max\{\tau_{j,m}\}$,
we then compute for $i\in\llbracket b\rrbracket$ the bin average
contraction coefficients $\bar{\kappa}_{i}$ in (\ref{eq:bin-average-contractions}).
We then set $\tau_{n}$ to the empirical version of (\ref{eq:tau-n-ensemble})
in (\ref{eq:tau-and-T-empirical}). Finally this now needs to be translated
into the number of integration steps $T_{n}$. Here the strategy consists
of choosing $T_{n+1}$ in such a way that the distance $|T_{n+1}\epsilon_{n}^{(i)}-\tau_{n}|$
is minimised in a statistical sense. A possible choice, presented
here when $N\rightarrow\infty$, is 
\[
\arg\min_{T>0}\int\left|T\epsilon-\tau_{n}\right|\nu_{\theta_{n}}\big({\rm d}\epsilon\big)\,,
\]
which is achieved for $\tau_{n}/T={\rm median}(\nu_{\theta_{n}})$,
therefore leading to a $T_{n+1}=\lceil\tau_{n}/{\rm median}(\nu_{\theta_{n}})\rceil$
for iteration $n+1$ of the SMC algorithm. The empirical version given
in (\ref{eq:tau-and-T-empirical}). In order to avoid incurring unecessary
high costs when the initial effective integration time provided by
the user is much smaller than the optimal one, we clip this value
$T_{n+1}$ to the maximum value $T_{\text{max}}$, corresponding to
the user's maximum budget.

\subsubsection{Simulations}

We ran experiments adapting both $\epsilon$ and $T$ on the Sonar
example of Section~\ref{subsec:Adapting-gamma-only}. Fig.~\ref{fig:contractivity_epsinit0dot001_third_iteration}
shows the contractivity curves as a function of the effective integration
time, as well as the binned average and the resulting minimum $\tau_{n}$
for the third iteration of a Hamiltonian Snippet initialized with
$N=500$, $T=100$, $\theta_{0}=0.001$ and skewness $3$. The left
pane demonstrates how the binned average is very effective at locating
the minimum, even when this minimum is reached only by a handful of
trajectories (see the histogram on the right hand side). Fig.~\ref{fig:contractivity_epsinit0dot001_30th_iteration}
shows the same curves at the $30^{\text{th}}$ iteration when much
more data has been assimilated through tempering and thus the target
is harder to navigate. Importantly, we see that even though some trajectories
do not contract, and even blow up, our binned average approach remains
stable and is still able to choose $\tau$ appropriately.

Due to estimation errors the procedure may mistakenly not choose the
first encountered trough, as shown in Fig.~\ref{fig:contractivity_epsinit0dot1_first_iteration}
where the initial $\theta_{0}$ was set to a large value. However,
as shown in Fig.~\ref{fig:contractivity_epsinit0dot1_second_iteration},
at the next iteration the algorithm has already identified the second
trough and by the third iteration in Fig.~\ref{fig:contractivity_epsinit0dot1_third_iteration}
has fully recovered.

Fig.~\ref{fig:boxplots_final_taus_N500_T80_skewness3_runs20} shows
the final values of $\tau$ obtained from $20$ runs of our adaptive
algorithm starting from $T_{0}=80$. Fig.~\ref{fig:evolution_of_T_log_scale}
shows the evolution of the optimal $T_{n}$ as a function of the tempering
parameter. Since the optimal stepsize is around $0.18$ we can see
that smaller stepsizes keep $T_{n}$ to its maximum value for a longer
time than larger stepsizes.

The strong regularity of the contractivity curves may seem surprising,
if not suspicious. However one can establish that the likelihood function
of a logistic regression model is $m-$strongly convex and $L-$smooth
\cite[Example 33]{10.1214/24-AAP2058}, the scenario covered by \cite{chen2022optimal};
this may be part of the explanation. We have observed similar regularity
on other models and are currently investigating this point separately,
included more localized adaptation of $T_{n}$.

\begin{figure}
\centering
\includegraphics[width=0.9\textwidth]{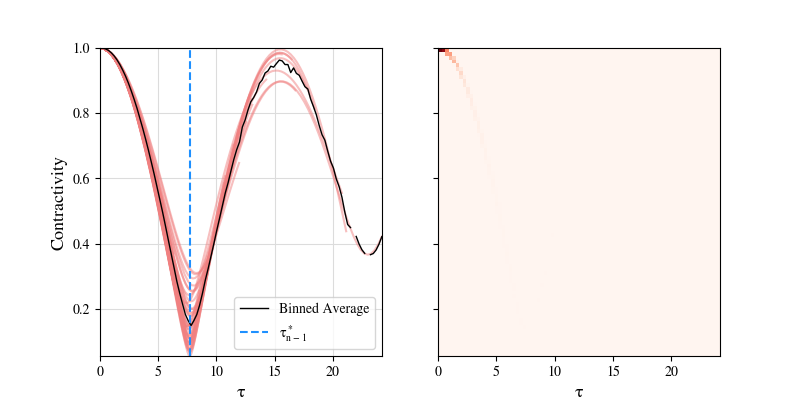}\caption{Left: Contractivity curves $m\protect\mapsto\kappa_{j,m}$ for all
$M=N/2$ pairs of coupled trajectories (red) binned averages $m\protect\mapsto\bar{\kappa}_{m}$
(black) and the minimum, $\tau_{n}$. Right: 2D histogram illustrating
the number of snippets achieving a given contraction level at time
$\tau$. SMC iteration 3, $\theta_{0}=0.001$.}

\label{fig:contractivity_epsinit0dot001_third_iteration}
\end{figure}

\begin{figure}
\centering
\includegraphics[width=0.9\textwidth]{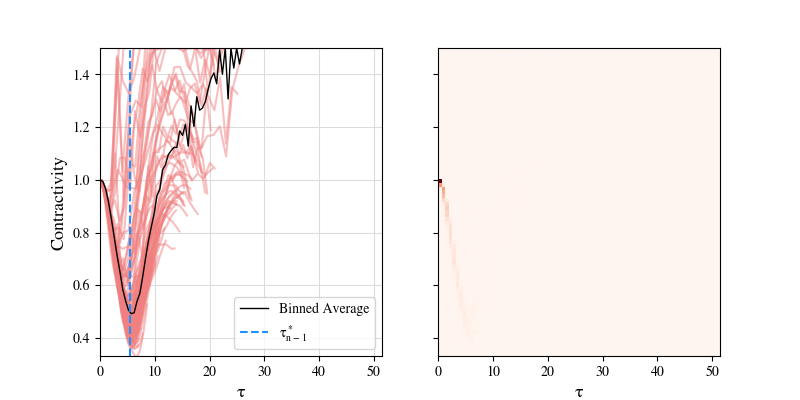}

\caption{Left: Contractivity curves $m\protect\mapsto\kappa_{j,m}$ for all
$M=N/2$ pairs of coupled trajectories (red) binned averages $m\protect\mapsto\bar{\kappa}_{m}$
(black) and the minimum, $\tau_{n}$. Right: 2D histogram illustrating
the number of snippets achieving a given contraction level at time
$\tau$. SMC iteration 30, $\theta_{0}=0.001$.}

\label{fig:contractivity_epsinit0dot001_30th_iteration}
\end{figure}

\begin{figure}
\centering
\includegraphics[width=0.9\textwidth]{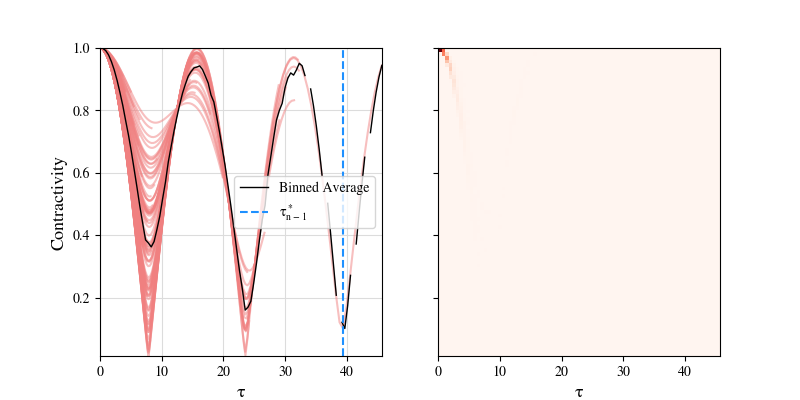}

\caption{Left: Contractivity curves $m\protect\mapsto\kappa_{j,m}$ for all
$M=N/2$ pairs of coupled trajectories (red) binned averages $m\protect\mapsto\bar{\kappa}_{m}$
(black) and the minimum, $\tau_{n}$. Right: 2D histogram illustrating
the number of snippets achieving a given contraction level at time
$\tau$. We can see that due to noise the binned average resulted
lower on the third trough whereas visually we can notice that the
first trough reaches a smaller contractivity. SMC iteration 1, $\theta_{0}=0.1$.}

\label{fig:contractivity_epsinit0dot1_first_iteration}
\end{figure}

\begin{figure}
\centering
\includegraphics[width=0.9\textwidth]{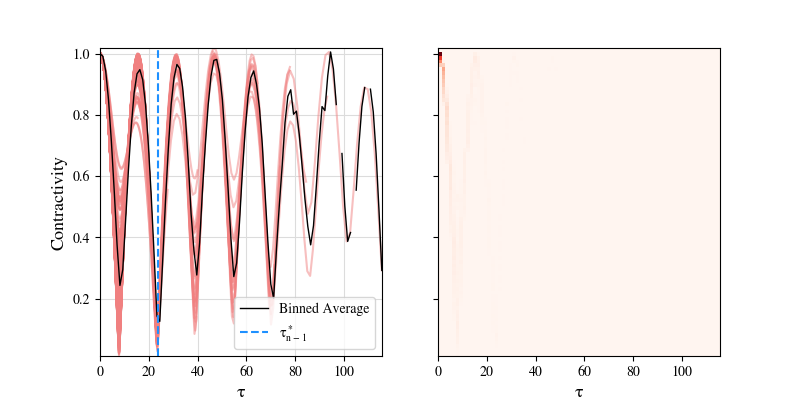}

\caption{Left: Contractivity curves $m\protect\mapsto\kappa_{j,m}$ for all
$M=N/2$ pairs of coupled trajectories (red) binned averages $m\protect\mapsto\bar{\kappa}_{m}$
(black) and the minimum, $\tau_{n}$. Right: 2D histogram illustrating
the number of snippets achieving a given contraction level at time
$\tau$. At the second iteration it has identified the second trough.
SMC iteration 2, $\theta_{0}=0.1$.}

\label{fig:contractivity_epsinit0dot1_second_iteration}
\end{figure}

\begin{figure}
\centering
\includegraphics[width=0.9\textwidth]{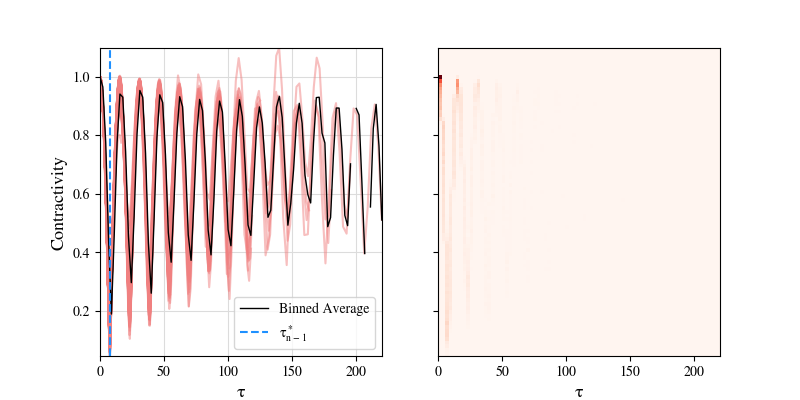}

\caption{Left: Contractivity curves $m\protect\mapsto\kappa_{j,m}$ for all
$M=N/2$ pairs of coupled trajectories (red) binned averages $m\protect\mapsto\bar{\kappa}_{m}$
(black) and the minimum, $\tau_{n}$. Right: 2D histogram illustrating
the number of snippets achieving a given contraction level at time
$\tau$. At the third iteration it has identified the first trough.
SMC iteration 3, $\theta_{0}=0.1$.}

\label{fig:contractivity_epsinit0dot1_third_iteration}
\end{figure}

\begin{figure}
\centering
\includegraphics[width=0.9\textwidth]{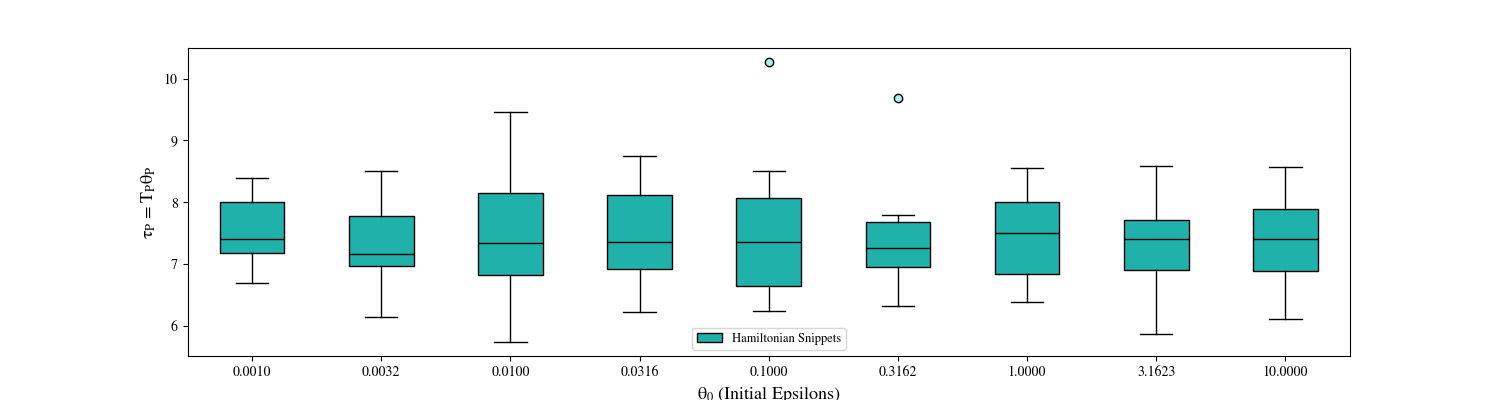}\caption{Boxplots of final integration time $\tau_{P}=T_{P}\times\theta_{P}$
over $20$ runs of our adaptive Hamiltonian Snippet, for different
initial values of $\theta_{0}$, the mean of $\nu_{\theta}$.}

\label{fig:boxplots_final_taus_N500_T80_skewness3_runs20}

\end{figure}

\begin{figure}
\centering
\includegraphics[width=0.9\textwidth]{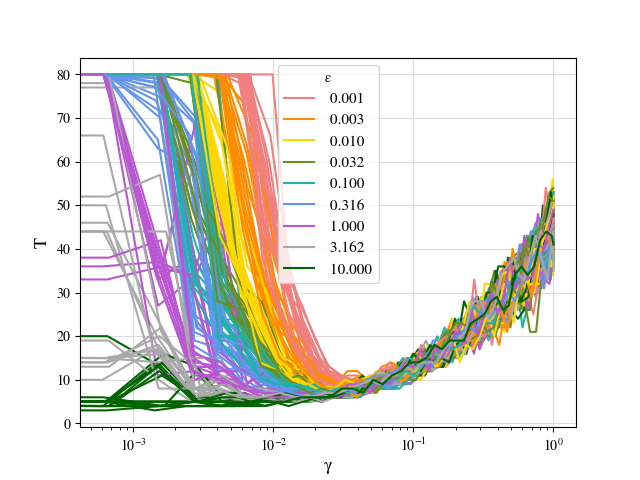}

\caption{Evolution of the number of leapfrog steps as a function of iteration,
for various different initial means of the epsilon distribution.}

\label{fig:evolution_of_T_log_scale}
\end{figure}

\section{Discussion \protect\label{sec:Discussion}}

We have shown how mappings used in various Monte Carlo schemes relying
on numerical integrators of ODE can be implemented to fully exploit
all computations to design robust and efficient sampling algorithms.
We have shown that the general framework we have identified, relying
on a mixture of distributions obtained as pushforwards of the distribution
of interest, facilitates the design of new algorithms, in particular
adaptive algorithms. Numerous questions remain open, including the
tradeoff between $N$ and $T$. A precise analysis of this question
is made particularly difficult by the fact that integration along
snippets is straightforwardly parallelizable, while resampling does
not lend itself to straightforward parallelisation.

Another point is concerned with the particular choice of mutation
Markov kernel $\bar{M}_{n}$, or $\bar{M}$, in (\ref{eq:def-bar-M})
or (\ref{eq:def-M-bar-nu}). Indeed such a kernel starts with a transition
from samples approximating the snippet distribution $\bar{\mu}_{n-1}$
to $\mu_{n-1}$, which is then followed by a reweighting of samples
leading to a representation of $\bar{\mu}_{n}$. Instead, for illustration,
one could suggest using an SMC sampler with (\ref{eq:pre-windows-of-states-kernel})
as mutation kernel. 

In relation to the discussion in Remark~\ref{rem:def-ESS-unfolded},
a natural question is how our scheme would compare with a ``Rao-Blackwellized''
SMC where weights of the type (\ref{eq:rao-blackwell-SMC}), derived
from (\ref{eq:def-Pnk-intro}) are used. 

We leave all these questions for future investigations. 

\section*{Acknowledgements}

The authors would like to thank Carl Dettmann for very useful discussions
on Boltzmann's conjecture. Research of CA and MCE supported by EPSRC
grant `CoSInES (COmputational Statistical INference for Engineering
and Security)' (EP/R034710/1), and EPSRC grant Bayes4Health, `New
Approaches to Bayesian Data Science: Tackling Challenges from the
Health Sciences' (EP/R018561/1). Research of CZ was supported by a
CSC Scholarship.

\appendix

\section{Notation and definitions \protect\label{sec:app-Notation}}

We will write $\mathbb{N}=\left\{ 0,1,2,\dots\right\} $ for the set
of natural numbers and $\mathbb{R}_{+}=\left(0,\infty\right)$ for
positive real numbers. Throughout this section $(\mathsf{E},\mathscr{E})$
is a generic measurable space. 
\begin{itemize}
\item For $A\subset\mathsf{E}$ we let $A^{\complement}$ be its complement.
\item $\mathcal{M}(\mathsf{E},\mathscr{E})$ (resp. $\mathcal{P}(\mathsf{E},\mathscr{E})$
is the set of measures (resp. probability distributions) on $(\mathsf{E},\mathscr{E})$
\item For a set $A\in\mathscr{E}$, its complement in $\mathsf{E}$ is denoted
by $A^{\complement}$. We denote the corresponding indicator function
by $\mathbf{1}_{A}:\mathsf{E}\to\left\{ 0,1\right\} $ and may use
the notation $\mathbf{1}\{z\in A\}:=\mathbf{1}_{A}(z)$.
\item For $\mu$ a probability measure on $(\mathsf{E},\mathscr{E})$ and
a measurable function $f\colon\mathsf{E}\rightarrow\mathbb{R}$ and
, we let $\mu(f):=\int f(x)\mu({\rm d}x)$.
\item For two probability measures $\mu$ and $\nu$ on $(\mathsf{E},\mathscr{E})$
we let $\mu\otimes\nu$ be a measure on $(\mathsf{E}\times\mathsf{E},\mathscr{E}\otimes\mathscr{E})$
such that $\mu\otimes\nu(A\times B)=\mu(A)\nu(B)$ for $A,B\in\mathscr{E}$. 
\item For a Markov kernel $P(x,{\rm d}y)$ on $\mathsf{E}\times\mathscr{E}$,
we write 
\begin{itemize}
\item $\mu\otimes P$ for the probability measure on $(\mathsf{E}\times\mathsf{E},\mathscr{E}\otimes\mathscr{E})$
such that for $\bar{A}\in\mathscr{E}\otimes\mathscr{E}$, the minimal
product $\sigma$-algebra, $\mu\otimes P(\bar{A})=\int_{\bar{A}}\mu({\rm d}x)P(x,{\rm d}y)$. 
\item $\mu\otimeswapped P$ for the probability measure on $(\mathsf{E}\times\mathsf{E},\mathscr{E}\otimes\mathscr{E})$
such that for $A,B\in\mathscr{E}$ $\mu\otimeswapped P(A\times B)=\mu\otimes P(B\times A)$.
\end{itemize}
\item For $\mu,\nu$ probability distributions on $(\mathsf{E},\mathscr{E})$
and kernels $M,L\colon\mathsf{E}\times\mathscr{E\rightarrow}[0,1]$
such that $\mu\otimes M\gg\nu\otimeswapped L$ then we denote
\[
\frac{{\rm d}\nu\otimeswapped L}{{\rm d}\mu\otimes M}(z,z')
\]
 the corresponding Radon-Nikodym derivative such that for $f\colon\mathsf{E}\times\mathsf{E}\rightarrow\mathbb{R}$,
\[
\int f(z,z')\frac{{\rm d}\nu\otimeswapped L}{{\rm d}\mu\otimes M}(z,z')\mu\otimes M\big({\rm d}(z,z')\big)=\int f(z,z')\nu\otimeswapped L\big({\rm d}(z',z)\big)\,.
\]
\item A point mass distribution at $x$ will be denoted by $\delta_{x}({\rm d}y)$;
it is such that for $f\colon\mathsf{E}\rightarrow\mathbb{R}$ 
\[
\int f(x)\delta_{x}({\rm d}y)=f(x)
\]
 
\item In order to alleviate notation, for $M\in\mathbb{N}$, $(z^{(i)},w_{i})\in\mathsf{E}\times[0,\infty)$,
$i\in\llbracket M\rrbracket$, we refer to $\big\{(z^{(i)},w_{i}),i\in\llbracket M\rrbracket\big\}$
as weighted samples to mean $\big\{(z^{(i)},\tilde{w}_{i}),i\in\llbracket M\rrbracket\big\}$
where $\tilde{w}_{i}\propto w_{i}$ but $\sum_{i=1}^{M}\tilde{w}_{i}=1$.
\item We say that a set of weighted samples, or particles, $\{(z_{i},w_{i})\in\mathsf{Z}\times\mathbb{R}_{+}\colon i\in\llbracket N\rrbracket\}$
for $N\geq1$ represents a distribution $\mu$ whenever for $f$ $\mu$-integrable
\[
\sum_{i=1}^{N}\frac{w_{i}}{\sum_{j=1}^{N}w_{j}}f(z_{i})\approx\mu(f)\,,
\]
in either in the $L^{p}$ sense for some $p\geq1$.
\item For $M\in\mathbb{N}$, $w_{i}\in[0,\infty)$, $i\in\llbracket M\rrbracket$,
we let $K\sim{\rm Cat}\left(w_{1},w_{2},\ldots,w_{M}\right)$ mean
that $\mathbb{P}(K=k)\propto w_{k}$.
\item For $M,N\in\mathbb{N}$, $w_{ij}\in[0,\infty)$, $i\in\llbracket M\rrbracket\times\llbracket N\rrbracket$,
we let $K\sim{\rm Cat}\left(w_{ij},i\in\llbracket M\rrbracket\times\llbracket N\rrbracket\right)$
mean that $\mathbb{P}(K=(k,l))\propto w_{kl}$.
\item for $f\colon\mathbb{R}^{m}\rightarrow\mathbb{R}^{n}$ we let 
\begin{itemize}
\item $\nabla\otimes f$ be the transpose of the Jacobian
\item for $n=1$ we let $\nabla f=(\nabla\otimes f)^{\top}$ be the gradient,
\item $\nabla\cdot f$ be the divergence.
\end{itemize}
\end{itemize}

\section{Background on Radon-Nikodym derivative \protect\label{sec:One-measure-theoretic}}

The general formalism required and used throughout the paper relies
on a unique measure theoretic tool, the Radon-Nikodym derivative.
We gather here definitions and intermediate results used throughout,
pointing out the simplicity of the tools involved and and the benefits
they bring.
\begin{defn}[Pushforward]
\label{def:pushforward} Let $\mu$ be a measure on $(\mathsf{Z},\mathscr{Z})$
and $\psi:(\mathsf{Z},\mathscr{Z})\to(\mathsf{Z}',\mathscr{F}')$
a measurable function. The pushforward of $\mu$ by $\psi$ is defined
by 
\[
\mu^{\psi}(A):=\mu(\psi^{-1}(A)),\qquad A\in\mathscr{Z}',
\]
where $\psi^{-1}(A)=\{z\in\mathsf{Z}:\psi(z)\in A\}$ is the preimage
of $A$ under $\psi$.
\end{defn}

If $\mu$ is a probability distribution then $\mu^{\psi}$ is the
probability measure associated with $\psi(Z)$ when $Z\sim\mu$.
\begin{defn}[Dominating and equivalent measures]
 For two measures $\mu$ and $\nu$ on the same measurable space
$(\mathsf{Z},\mathscr{Z})$, 
\begin{enumerate}
\item $\nu$ is said to dominate $\mu$ if for all measurable $A\in\mathscr{Z}$,
$\mu(A)>0\Rightarrow\nu(A)>0$ -- this is denoted $\nu\gg\mu$.
\item $\mu$ and $\nu$ are equivalent, written $\mu\equiv\nu$, if $\mu\gg\nu$
and $\nu\gg\mu$.
\end{enumerate}
\end{defn}

We will need the notion of Radon-Nikodym derivative \cite[Theorems 32.2 \& 16.11]{billingsley1995probability}:
\begin{thm}[Radon--Nikodym]
\label{thm:Radon-Nikodym}Let $\mu$ and $\nu$ be $\sigma$-finite
measures on $(\mathsf{Z},\mathscr{Z})$. Then $\nu\ll\mu$ if and
only if there exists an essentially unique, measurable, non-negative
function $f\colon\mathsf{Z}\rightarrow[0,\infty)$ such that
\[
\int_{A}f(z)\mu({\rm d}z)=\nu(A),\qquad A\in\mathscr{E}.
\]
Therefore we can view ${\rm d}\nu/{\rm d}\mu:=f$ as the density of
$\nu$ w.r.t $\mu$ and in particular if $g$ is integrable w.r.t.
$\nu$ then
\[
\int g(z)\frac{{\rm d}\nu}{{\rm d}\mu}(z)\mu({\rm d}z)=\int g(z)\nu({\rm d}z)\,.
\]
\end{thm}

If $\mu$ is a measure and $f$ a non-negative, measurable function
then $\mu\cdot f$ is the measure $(\mu\cdot f)(A)=\int{\bf 1}_{A}(z)f(z)\mu({\rm d}z)$,
i.e. the measure $\nu=\mu\cdot f$ such that $f$ is the Radon--Nikodym
derivative ${\rm d}\nu/{\rm d}\mu=f$.

The following establishes the expression of an expectation with respect
to the pushforward $\mu^{\psi}$ in terms of expectations with respect
to $\mu$ \cite[Theorem 16.13]{billingsley1995probability}.
\begin{thm}[Change of variables]
\label{thm:change-of-variables} A function $f:\mathsf{Z}'\to\mathbb{R}$
is integrable w.r.t. $\mu^{\psi}$ if and only if $f\circ\psi$ is
integrable w.r.t. $\mu$, in which case
\begin{equation}
\int_{\mathsf{Z}'}f(z)\mu^{\psi}({\rm d}z)=\int_{\mathsf{Z}}f\circ\psi(z)\mu({\rm d}z)\,.\label{eq:change-of-variable}
\end{equation}
\end{thm}

We now establish results useful throughout the manuscript. The central
identity used throughout the manuscript is a direct application of
Theorem~\ref{thm:change-of-variables} for $\psi\colon\mathsf{Z}\rightarrow\mathsf{Z}$
invertible
\[
\int_{\mathsf{Z}'}f\circ\psi(z)\mu^{\psi^{-1}}({\rm d}z)=\int_{\mathsf{Z}}f(z)\mu({\rm d}z)
\]
which seems tautological since it can be summarized as follows: for
$Z\sim\mu$, then $\psi^{-1}(Z)\sim\mu^{\psi^{-1}}$ and $\psi\circ\psi^{-1}(Z)\sim\mu$!
However the interest of the approach stems from the following properties. 
\begin{lem}
\label{thm:measure-theoretic-transform}Let $\psi\colon\mathsf{Z}\rightarrow\mathsf{Z}$
be measurable and integrable, $\mu$ and $\upsilon$ be $\sigma$-finite
measures on $(\mathsf{Z},\mathscr{Z})$ such that $\upsilon\gg\mu$
and $\upsilon\gg\upsilon^{\psi^{-1}}$. Then
\begin{enumerate}
\item $\upsilon^{\psi^{-1}}\gg\mu^{\psi^{-1}}$ and therefore $\upsilon\gg\mu^{\psi^{-1}}$,
\item for $\upsilon$-almost all $z\in\mathsf{Z}$,
\[
\frac{{\rm d}\mu^{\psi^{-1}}}{{\rm d}\upsilon}(z)=\frac{{\rm d}\mu}{{\rm d}\upsilon}\circ\psi(z)\frac{{\rm d}\upsilon^{\psi^{-1}}}{{\rm d}\upsilon}
\]
\item we have
\[
\mu\gg\mu^{\psi^{-1}}\iff\upsilon\big(\big\{ z\in\mathsf{Z}\colon{\rm d}\mu^{\psi^{-1}}/{\rm d}\upsilon(z)>0,{\rm d}\mu/{\rm d}\upsilon(z)=0\big\}\big)=0
\]
in which case for $\upsilon$-almost all $z\in\mathsf{Z}$
\[
\frac{{\rm d}\mu^{\psi^{-1}}}{{\rm d}\mu}(z)=\begin{cases}
\frac{\mu\circ\psi}{\mu}(z) & \mu(z)>0\\
0 & \text{otherwise}
\end{cases}
\]
and therefore
\[
\int_{\mathsf{Z}'}f\circ\psi(z)\mu^{\psi^{-1}}({\rm d}z)=\int_{\mathsf{Z}'}f\circ\psi(z)\frac{\mu\circ\psi(z)}{\mu(z)}\frac{{\rm d}\upsilon^{\psi}}{{\rm d}\upsilon}\mu({\rm d}z).
\]
\end{enumerate}
\end{lem}

\begin{proof}
For the first part of the first statement, let $A\in\mathscr{Z}$
such that $\upsilon^{\psi^{-1}}(A)=\upsilon\big(\psi(A)\big)>0$,
then since $\psi(A)\in\mathscr{Z}$ and $\upsilon\gg\mu$ we deduce
$\mu\big(\psi(A)\big)=\mu^{\psi^{-1}}(A)>0$ and we conclude; the
second parts follows from $\upsilon\gg\upsilon^{\psi^{-1}}$. For
the second statement for $f\colon\mathsf{Z}\rightarrow\mathbb{R}$
bounded and measurable, 
\begin{align*}
\int f(z)\frac{{\rm d}\mu^{\psi^{-1}}}{{\rm d}\upsilon}(z)\upsilon({\rm d}z) & =\int f(z)\mu^{\psi^{-1}}({\rm d}z)\\
 & =\int f\circ\psi^{-1}(z)\mu({\rm d}z)\\
 & =\int f\circ\psi^{-1}(z)\frac{{\rm d}\mu}{{\rm d}\upsilon}(z)\upsilon({\rm d}z)\\
 & =\int f(z)\frac{{\rm d}\mu}{{\rm d}\upsilon}\circ\psi(z)\upsilon^{\psi^{-1}}({\rm d}z)\\
 & =\int f(z)\frac{{\rm d}\mu}{{\rm d}\upsilon}\circ\psi(z)\frac{{\rm d}\upsilon^{\psi^{-1}}}{{\rm d}\upsilon}(z)\upsilon({\rm d}z)
\end{align*}
The third statement is given as \cite[Problem 32.6.]{billingsley1995probability},
which we solve in Lemma~\ref{lem:Billingsley-problem324}.
\end{proof}
\begin{cor}
In the scenario when $\psi\colon\mathsf{Z}\rightarrow\mathsf{Z}$
and $\upsilon$ are such that $\upsilon^{\psi}=\upsilon$ then ${\rm d}\upsilon^{\psi}/{\rm d}\upsilon\equiv1$.
\end{cor}

\begin{lem}[Billingsley, Problem 32.6.]
 \label{lem:Billingsley-problem324} Assume $\mu,\nu$ and $\upsilon$
are $\sigma$- finite and that $\upsilon\gg\nu,\mu$. Then $\mu\gg\nu$
if and only if $\upsilon\big(\big\{ z\in\mathsf{Z}\colon{\rm d}\nu/{\rm d}\upsilon(z)>0,{\rm d}\mu/{\rm d}\upsilon(z)=0\big\}\big)=0$,
in which case
\[
\frac{{\rm d}\nu}{{\rm d}\mu}(z)=\mathbf{1}\{z\in\mathsf{Z}:{\rm d}\mu/{\rm d}\upsilon(z)>0\}\frac{{\rm d}\nu}{{\rm d}\upsilon}/\frac{{\rm d}\mu}{{\rm d}\upsilon}(z)\,.
\]
\end{lem}

\begin{proof}
Let $S:=\big\{ z\in\mathsf{Z}\colon{\rm d}\nu/{\rm d}\upsilon(z)>0,{\rm d}\mu/{\rm d}\upsilon(z)=0\big\}$.
For $f\colon\mathsf{Z}\rightarrow\mathbb{R}$ integrable we always
have
\begin{align*}
\int f(z)\nu({\rm d}z) & =\int\mathbf{1}\{z\in S\}f(z)\frac{{\rm d}\nu}{{\rm d}\upsilon}(z)\upsilon({\rm d}z)+\int\mathbf{1}\big\{ z\in S^{\complement}\big\} f(z)\frac{{\rm d}\nu}{{\rm d}\upsilon}/\frac{{\rm d}\mu}{{\rm d}\upsilon}(z)\mu({\rm d}z)\,.
\end{align*}
Assume $\mu\gg\nu$ then from above for any $f\colon\mathsf{Z}\rightarrow\mathbb{R}$
integrable
\begin{align*}
\int\mathbf{1}\{z\in S\}f(z)\frac{{\rm d}\nu}{{\rm d}\mu}(z)\mu({\rm d}z)= & \int\mathbf{1}\{z\in S\}f(z)\frac{{\rm d}\nu}{{\rm d}\mu}(z)\frac{{\rm d}\mu}{{\rm d}\upsilon}(z)\upsilon({\rm d}z)\\
= & 0\,,
\end{align*}
and therefore $\upsilon\big(S\big)=0$ and we conclude from the first
identity above. Now assume that $\upsilon\big(S\big)=0$, then
\begin{align*}
\int f(z)\nu({\rm d}z) & =\int\mathbf{1}\{z\in S\}f(z)\frac{{\rm d}\nu}{{\rm d}\upsilon}(z)\upsilon({\rm d}z)+\int\mathbf{1}\big\{ z\in S^{\complement}\big\} f(z)\frac{{\rm d}\nu}{{\rm d}\upsilon}/\frac{{\rm d}\mu}{{\rm d}\upsilon}(z)\mu({\rm d}z)\\
 & =\int\mathbf{1}\big\{ z\in S^{\complement}\big\} f(z)\frac{{\rm d}\nu}{{\rm d}\upsilon}/\frac{{\rm d}\mu}{{\rm d}\upsilon}(z)\mu({\rm d}z)\,.
\end{align*}
The equivalence is therefore established and when either conditions
is satisfied we have
\[
\frac{{\rm d}\nu}{{\rm d}\mu}(z)=\mathbf{1}\big\{ z\in S^{\complement}\big\}\frac{{\rm d}\nu}{{\rm d}\upsilon}/\frac{{\rm d}\mu}{{\rm d}\upsilon}(z)
\]
and we conclude.
\end{proof}

\section{Proofs for Section~\ref{sec:A-simple-example:} \protect\label{app:Proofs-for-Section}}

The results of the following lemma can be deduced from Lemma~\ref{lem:barM-properties}
but we provide more direct arguments for the present scenario.
\begin{lem}
\label{lem:bar-weight-derivation} Assume $\mu_{n-1}\gg\bar{\mu}_{n}$,
$\mu_{n-1}R_{n}=\mu_{n-1}$ and let $\bar{M}_{n}\colon\mathsf{Z}\times\mathscr{Z}\rightarrow[0,1]$
be as in (\ref{eq:def-bar-M}). Then for $n\in\llbracket P\rrbracket$,
\begin{enumerate}
\item $\bar{\mu}_{n-1}\bar{M}_{n}=\mu_{n-1}$,
\item the near optimal kernel $\bar{L}_{n-1}\colon\mathsf{Z}\times\mathscr{Z}\rightarrow[0,1]$is
given for any $(z,A)\in\mathsf{Z}\times\mathscr{Z}$ by
\begin{equation}
\bar{L}_{n-1}(z,A):=\frac{{\rm d}\bar{\mu}_{n-1}\otimes\bar{M}_{n}(A\times\cdot)}{{\rm d}\bar{\mu}_{n-1}\bar{M}_{n}}(z)\,,\label{eq:backward-bar-L}
\end{equation}
is well defined $\bar{\mu}_{n-1}\bar{M}_{n}$-a.s.  
\item $\bar{L}_{n-1}$ in (\ref{eq:backward-bar-L}) is such that \textup{for
any} $A,B\in\mathscr{Z}$
\[
\bar{\mu}_{n-1}\otimes\bar{M}_{n}(A\times B)=\int_{B}\bar{\mu}_{n-1}\bar{M}_{n}({\rm d}z)\bar{L}_{n-1}(z,A)\,,
\]
\item the importance weight in (\ref{eq:weight-folded-HMC}) is well defined
and
\[
\bar{w}_{n}(z,z'):=\frac{{\rm d}\bar{\mu}_{n}}{{\rm d}\mu_{n-1}}(z')\,.
\]
\end{enumerate}
\end{lem}

\begin{proof}[Proof of Lemma~\ref{lem:bar-weight-derivation}]
The first statement, follows from the definition in (\ref{eq:def-bar-M})
of $\bar{M}_{n}$ and the identity (\ref{eq:unfolding-HMC-setup}).
The second statement is a consequence of the following classical arguments,
justifying Bayes' rule. For $A\in\mathscr{Z}$ fixed, consider the
measure
\[
\mathscr{Z}\ni B\mapsto\bar{\mu}_{n-1}\otimes\bar{M}_{n}\big(A\times B\big)\leq\bar{\mu}_{n-1}\otimes\bar{M}_{n}\big(\mathsf{Z}\times B\big)=\bar{\mu}_{n-1}\bar{M}_{n}(B)\,,
\]
implying $\bar{\mu}_{n-1}\bar{M}_{n}\gg\bar{\mu}_{n-1}\otimes\bar{M}_{n}\big(A\times\cdot)$
from which we deduce the existence of a Radon-Nikodym derivative such
that
\begin{align*}
\bar{\mu}_{n-1}\otimes\bar{M}_{n}\big(A\times B\big) & =\int_{B}\frac{{\rm d}\bar{\mu}_{n-1}\otimes\bar{M}_{n}\big(A\times\cdot\big)}{{\rm d}\bar{\mu}_{n-1}\bar{M}_{n}}(z')\,\bar{\mu}_{n-1}\bar{M}_{n}({\rm d}z')\,.
\end{align*}
For $(z,A)\in\mathsf{Z}\times\mathscr{Z}$, we let
\[
\bar{L}_{n-1}(z,A):=\frac{{\rm d}\bar{\mu}_{n-1}\otimes\bar{M}_{n}\big(A\times\cdot\big)}{{\rm d}\bar{\mu}_{n-1}\bar{M}_{n}}(z)\,,
\]
and note that almost surely $\bar{L}_{n-1}(z,A)\in[0,1]$ and $\bar{L}_{n-1}(z,\mathsf{Z})=1$.
For the third statement note that from the second statement, for $A,B\in\mathscr{Z}$
we have $\bar{\mu}_{n-1}\otimes\bar{M}_{n}\big(A\times B\big)=\bar{\mu}_{n-1}\bar{M}_{n}\otimeswapped\bar{L}_{n-1}\big(A\times B\big)$
and from Fubini's and the $\pi-\lambda$ theorem \cite[Theorems 3.1 and 3.2]{billingsley1995probability}
$\bar{\mu}_{n-1}\otimes\bar{M}_{n}$ and $\bar{\mu}_{n-1}\bar{M}_{n}\otimeswapped\bar{L}_{n-1}$
are probability distributions coinciding on $\mathscr{Z}\otimes\mathscr{Z}$.
To conclude proof of the third statement, for $f\colon\mathsf{Z}^{2}\rightarrow[0,1]$
measurable, we successively apply the definition of the Radon-Nikodym
derivative, Fubini's theorem, use that $\mu_{n-1}\gg\bar{\mu}_{n}$,
the first statement of this lemma, Fubini again, the second statement
\begin{align*}
\int f(z,z')\frac{{\rm d}\bar{\mu}_{n}}{{\rm d}\mu_{n-1}}(z,z')\bar{\mu}_{n-1}\otimes\bar{M}_{n}\big({\rm d}(z,z')\big) & =\int f(z,z')\frac{{\rm d}\bar{\mu}_{n}}{{\rm d}\mu_{n-1}}(z')\,\bar{\mu}_{n-1}\bar{M}_{n}\otimeswapped\bar{L}_{n}\big({\rm d}(z,z')\big)\\
 & =\int f(z,z')\frac{{\rm d}\bar{\mu}_{n}}{{\rm d}\mu_{n-1}}(z')\,\mu_{n-1}({\rm d}z')\bar{L}_{n}\big(z',{\rm d}z\big)\\
 & =\int f(z,z')\bar{\mu}_{n}({\rm d}z')\bar{L}_{n}\big(z',{\rm d}z\big)
\end{align*}
therefore establishing that $\bar{\mu}_{n-1}\otimes\bar{M}_{n}$-almost
surely
\[
\frac{{\rm d}\bar{\mu}_{n-1}\bar{M}_{n}\otimeswapped\bar{L}_{n-1}}{{\rm d}\bar{\mu}_{n-1}\otimes\bar{M}_{n}}(z,z')=\frac{{\rm d}\bar{\mu}_{n}}{{\rm d}\mu_{n-1}}(z,z')\,.
\]

\end{proof}
\begin{proof}[Proof of Proposition~\ref{prop:unfolded-folded-equivalent}]
 In Alg.~\ref{alg:Folded-PDMP-SMC-1} for any $A\in\mathscr{Z}^{\otimes N}$
we have with $(b_{1},\ldots,b_{N})$ the random vector taking values
in $\llbracket N\rrbracket^{N}$ involved in the resampling step,
\begin{align*}
\mathbb{P}_{3}\big(z_{n}^{(1:N)}\in A\mid z_{n-1}^{(1:N)}\big)= & \mathbb{E}_{3}\left\{ \mathbf{1}_{A}\{(\check{x}_{n,a_{i}}^{(i)},v_{n}^{(i)}),i\in\llbracket N\rrbracket\}\mathbf{1}\{\check{z}_{n}^{(i)}=z_{n-1}^{(b_{i})},i\in\llbracket N\rrbracket\}\mid z_{n-1}^{(j)},j\in\llbracket N\rrbracket\right\} \\
= & \mathbb{E}_{3}\left\{ \mathbf{1}_{A}\{(x_{n-1,a_{i}}^{(b_{i})},v_{n}^{(i)}),i\in\llbracket N\rrbracket\}\mid z_{n-1}^{(j)},j\in\llbracket N\rrbracket\right\} \,.
\end{align*}
Letting for $(\alpha,\beta)\in\llbracket0,T\rrbracket^{N}\times\llbracket N\rrbracket^{N}$
\[
E_{A}(\alpha,\beta)=E_{A}(\alpha,\beta,z_{n-1}^{(1:N)}):=\int\mathbf{1}_{A}\{(x_{n-1,\alpha_{i}}^{(\beta_{i})},v_{n}^{(i)}),i\in\llbracket N\rrbracket\}\varpi_{n}^{\otimes N}({\rm d}v_{n}^{(1:N)})
\]
from the tower property we have
\[
\mathbb{E}_{3}\left\{ \mathbf{1}_{A}\{(x_{n-1,a_{i}}^{(b_{i})},v_{n}^{(i)}),i\in\llbracket N\rrbracket\}\mid z_{n-1}^{(j)},b_{j}=\beta_{j},j\in\llbracket N\rrbracket\right\} =\sum_{\alpha\in\llbracket0,T\rrbracket^{N}}E_{A}(\alpha,\beta)\prod_{i=1}^{N}\frac{1}{T+1}\frac{{\rm d}\mu_{n,\alpha_{i}}}{{\rm d}\bar{\mu}_{n}}\big(z_{n-1}^{(\beta_{i})}\big).
\]
Since
\[
\mathbb{P}_{3}\left(b_{1}=\beta_{1},\ldots,b_{N}=\beta_{N}\mid z_{n-1}^{(i)},i\in\llbracket N\rrbracket\right)=\biggl(\sum_{j=1}^{N}\frac{{\rm d}\bar{\mu}_{n}}{{\rm d}\mu_{n-1}}\big(z_{n-1}^{(j)}\big)\biggr)^{-N}\prod_{i=1}^{N}\frac{{\rm d}\bar{\mu}_{n}}{{\rm d}\mu_{n-1}}\big(z_{n-1}^{(\beta_{i})}\big)
\]
we deduce
\begin{align*}
\mathbb{P}_{3}\big(z_{n}^{(1:N)}\in A\mid z_{n-1}^{(1:N)}\big)\propto & \sum_{\alpha,\beta}E_{A}(\alpha,\beta,z_{n-1}^{(1:N)})\prod_{i=1}^{N}\frac{{\rm d}\bar{\mu}_{n}}{{\rm d}\mu_{n-1}}\big(z_{n-1}^{(\beta_{i})}\big)\prod_{i=1}^{N}\frac{{\rm d}\mu_{n,\alpha_{i}}}{{\rm d}\bar{\mu}_{n}}\big(z_{n-1}^{(\beta_{i})}\big)\\
= & \sum_{\alpha,\beta}E_{A}(\alpha,\beta,z_{n-1}^{(1:N)})\prod_{i=1}^{N}\frac{{\rm d}\bar{\mu}_{n}}{{\rm d}\mu_{n-1}}\big(z_{n-1}^{(\beta_{i})}\big)\frac{{\rm d}\mu_{n,\alpha_{i}}}{{\rm d}\bar{\mu}_{n}}\big(z_{n-1}^{(\beta_{i})}\big)\\
= & \sum_{\alpha,\beta}E_{A}(\alpha,\beta,z_{n-1}^{(1:N)})\prod_{i=1}^{N}\frac{{\rm d}\mu_{n,\alpha_{i}}}{{\rm d}\mu_{n-1}}\big(z_{n-1}^{(\beta_{i})}\big)\,.
\end{align*}
Now notice that
\begin{align*}
\mathbb{P}_{2}\left(\bar{z}_{n}^{(j)}=z_{n-1,\alpha_{j}}^{(\beta_{j})}\mid z_{n-1}^{(1:N)}\right) & =\frac{\frac{{\rm d}\mu_{n,\alpha_{j}}}{{\rm d}\mu_{n-1}}\big(z_{n-1}^{(\beta_{j})}\big)}{\sum_{i,k}\frac{{\rm d}\mu_{n,k}}{{\rm d}\mu_{n-1}}\big(z_{n-1}^{(i)}\big)}\\
 & \propto\frac{\frac{{\rm d}\mu_{n,\alpha_{j}}}{{\rm d}\mu_{n-1}}\big(z_{n-1}^{(\beta_{j})}\big)}{\sum_{i=1}^{N}\frac{{\rm d}\bar{\mu}_{n}}{{\rm d}\mu_{n-1}}\big(z_{n-1}^{(i)}\big)}\,,
\end{align*}
and the first statement follows from conditional independence of $(b_{1},a_{1}),\ldots,(b_{N},a_{N})$
and the fact that $\mathbb{P}_{2}\big(z_{n}^{(1:N)}\in A\mid\bar{z}_{n-1}^{(1:N)}\big)=E_{A}(\alpha,\beta,\bar{z}_{n-1}^{(1:N)})$.
Since by construction $\mathbb{P}_{3}\big((z_{0}^{(1)},\ldots,z_{0}^{(N)})\in A\big)=\mathbb{P}_{2}\big((z_{0}^{(1)},\ldots,z_{0}^{(N)})\in A\big)$
for any $A\in\mathscr{Z}$ the second statement follows from a standard
Markov chain argument. Now for $g\colon\mathsf{Z}\rightarrow\mathbb{R}$
and $\iota\in\llbracket N\rrbracket$ 
\begin{align*}
\mathbb{E}_{3}\big(g(\check{z}_{n}^{(\iota)})\mid z_{n-1}^{(j)},j\in\llbracket N\rrbracket\big)= & \mathbb{E}_{3}\big(g(z_{n-1}^{(b_{\iota})})\mid z_{n-1}^{(j)},j\in\llbracket N\rrbracket\big)\\
= & \sum_{i=1}^{N}g(z_{n-1}^{(i)})\frac{\frac{{\rm d}\bar{\mu}_{n}}{{\rm d}\mu_{n-1}}\big(z_{n-1}^{(i)}\big)}{\sum_{j=1}^{N}\frac{{\rm d}\bar{\mu}_{n}}{{\rm d}\mu_{n-1}}\big(z_{n-1}^{(j)}\big)}\,.
\end{align*}
Therefore for any $k\in\llbracket0,T\rrbracket$, 
\begin{align*}
\mathbb{E}_{3}\left(\frac{{\rm d}\mu_{n,k}}{{\rm d}\bar{\mu}_{n}}(\check{z}_{n}^{(\iota)})f\circ\psi_{n,k}(\check{z}_{n}^{(\iota)})\mid z_{n-1}^{(j)},j\in\llbracket N\rrbracket\right) & =\sum_{i=1}^{N}\frac{{\rm d}\mu_{n,k}}{{\rm d}\bar{\mu}_{n}}(z_{n-1}^{(i)})f\circ\psi_{n,k}(z_{n-1}^{(i)})\frac{\frac{{\rm d}\bar{\mu}_{n}}{{\rm d}\mu_{n-1}}\big(z_{n-1}^{(i)}\big)}{\sum_{j=1}^{N}\frac{{\rm d}\bar{\mu}_{n}}{{\rm d}\mu_{n-1}}\big(z_{n-1}^{(j)}\big)}\\
 & =\sum_{i=1}^{N}\frac{\frac{{\rm d}\mu_{n,k}}{{\rm d}\mu_{n-1}}(z_{n-1}^{(i)})}{\sum_{j=1}^{N}\sum_{l=0}^{T}\frac{{\rm d}\mu_{n,l}}{{\rm d}\mu_{n-1}}\big(z_{n-1}^{(j)}\big)}f\circ\psi_{n,k}(z_{n-1}^{(i)})\,.
\end{align*}
The third statement follows.
\end{proof}

\section{Sampling Markov snippets\protect\label{sec:Sampling-a-mixture:}}

In this section we develop the Markov snippet framework, largely inspired
by the WF-SMC framework of \cite{dau2020waste} but provide here a
detailed derivation following the standard SMC sampler framework \cite{del2006sequential}
which allows us to consider much more general mutation kernels; integrator
snippet SMC is recovered as a particular case. Importantly we provide
recipes to compute some of the quantities involved using simple criteria
(see Lemma~\ref{lem:ruse-de-sioux} and Corollary~\ref{cor:correspond-ruse-de-sioux})
which allow us to consider unusual scenarios such as in Appendix~\ref{subsec:Sampling-randomized-integrator}.

\subsection{Markov snippet SMC sampler or waste free SMC with a difference \protect\label{sec:justification-waste-free}}

Given a sequence $\big\{\mu_{n},n\in\llbracket0,P\rrbracket\big\}$
of probability distributions defined on a measurable space $(\mathsf{Z},\mathscr{Z})$
introduce the sequence of distributions defined on $\big(\llbracket0,T\rrbracket\times\mathsf{Z}^{T+1},\mathscr{P}(\llbracket0,T\rrbracket)\otimes\mathscr{Z}^{\otimes(T+1)}\big)$
such that for any $(n,k,\mathsf{z})\in\llbracket0,P\rrbracket\times\llbracket0,T\rrbracket\times\mathsf{Z}$ 

\[
\bar{\mu}_{n}(k,{\rm d}\mathsf{z}):=\frac{1}{T+1}w_{n,k}(\mathsf{z})\mu_{n}\otimes M_{n}^{\otimes T}({\rm d}\mathsf{z})
\]
where for $M_{n},L_{n-1,k}\colon\mathsf{Z}\times\mathscr{Z}\rightarrow[0,1]$,
$k\in\llbracket0,P\rrbracket$ and any $\mathsf{z}=(z_{0},z_{1},\ldots,z_{T})\in\mathsf{Z}^{T+1}$,
\[
w_{n,k}(\mathsf{z}):=\frac{{\rm d}\mu_{n}\otimeswapped L_{n-1,k}}{{\rm d}\mu_{n}\otimes M_{n}^{k}}(z_{0},z_{k})\,,
\]
 is assumed to exist for now and $w_{n,0}(\mathsf{z}):=1$. This yields
the marginals
\begin{equation}
\bar{\mu}_{n}({\rm d}\mathsf{z}):=\sum_{k=0}^{T}\bar{\mu}_{n}(k,{\rm d}\mathsf{z})=\frac{1}{T+1}\sum_{k=1}^{n}w_{n,k}(\mathsf{z})\,\mu_{n}\otimes M_{n}^{\otimes T}({\rm d}\mathsf{z})\,.\label{eq:bar-mu-markov-snippet}
\end{equation}
Further, consider $R_{n}\colon\mathsf{Z}\times\mathscr{Z}\rightarrow[0,1]$
such that $\mu_{n-1}R_{n}=\mu_{n-1}$ and define $\bar{M}_{n}\colon\mathsf{Z}^{T+1}\times\mathscr{Z}^{\otimes(T+1)}\rightarrow[0,1]$
\[
\bar{M}_{n}(\mathsf{z},{\rm d}\mathsf{z}'):=\sum_{k=0}^{T}\bar{\mu}_{n-1}(k\mid\mathsf{z})R_{n}(z_{k},{\rm d}z'_{0})M_{n}^{\otimes T}\big(z'_{0},{\rm d}\mathsf{z}'_{-0}\big)\,.
\]
Note that \cite{dau2020waste} set $R_{n}=M_{n}$, which we do not
want in our later application and further assume that $M_{n}$ is
$\mu_{n-1}-$invariant, which is not necessary and too constraining
for our application in Appendix~\ref{sec:Sampling-HMC-trajectories}
(corresponding to the application in the introductory Subsection~\ref{subsec:Outline-justification}).
We only require the condition $\mu_{n-1}R_{n}=\mu_{n-1}$ is required.
As in Subsection~\ref{subsec:Outline-justification} we consider
the optimal backward kernel $\bar{L}_{n}\colon\mathsf{Z}^{T+1}\times\mathscr{Z}^{\otimes(T+1)}\rightarrow[0,1]$,
given for $(\mathsf{z},A)\in\mathsf{Z}^{T+1}\times\mathscr{Z}^{\otimes(T+1)}$
by
\[
\bar{L}_{n-1}(\mathsf{z},A)=\frac{{\rm d}\bar{\mu}_{n-1}\otimes\bar{M}_{n}(A\times\cdot)}{{\rm d}\bar{\mu}_{n-1}\bar{M}_{n}}(\mathsf{z})\,,
\]
and as established in Lemma~\ref{lem:unfolding-mixture} and Lemma~\ref{lem:barM-properties},
one obtains
\begin{align}
\bar{w}_{n}(\mathsf{z}') & =\frac{{\rm d}\bar{\mu}_{n}\otimeswapped\bar{L}_{n-1}}{{\rm d}\bar{\mu}_{n-1}\otimes\bar{M}_{n}}(\mathsf{z},\mathsf{z}')=\frac{{\rm d}\mu_{n}}{{\rm d}\mu_{n-1}}(z'_{0})\frac{1}{T+1}\sum_{k=0}^{T}w_{n,k}(\mathsf{z}')\,.\label{eq:expression-bar-w-waste-free}
\end{align}
The corresponding folded version of the algorithm is given in Alg.~\ref{alg:Folded-Waste-Free-general},
with $\check{\mathsf{z}}_{n}:=(\check{z}_{n,0},\check{z}_{n,1},\ldots,\check{z}_{n,T})$...: 
\begin{itemize}
\item $\big\{(\check{\mathsf{z}}_{n}^{(i)},1),i\in\llbracket N\rrbracket\big\}$
represents $\bar{\mu}_{n}({\rm d}\mathsf{z})$,
\item $\big\{(z_{n,k}^{(i)},w_{n+1,k}),(i,k)\in\llbracket N\rrbracket\times\llbracket0,T\rrbracket\big\}$
represent $\mu_{n}({\rm d}z)$, 
\item $\big\{\big(\mathsf{z}_{n}^{(i)},\bar{w}_{n+1}\big(\mathsf{z}_{n}^{(i)}\big)\big),i\in\llbracket N\rrbracket\big\}$
represents $\bar{\mu}_{n+1}({\rm d}\mathsf{z})$ and so do $\big\{(\check{\mathsf{z}}_{n+1}^{(i)},1),i\in\llbracket N\rrbracket\big\}$. 
\end{itemize}
\begin{algorithm}
sample $\check{\mathsf{z}}_{0}^{(i)}\overset{{\rm iid}}{\sim}\mu_{0}^{\otimes(T+1)}$,
set $w_{0,k}(\check{\mathsf{z}}_{0}^{(i)})=1$ for $(i,k)\in\llbracket N\rrbracket\times\llbracket T\rrbracket$.

\For{$n=0,\ldots,P-1$}{

\For{$i\in\llbracket N\rrbracket$}{

sample $a_{i}\sim{\rm Cat}\left(1,w_{n,1}(\check{\mathsf{z}}_{n}^{(i)}),w_{n,2}(\check{\mathsf{z}}_{n}^{(i)}),\ldots,w_{n,T}(\check{\mathsf{z}}_{n}^{(i)})\right)$

sample $z_{n,0}^{(i)}=z_{n}^{(i)}\sim R_{n+1}(\check{z}_{n,a_{i}},\cdot)$

\For{$k\in\llbracket T\rrbracket$}{

sample $z_{n,k}^{(i)}\sim M_{n+1}(z_{n,k-1}^{(i)},\cdot)$

compute
\[
w_{n+1,k}\big(\mathsf{z}_{n}^{(i)}\big)=\frac{{\rm d}\mu_{n}\otimeswapped L_{n}^{k}}{{\rm d}\mu_{n}\otimes M_{n+1}^{k}}(z_{n,0}^{(i)},z_{n,k}^{(i)})
\]

}

compute

\[
\bar{w}_{n+1}\big(\mathsf{z}_{n}^{(i)}\big)=\frac{{\rm d}\mu_{n+1}}{{\rm d}\mu_{n}}\big(z_{n,0}^{(i)}\big)\frac{1}{T+1}\sum_{k\in\llbracket T\rrbracket}w_{n+1,k}\big(\mathsf{z}_{n}^{(i)}\big)\,,
\]

}

\For{$j\in\mathbb{N}$}{

sample $b_{j}\sim{\rm Cat}\left(\bar{w}_{n+1}\big(\mathsf{z}_{n}^{(1)}\big),\bar{w}_{n+1}\big(\mathsf{z}_{n}^{(2)}\big),\ldots,\bar{w}_{n+1}\big(\mathsf{z}_{n}^{(N)}\big)\right)$

set $\check{\mathsf{z}}_{n+1}^{(j)}=\mathsf{z}_{n}^{(b_{j})}$

}

}

\caption{(Folded) Markov Snippet SMC algorithm \protect\label{alg:Folded-Waste-Free-general}}
\end{algorithm}

\begin{rem}
Other choices are possible and we detail here an alternative, related
to the Waste-free framework of \cite{dau2020waste}. With notation
as above here we take
\end{rem}

\[
\bar{\mu}_{n}(k,{\rm d}\mathsf{z}):=\frac{1}{T+1}w_{n,k}(\mathsf{z})\mu_{n-1}\otimes M_{n}^{\otimes T}({\rm d}\mathsf{z})
\]
with the weights now defined as
\[
w_{n,k}(\mathsf{z}):=\frac{{\rm d}\mu_{n}\otimeswapped L_{n-1,k}}{{\rm d}\mu_{n-1}\otimes M_{n}^{k}}(z_{0},z_{k})\,.
\]
With these choices we retain the fundamental property that for $f\colon\mathsf{Z}\rightarrow\mathbb{R}$,
with $\bar{f}(k,\mathsf{z}):=f(z_{k})$ then $\bar{\mu}_{n}\big(\bar{f}\big)=\mu_{n}(f)$.
Now with
\[
\bar{M}_{n}(\mathsf{z},{\rm d}\mathsf{z}'):=\sum_{k=0}^{T}\bar{\mu}_{n-1}(k\mid\mathsf{z})R_{n}(z_{k},{\rm d}z'_{0})M_{n}^{\otimes T}\big(z'_{0},{\rm d}\mathsf{z}'_{-0}\big)\,,
\]
assuming $\mu_{n-1}R_{n}=\mu_{n-1}$, we have the property that for
any $A\in\mathscr{Z}^{\otimes(T+1)}$ $\bar{\mu}_{n-1}\bar{M}_{n}(A)=\mu_{n-1}\otimes M_{n}^{\otimes T}\big(A\big)$,
yielding
\[
\bar{w}_{n}(\mathsf{z}')=\frac{1}{T+1}\sum_{k=0}^{T}w_{n,k}(\mathsf{z}')\,.
\]
  Finally choosing $L_{n,k}$ to be the optimized backward kernel,
the importance weight of the algorithm is, 
\[
w_{n,k}(\mathsf{z}')=\frac{{\rm d}\mu_{n}\otimeswapped L_{n-1,k}}{{\rm d}\mu_{n-1}\otimes M_{n}^{k}}(z_{0},z_{k})=\frac{{\rm d}\mu_{n}}{{\rm d}\mu_{n-1}}(z_{k})\,.
\]
 We note that continuous time Markov process snippets could also
be used. For example piecewise deterministic Markov processes such
as the Zig-Zag process \cite{bierkens2017piecewise} or the Bouncy
Particle Sampler \cite{bouchard2018bouncy} could be used in practice
since finite time horizon trajectories can be parametrized in terms
of a finite number of parameters. We do not pursue this here.

\subsection{Theoretical justification}

In this section we provide the theoretical justification for the correctness
of Alg.~\ref{alg:Folded-Waste-Free-general} and an alternative proof
for Alg.~\ref{alg:Folded-PDMP-SMC-1}, seen as a particular case
of Alg.~\ref{alg:Folded-Waste-Free-general}.

Throughout this section we use the following notation where $\mu$
(resp. $\nu$) plays the rôle of $\mu_{n+1}$ (resp. $\mu_{n}$),
for notational simplicity. Let $\mu\in\mathcal{P}\big(\mathsf{Z},\mathscr{Z}\big)$
and $M,L\colon\mathsf{Z}\times\mathscr{Z}\rightarrow[0,1]$ be two
Markov kernels such that the following condition holds. For $T\in\mathbb{N}\setminus\{0\}$
let $\mathsf{z}:=(z_{0},z_{1},\ldots,z_{T})\in\mathsf{Z}^{T+1}$ and
for $k\in\llbracket0,T\rrbracket$ assume that $\mu\otimes M^{k}\gg\mu\otimeswapped(L^{k})$,
implying the existence of the Radon-Nikodym derivatives 
\begin{equation}
w_{k}(\mathsf{z}):=\frac{{\rm d}\mu\otimeswapped(L^{k})}{{\rm d}\mu\otimes M^{k}}(z_{0},z_{k})\,,\label{eq:wk-mathsf-z}
\end{equation}
with the convention $w_{0}(\mathsf{z})=1$. We let $w(\mathsf{z}):=(T+1)^{-1}\sum_{k=0}^{T}w_{k}(\mathsf{z})$.
For $\mathsf{z}\in\mathsf{Z}^{T+1}$, define 
\begin{align*}
M^{\otimes T}(z_{0},{\rm d}\mathsf{z}_{-0}): & =\prod_{i=1}^{T}M(z_{i-1},{\rm d}z_{i})\,,
\end{align*}
and introduce the mixture of distributions, defined on $\big(\mathsf{Z}^{T+1},\mathscr{Z}^{\otimes(T+1)}\big)$,
\begin{equation}
\bar{\mu}({\rm d}\mathsf{z})=\sum_{k=0}^{T}\bar{\mu}(k,{\rm d}\mathsf{z})\,,\label{eq:def-barmu-mathsfz}
\end{equation}
where for $(k,\mathsf{z})\in\llbracket0,T\rrbracket\times\mathsf{Z}^{T+1}$
\[
\bar{\mu}(k,{\rm d}\mathsf{z}):=\frac{1}{T+1}w_{k}(\mathsf{z})\mu\otimes M^{\otimes T}({\rm d}\mathsf{z})\,,
\]
so that one can write $\bar{\mu}({\rm d}\mathsf{z})=w(\mathsf{z})\mu\otimes M^{\otimes T}({\rm d}\mathsf{z})$
with $w(\mathsf{z}):=(T+1)^{-1}\sum_{k=0}^{T}w_{k}(\mathsf{z})$. 

We first establish properties of $\bar{\mu}$, showing how samples
from $\bar{\mu}$ can be used to estimate expectations with respect
to $\mu$. 
\begin{lem}
\label{lem:unfolding-mixture} For any $f\colon\mathsf{Z}\rightarrow\mathbb{R}$
such that $\mu(|f|)<\infty$
\end{lem}

\begin{enumerate}
\item we have for $k\in\llbracket0,T\rrbracket$
\[
\int f(z')\frac{{\rm d}\mu\otimeswapped(L^{k})}{{\rm d}\mu\otimes M^{k}}(z,z')\mu({\rm d}z)M^{k}(z,{\rm d}z')=\mu(f)\,,
\]
\item with $\bar{f}(k,\mathsf{z}):=f(z_{k})$, 
\begin{enumerate}
\item then
\begin{align*}
\bar{\mu}(\bar{f}) & =\mu(f)\,,
\end{align*}
\item in particular,
\begin{equation}
\int\sum_{k=0}^{T}\bar{f}(k,\mathtt{\mathsf{z}})\bar{\mu}(k\mid\mathsf{z})\bar{\mu}({\rm d}\mathsf{z})=\mu(f)\,.\label{eq:unfold-mixture-general}
\end{equation}
\end{enumerate}
\end{enumerate}
\begin{proof}
The first statement follows directly from the definition of the Radon-Nikodym
derivative
\begin{align*}
\int f(z')\frac{{\rm d}\mu\otimeswapped L^{k}}{{\rm d}\mu\otimes M^{k}}(z,z')\mu\otimes M^{k}\big({\rm d(}z,z')\big) & =\int f(z')\mu({\rm d}z')L^{k}(z',{\rm d}z)\\
 & =\int f(z')\mu({\rm d}z')\,.
\end{align*}
The second statement follows from the definition of $\bar{\mu}$ and
\[
\frac{1}{T+1}\sum_{k=0}^{T}\int\bar{f}(k,\mathsf{z})w_{k}(\mathsf{z})\mu\otimes M^{\otimes T}({\rm d}\mathsf{z})=\frac{1}{T+1}\sum_{k=0}^{T}\int f(z_{k})\frac{{\rm d}\mu\otimeswapped(L^{k})}{{\rm d}\mu\otimes M^{k}}(z_{0},z_{k})\mu\otimes M^{k}\big({\rm d}(z_{0},z_{k})\big)\,,
\]
and the first statement. The last statement follows from the tower
property for expectations.
\end{proof}
\begin{cor}
Assume that $\mathsf{z}\sim\bar{\mu}$, then
\[
\sum_{k=0}^{T}f(z_{k})\frac{w_{k}(\mathsf{z})}{\sum_{l=0}^{T}w_{l}(\mathsf{z})}
\]
is an unbiased estimator of $\mu(f)$ since we notice that for $k\in\llbracket0,T\rrbracket$,
\[
\bar{\mu}(k\mid\mathsf{z})=\frac{w_{k}(\mathsf{z})}{\sum_{l=0}^{T}w_{l}(\mathsf{z})}\,.
\]
This justifies algorithms which sample from the mixture $\bar{\mu}$
directly in order to estimate expectations with respect to $\mu$.
\end{cor}

Let $\nu\in\mathcal{P}\big(\mathsf{Z},\mathscr{Z}\big)$ and $\bar{\nu}\in\mathcal{P}\big(\mathsf{Z}^{T+1},\mathscr{Z}^{\otimes(T+1)}\big)$
be derived from $\nu$ in the same way $\bar{\mu}$ is from $\mu$
in (\ref{eq:def-barmu-mathsfz}), but for possibly different Markov
kernels $M_{\nu},L_{\nu}\colon\mathsf{Z}\times\mathscr{Z}\rightarrow[0,1]$.
Define now the kernel $\bar{M}\colon\mathsf{Z}^{T+1}\times\mathscr{Z}^{\otimes(T+1)}\rightarrow[0,1]$
be the Markov kernel
\begin{equation}
\bar{M}(\mathsf{z},{\rm d}\mathsf{z}'):=\sum_{k=0}^{T}\bar{\nu}(k\mid\mathsf{z})R(z_{k},{\rm d}z'_{0})M^{\otimes T}\big(z'_{0},{\rm d}\mathsf{z}'_{-0}\big)\,,\label{eq:def-M-bar-nu}
\end{equation}
with $R\colon\mathsf{Z}\times\mathscr{Z}\rightarrow[0,1]$. Remark
that $M,L,M_{\nu},L_{\nu}$ and $R$ can be made dependent on both
$\mu$ and $\nu$, provided they satisfy all the conditions stated
above and below.

The following justifies the existence of $\bar{w}_{n+1}$ in (\ref{eq:expression-bar-w-waste-free})
for a particular choice of backward kernel and provides a simplified
expression for a particular choices of $R$.
\begin{lem}
\label{lem:barM-properties} With the notation (\ref{eq:wk-mathsf-z})-(\ref{eq:def-barmu-mathsfz})
and (\ref{eq:def-M-bar-nu}),
\begin{enumerate}
\item there exists $\bar{M}^{*}\colon\mathsf{Z}^{T+1}\times\mathscr{Z}^{\otimes(T+1)}\rightarrow[0,1]$
such that for any $A,B\in\mathscr{Z}^{\otimes(T+1)}$
\[
\bar{\nu}\otimes\bar{M}(A\times B)=(\bar{\nu}\bar{M})\otimes\bar{M}^{*}(B\times A)
\]
\item we have
\[
\bar{\nu}\bar{M}=(\nu R)\otimes M^{\otimes T}\,,
\]
\item \textup{assuming} $\nu R\gg\mu$ then with the choice $\bar{L}=\bar{M}^{*}$
we have $\bar{\nu}\otimes\bar{M}\gg\bar{\mu}\otimeswapped\bar{L}$
and for $\mathsf{z},\mathsf{z}'\in\mathsf{\mathsf{Z}}^{T+1}$ almost
surely
\[
\bar{w}(\mathsf{z},\mathsf{z}')=\frac{{\rm d}\bar{\mu}\otimeswapped\bar{L}}{{\rm d}\bar{\nu}\otimes\bar{M}}(\mathsf{z},\mathsf{z}')=\frac{{\rm d}\mu}{{\rm d}\nu R}(z'_{0})\frac{1}{T+1}\sum_{k=0}^{T}w_{k}(\mathsf{z}')=\frac{{\rm d}\mu}{{\rm d}\nu R}(z'_{0})w(\mathsf{z}')=:\bar{w}(\mathsf{z}')\,,
\]
that is we introduce notation reflecting this last simplication. 
\item Notice the additional simplification when $\nu R=\nu$.
\end{enumerate}
\end{lem}

\begin{proof}
The first statement is a standard result and $\bar{M}^{*}$ is a conditional
expectation. We however provide the short argument taking advantage
of the specific scenario. For a fixed $A\in\mathscr{Z}^{\otimes(T+1)}$
consider the finite measure
\[
B\mapsto\bar{\nu}\otimes\bar{M}(A\times B)\leq\bar{\nu}\bar{M}(B)\,,
\]
such that $\bar{\nu}\bar{M}(B)=0\implies\bar{\nu}\otimes\bar{M}(A\times B)=0$,
that is $\bar{\nu}\bar{M}\gg\bar{\nu}\otimes\bar{M}(A\times\cdot)$.
Consequently we have the existence of a Radon-Nikodym derivative such
that
\[
\bar{\nu}\otimes\bar{M}(A\times B)=\int_{B}\frac{{\rm d}\bar{\nu}\otimes\bar{M}(A\times\cdot)}{{\rm d}(\bar{\nu}\bar{M})}(\mathsf{z})\bar{\nu}\bar{M}({\rm d}\mathsf{z})
\]
which indeed has the sought property. This is a kernel since for any
fixed $A\in\mathscr{Z}$, for any $B\in\mathscr{Z}$, $\bar{\nu}\otimes\bar{M}(A\times B)\leq1$
and therefore $\bar{M}^{*}(\mathsf{z},A)\leq1$ $\bar{\nu}\bar{M}$-a.s.,
with equality for $A=\mathsf{X}$. For the second statement we have,
for any $(\mathsf{z},A)\in\mathsf{\mathsf{Z}}^{T+1}\times\mathscr{Z}^{\otimes(T+1)}$,
\[
\bar{M}(\mathsf{z},A)=\sum_{k=0}^{T}\bar{\nu}(k\mid\mathsf{z})R\otimes M^{\otimes T}\big(z{}_{k},A\big),
\]
and we can apply (\ref{eq:unfold-mixture-general}) in Lemma~\ref{lem:unfolding-mixture}
with the substitutions $\mu\leftarrow\nu$ and $M\leftarrow M_{\nu},L\leftarrow L_{\nu}$
to conclude. For the third statement, since $\nu R\gg\mu$ then $\bar{\nu}\bar{M}=(\nu R)\otimes M^{\otimes T}\gg\mu\otimes M^{\otimes T}$
because for any $A\in\mathscr{Z}^{\otimes(T+1)}$,
\[
\int\mathbf{1}\{\mathsf{z}\in A\}\nu R({\rm d}z_{0})M^{\otimes T}(z_{0},{\rm d}\mathsf{z}_{-0})=0\implies\begin{cases}
\int\mathbf{1}\{\mathsf{z}\in A\}M^{\otimes T}(z_{0},{\rm d}\mathsf{z}_{-0})=0 & \nu R-a.s.\\
{\rm or}\\
\nu R\big(\mathfrak{P}(A)\big)=0
\end{cases}
\]
where $\mathfrak{P}\colon\mathsf{Z}^{T+1}\rightarrow\mathsf{Z}$ is
such that $\mathfrak{P}(\mathsf{z})=z_{0}$. In either case, since
$\nu R\gg\mu$, this implies that 
\[
\int\mathbf{1}\{\mathsf{z}\in A\}\mu({\rm d}z_{0})M^{\otimes T}(z_{0},{\rm d}\mathsf{z}_{-0})=0\,,
\]
that is $\bar{\nu}\bar{M}\gg\bar{\mu}$ from the definition of $\bar{\mu}$.
Consequently
\begin{align*}
\bar{\mu}\otimes\bar{M}^{*}(A\times B) & =\int_{A}\bar{M}^{*}(\mathsf{z},B)\frac{{\rm d}\bar{\mu}}{{\rm d}(\bar{\nu}\bar{M})}(\mathsf{z})\bar{\nu}\bar{M}({\rm d}\mathsf{z})\\
 & =\int_{A}\frac{{\rm d}\bar{\mu}}{{\rm d}(\bar{\nu}\bar{M})}(\mathsf{z})\bar{M}^{*}(\mathsf{z},B)\bar{\nu}\bar{M}({\rm d}\mathsf{z})\\
 & =\int\mathbf{1}\{\mathsf{z}\in A,\mathsf{z}'\in B\}\frac{{\rm d}\bar{\mu}}{{\rm d}(\bar{\nu}\bar{M})}(\mathsf{z})\bar{\nu}({\rm d}\mathsf{z}')\bar{M}(\mathsf{z}',{\rm d}\mathsf{z})
\end{align*}
and indeed from Fubini's and Dynkin's $\pi-\lambda$ theorems $\bar{\nu}\otimes\bar{M}\gg\bar{\mu}\otimeswapped\bar{L}$
and 
\[
\frac{{\rm d}\bar{\mu}\otimeswapped\bar{L}}{{\rm d}\bar{\nu}\otimes\bar{M}}(\mathsf{z},\mathsf{z}')=\frac{{\rm d}\bar{\mu}}{{\rm d}(\bar{\nu}\bar{M})}(\mathsf{z}')\,.
\]
Now for $f\colon\mathsf{Z}^{T+1}\rightarrow[0,1]$ and using the second
statement and $\nu R\gg\mu$,
\begin{align*}
\int f(\mathsf{z})\frac{{\rm d}\bar{\mu}}{{\rm d}(\bar{\nu}\bar{M})}(\mathsf{z})\,(\bar{\nu}\bar{M})({\rm d}\mathsf{z}) & =\int f(\mathsf{z})w(\mathsf{z})\mu\otimes M^{\otimes T}({\rm d}\mathsf{z})\\
 & =\int f(\mathsf{z})w(\mathsf{z})\frac{{\rm d}\mu}{{\rm d}(\nu R)}(z_{0})(\nu R)({\rm d}z_{0})M^{\otimes T}(z_{0},{\rm d}\mathsf{z}_{-0})\\
 & =\int f(\mathsf{z})w(\mathsf{z})\frac{{\rm d}\mu}{{\rm d}(\nu R)}(z_{0})\bar{\nu}\bar{M}({\rm d}\mathsf{z})\,.
\end{align*}
We therefore conclude that 
\[
\frac{{\rm d}\bar{\mu}\otimeswapped\bar{L}}{{\rm d}\bar{\nu}\otimes\bar{M}}(\mathsf{z},\mathsf{z}')=w(\mathsf{z}')\frac{{\rm d}\mu}{{\rm d}(\nu R)}(z'_{0})=w(\mathsf{z}')\frac{{\rm d}\mu}{{\rm d}\nu}(z'_{0})
\]
where the last inequality holds only when $\nu R=\nu$.
\end{proof}
The following result is important in two respects. First it establishes
that if $M,L$ satisfy a simple property then $w_{k}$ always have
a simple expresssion in terms of certain densities of $\mu$ and $\nu$
-- this implies in particular that in Appendix~\ref{sec:justification-waste-free}
the kernel $M_{n}$ is not required to leave $\mu_{n-1}$ invariant
to make the method implementable \cite{dau2020waste}. Second it provides
a direct justification of the validity of advanced schemes -- see
Example~\ref{exa:delayed-rejection-constrained-set}. This therefore
establishes that generic and widely applicable sufficient conditions
for $\bar{w}(\mathsf{z},\mathsf{z}')$ to be tractable are $\nu R=\nu$
and the notion of $(\upsilon,M,M^{*})$-reversibility.
\begin{lem}
\label{lem:ruse-de-sioux} Let $\mu,\nu\in\mathcal{P}\big(\mathsf{Z},\mathscr{Z}\big)$,
$\upsilon$ be a $\sigma$-finite measure on $\big(\mathsf{Z},\mathscr{Z}\big)$
such that $\upsilon\gg\mu,\nu$ and assume that we have a pair of
Markov kernels $M,M^{*}\colon\mathsf{Z}\times\mathscr{Z}\rightarrow[0,1]$
such that
\begin{align}
\upsilon({\rm d}z)M(z,{\rm d}z') & =\upsilon({\rm d}z')M^{*}(z',{\rm d}z)\,.\label{eq:nu-detailed-balance}
\end{align}
 We call this property $(\upsilon,M,M^{*})$-reversibility. Then for
$z,z'\in\mathsf{Z}$ such that $\nu(z):={\rm d}\nu/{\rm d}\upsilon(z)>0$
we have
\[
\frac{{\rm d}\mu\otimeswapped M^{*}}{{\rm d}\nu\otimes M}(z,z')=\frac{{\rm d}\mu/{\rm d}\upsilon(z')}{{\rm d}\nu/{\rm d}\upsilon(z)}\,.
\]
\end{lem}

\begin{proof}
For $z,z'\in\mathsf{Z}$ such that ${\rm d}\nu/{\rm d}\upsilon(z)>0$
we have
\begin{align*}
\mu({\rm d}z')M^{*}(z',{\rm d}z) & =\frac{{\rm d}\mu}{{\rm d}\upsilon}(z')\upsilon({\rm d}z')M^{*}(z',{\rm d}z)\\
 & =\frac{{\rm d}\mu}{{\rm d}\upsilon}(z')\upsilon({\rm d}z)M(z,{\rm d}z')\\
 & =\frac{{\rm d}\mu/{\rm d}\upsilon(z')}{{\rm d}\nu/{\rm d}\upsilon(z)}\frac{{\rm d}\nu}{{\rm d}\upsilon}(z)\upsilon({\rm d}z)M(z,{\rm d}z')\\
 & =\frac{{\rm d}\mu/{\rm d}\upsilon(z')}{{\rm d}\nu/{\rm d}\upsilon(z)}\nu({\rm d}z)M(z,{\rm d}z')\,.
\end{align*}
\end{proof}
\begin{cor}
\label{cor:correspond-ruse-de-sioux} Let $\mu_{n-1},\mu_{n}\in\mathcal{P}\big(\mathsf{Z},\mathscr{Z}\big)$
and $\upsilon$ be a $\sigma$-finite measure such that $\upsilon\gg\mu_{n-1},\mu_{n}$,
let $M_{n},L_{n-1}\colon\big(\mathsf{Z},\mathscr{Z}\big)\rightarrow[0,1]$
such that
\[
\upsilon({\rm d}z)M_{n}(z,{\rm d}z')=\upsilon({\rm d}z')L_{n-1}(z',{\rm d}z)\,,
\]
then for any $z,z'\in\mathsf{Z}$ such that ${\rm d}\mu_{n}/{\rm d}\upsilon(z)>0$
and $k\in\mathbb{N}$
\begin{align*}
\frac{\mu_{n}\otimeswapped L_{n-1}^{k}}{\mu_{n}\otimes M_{n}^{k}}(z,z')=\frac{{\rm d}\mu_{n}/{\rm d}\upsilon(z')}{{\rm d}\mu_{n}/{\rm d}\upsilon(z)}\,.
\end{align*}
and provided $\mu_{n-1}R_{n}=\mu_{n-1}$ we can deduce
\begin{align*}
\bar{w}_{n}(\mathsf{z}) & =\frac{{\rm d}\mu_{n}/{\rm d}\upsilon(z_{0})}{{\rm d}\mu_{n-1}/{\rm d}\upsilon(z_{0})}\frac{1}{T+1}\sum_{k=0}^{T}\frac{{\rm d}\mu_{n}/{\rm d}\upsilon(z_{k})}{{\rm d}\mu_{n}/{\rm d}\upsilon(z_{0})}\\
 & =\frac{1}{T+1}\sum_{k=0}^{T}\frac{{\rm d}\mu_{n}/{\rm d}\upsilon(z_{k})}{{\rm d}\mu_{n-1}/{\rm d}\upsilon(z_{0})}
\end{align*}
where we have used Lemma ~\ref{lem:Billingsley-problem324}.
\end{cor}

We have shown earlier that standard integrator based mutation kernels
used in the context of Monte Carlo method satisfy (\ref{eq:nu-detailed-balance})
with $\upsilon$ the Lebesgue measure but other scenarios involved
that of the preconditioned Crank--Nicolson (pCN) algorithm where
$\upsilon$ is the distribution of a Gaussian process.

\subsection{Revisiting sampling with integrator snippets \protect\label{sec:Sampling-HMC-trajectories}}

In this scenario we have $\mu({\rm d}z)=\pi({\rm d}x)\varpi({\rm d}v)$
assumed to have a density w.r.t. a $\sigma$-finite measure $\upsilon$,
and $\psi$ is an invertible mapping $\psi\colon\mathsf{Z}\rightarrow\mathsf{Z}$
such that $\upsilon^{\psi}=\upsilon$ and $\psi^{-1}=\sigma\circ\psi\circ\sigma$
with $\sigma\colon\mathsf{Z}\rightarrow\mathsf{Z}$ such that $\mu\circ\sigma(z)=\mu(z)$.
In his manuscript we focus primarily on the scenario where $\psi$
is a discretization of Hamilton's equations for a potential $U\colon\mathsf{X}\rightarrow\mathbb{R}$
e.g. a leapfrog integrator. We consider now the scenario where, in
the framework developed in Appendix~\ref{sec:justification-waste-free},
we let $M(z,{\rm d}z'):=\delta_{\psi(z)}\big({\rm d}z'\big)$ be the
deterministic kernel which maps the current state $z\in\mathsf{Z}$
to $\psi(z)$. Define $\Psi(z,{\rm d}z'):=\delta_{\psi(z)}({\rm d}z')$
and $\Psi^{*}(z,{\rm d}z'):=\delta_{\psi^{-1}(z)}({\rm d}z')$; we
exploit the ideas of \cite[Proposition 4]{andrieu2020general} to
establish that $\Psi^{*}$ is the $\upsilon-$adjoint of $\Psi$ if
$\upsilon$ is invariant under $\psi$. 
\begin{lem}
\label{lem:nu-Psi-Psi-star-triplet}Let $\mu$ be a probability measure
and $\upsilon$ a $\sigma$-finite measure, on $(\mathsf{Z},\mathscr{Z})$
such that $\upsilon\gg\mu$. Denote $\mu(z):={\rm d}\mu/{\rm d}\upsilon(z)$
for any $z\in\mathsf{Z}$. Let $\psi\colon\mathsf{Z}\rightarrow\mathsf{Z}$
be an invertible and volume preserving mapping, i.e. such that $\upsilon^{\psi}(A)=\upsilon\big(\psi^{-1}(A)\big)=\upsilon(A)$
for all $A\in\mathscr{Z}$, then 
\begin{enumerate}
\item $(\upsilon,\Psi,\Psi^{*})$ form a reversible triplet, that is for
all $z,z'\in\mathsf{Z}$,
\[
\upsilon({\rm d}z)\delta_{\psi(z)}({\rm d}z')=\upsilon({\rm d}z')\delta_{\psi^{-1}(z')}({\rm d}z),
\]
\item for all $z,z'\in\mathsf{Z}$ such that $\mu(z)>0$
\[
\mu({\rm d}z')\delta_{\psi^{-1}(z')}({\rm d}z)=\frac{\mu\circ\psi(z)}{\mu(z)}\mu({\rm d}z)\delta_{\psi(z)}({\rm d}z')\,.
\]
\end{enumerate}
\end{lem}

\begin{proof}
For the first statement
\begin{align*}
\int f(z)g\circ\psi(z)\upsilon({\rm d}z) & =\int f\circ\psi^{-1}\circ\psi(z)g\circ\psi(z)\upsilon({\rm d}z)\\
 & =\int f\circ\psi^{-1}(z)g(z)\upsilon^{\psi}({\rm d}z)\\
 & =\int f\circ\psi^{-1}(z)g(z)\upsilon({\rm d}z).
\end{align*}
We have
\begin{align*}
\mu({\rm d}z')\delta_{\psi^{-1}(z')}({\rm d}z) & =\mu(z')\upsilon({\rm d}z')\delta_{\psi^{-1}(z')}({\rm d}z)\\
 & =\mu(z')\upsilon({\rm d}z)\delta_{\psi(z)}({\rm d}z')\\
 & =\mu\circ\psi(z)\upsilon({\rm d}z)\delta_{\psi(z)}({\rm d}z')\\
 & =\frac{\mu\circ\psi(z)}{\mu(z)}\mu(z)\upsilon({\rm d}z)\delta_{\psi(z)}({\rm d}z')\\
 & =\frac{\mu\circ\psi(z)}{\mu(z)}\mu({\rm d}z)\delta_{\psi(z)}({\rm d}z')
\end{align*}
\begin{align*}
\int f(z')g(z)\mu({\rm d}z')\delta_{\psi^{-1}(z')}({\rm d}z) & =\int f(z')g(z)\mu(z')\upsilon({\rm d}z')\delta_{\psi^{-1}(z')}({\rm d}z)\\
 & =\int f(z')g(z)\mu(z')\upsilon({\rm d}z)\delta_{\psi(z)}({\rm d}z')\\
 & =\int f\circ\psi(z)g(z)\mu\circ\psi(z)\upsilon({\rm d}z)\\
 & =\int f\circ\psi(z)g(z)\frac{\mu\circ\psi(z)}{\mu(z)}\mu(z)\upsilon({\rm d}z)\\
 & =\int f(z')g(z)\frac{\mu\circ\psi(z)}{\mu(z)}\mu({\rm d}z)\delta_{\psi(z)}({\rm d}z')
\end{align*}
As a result
\begin{align*}
\int f(z')\mu({\rm d}z') & =\int f(z')\mu({\rm d}z')\delta_{\psi^{-1}(z')}({\rm d}z)\\
 & =\int f(z')\frac{\mu\circ\psi(z)}{\mu(z)}\mu({\rm d}z)\delta_{\psi(z)}({\rm d}z')\\
 & =\int f\circ\psi(z)\frac{\mu\circ\psi(z)}{\mu(z)}\mu({\rm d}z)
\end{align*}
\end{proof}
\begin{cor}
\label{cor:weights-HMC-scenario} With the assumptions of Lemma~\ref{lem:nu-Psi-Psi-star-triplet}
above for $k\in\llbracket0,T\rrbracket$ the weight (\ref{eq:wk-mathsf-z})
for $M_{\mu}=\Psi^{k},L_{\mu}=\Psi^{-k}$ admits the expression
\[
w_{k}(\mathsf{z})=\frac{{\rm d}\mu\otimeswapped\Psi{}^{-k}}{{\rm d}\mu\otimes\Psi^{k}}(z_{0},z_{k})=\frac{\mu\circ\psi^{k}(z_{0})}{\mu(z_{0})}\,,
\]
Further, for $R$ a $\nu$\textup{-}invariant Markov kernel the expression
for the weight (\ref{eq:expression-bar-w-waste-free}) becomes 
\begin{align*}
\bar{w}(\mathtt{\mathsf{z}},\mathtt{\mathsf{z}}') & =\frac{\mu(z'_{0})}{\nu(z'_{0})}\frac{1}{T+1}\sum_{k=0}^{T}\frac{\mu\circ\psi^{k}(z'_{0})}{\mu(z'_{0})},
\end{align*}
hence recovering the expression used in Section~\ref{sec:A-simple-example:}.
This together with the results of Appendix~\ref{sec:justification-waste-free}
provides an alternative justification of correctness of Alg.~\ref{alg:Folded-PDMP-SMC-1}
and hence Alg.~\ref{alg:Unfolded-PDMP-SMC-1}. Note that this choice
for $L_{\mu}$ corresponds to the so-called optimal scenario; this
can be deduced from Lemma~\ref{lem:nu-Psi-Psi-star-triplet} or by
noting that
\begin{align*}
\int f(z,z')\mu\Psi({\rm d}z')\Psi^{-1}(z',{\rm d}z) & =\int f(\psi^{-1}(z'),z')\mu\Psi({\rm d}z')\\
 & =\int f(z,\psi(z))\mu({\rm d}z)\\
 & \int f(z,z')\mu({\rm d}z)\Psi(z,{\rm d}z')\,.
\end{align*}
\end{cor}

\subsection{More complex integrator snippets \protect\label{subsec:Sampling-randomized-integrator}}

Here we provide examples of more complex integrator snippets to which
earlier theory immediately applies thanks to the abstract point of
view adopted throughout. 
\begin{example}
\label{exa:delayed-rejection-generic} Let $\upsilon\gg\mu$, so that
$\mu(z):={\rm d}\mu/{\rm d}\upsilon(z)$ is well defined. Let, for
$i\in\{1,2\}$, $\psi_{i}\colon\mathsf{Z}\rightarrow\mathsf{Z}$ be
invertible and such that $\upsilon^{\psi_{i}}=\upsilon$, $\psi_{i}^{-1}=\sigma\circ\psi_{i}\circ\sigma$
for $\sigma\colon\mathsf{Z}\rightarrow\mathsf{Z}$ such that $\sigma^{2}={\rm Id}$
and $\upsilon^{\sigma}=\upsilon$. Then one can consider the delayed
rejection MH transition probability 
\[
M(z,{\rm d}z')=\alpha_{1}(z)\delta_{\psi_{1}(z)}\big({\rm d}z'\big)+\bar{\alpha}_{1}(z)\big[\alpha_{2}(z)\delta_{\psi_{2}(z)}({\rm d}z')+\bar{\alpha}_{2}(z)\delta_{z}({\rm d}z')\big]\,,
\]
with $\alpha_{i}(z)=1\wedge r_{i}(z)$ $\bar{\alpha}_{i}(z)=1-\alpha_{i}(z)$
and $z\in S_{1}\cap S_{2}$ given below
\[
r_{1}(z)=\frac{{\rm d}\upsilon^{\psi_{1}^{-1}}}{{\rm d}\upsilon}(z),\quad r_{2}(z)=\frac{\bar{\alpha}_{1}\circ\sigma\circ\psi_{2}(z)}{\bar{\alpha}_{1}(z)}\frac{{\rm d}\upsilon^{\psi_{2}^{-1}}}{{\rm d}\upsilon}(z)\,.
\]
A particular example is when $\upsilon$ is the Lebesgue measure on
$\mathsf{X}\times\mathsf{V}$ and $\psi_{1}$ is volume preserving,
for instance the leapfrog integrator for Hamilton's equations for
some potential $U$ or the bounce update (\ref{eq:def-bounce}). Following
\cite{andrieu2020general} notice that $\upsilon$ has density $\upsilon(z)=1/2$
with respect to the measure $\upsilon+\upsilon^{\psi_{i}}=2\upsilon$.
Now define $S_{1}:=\big\{ z\in\mathsf{Z}\colon\upsilon(z)\wedge\upsilon\circ\psi_{1}(z)>0\big\}$,
we have
\[
r_{1}(z):=\begin{cases}
\frac{\upsilon\circ\psi_{1}(z)}{\upsilon(z)}=1 & \text{for }z\in S_{1}\\
0 & \text{otherwise}
\end{cases},
\]
 and with $S_{2}:=\big\{ z\in\mathsf{Z}\colon\big[\bar{\alpha}_{1}(z)\,\upsilon(z)\big]\wedge\big[\bar{\alpha}_{1}\circ\sigma\circ\psi_{2}(z)\,\upsilon\circ\psi_{2}(z)\big]>0\big\}$
\[
r_{2}(z):=\begin{cases}
\frac{\bar{\alpha}_{1}\circ\sigma\circ\psi_{2}(z)\,\upsilon\circ\psi_{2}(z)}{\bar{\alpha}_{1}(z)\,\upsilon(z)}=1 & \text{for }z\in S_{2}\\
0 & \text{otherwise}
\end{cases}.
\]
For instance $\psi_{1}$ can be a HMC update, while 
\begin{equation}
\psi_{2}=\psi_{1}\circ{\rm b}\circ\psi_{1}\,\text{with}\,{\rm b}(x,v)=\big(x,v-2\bigl\langle v,n(x)\bigr\rangle n(x)\big)\label{eq:def-bounce}
\end{equation}
for some unit length vector field $n\colon\mathsf{X}\rightarrow\mathsf{X}$.
This can therefore be used as part of Alg.~\ref{alg:Folded-Waste-Free-general};
care must be taken when compute the weights $w_{n,k}$ and $\bar{w}_{n}$,
see Appendix~\ref{sec:Sampling-a-mixture:}-\ref{sec:Sampling-HMC-trajectories}
provide the tools to achieve this. For example, in the situation where
$\upsilon$ is the Lebesgue measure on $\mathsf{Z}=\mathbb{R}^{d}$
and $\psi_{2}={\rm Id}$ then we recover Alg.~\ref{alg:Folded-PDMP-SMC-1}
or equivalently Alg.~\ref{alg:Unfolded-PDMP-SMC-1}. 
\end{example}

\begin{example}
\label{exa:delayed-rejection-constrained-set} An interesting instance
of Example~\ref{exa:delayed-rejection-generic} is concerned with
the scenario where interest is in sampling $\mu$ constrained to some
set $C\subset\mathsf{Z}$ such that $\upsilon(C)<\infty$. Define
$\mu$ constrainted to $C$, $\mu_{C}(\cdot):=\mu(C\cap\cdot)/\mu(C)$
and similarly $\upsilon_{C}(\cdot):=\upsilon(C\cap\cdot)/\upsilon(C)$.
We let $M$ be defined as above but targeting $\upsilon_{C}$. Naturally
$\upsilon_{C}$ has a density w.r.t. $\upsilon$, $\upsilon_{C}(z)=\mathbf{1}_{C}(z)/\upsilon(C)$
for $z\in\mathsf{Z}$ and for $i\in\{1,2\}$ we have $\upsilon_{C}\circ\psi_{i}(z)=\mathbf{1}_{\psi_{i}^{-1}(C)}(z)$
, $\upsilon\circ\psi_{1}\circ\psi_{2}(z)=\mathbf{1}_{\psi_{2}^{-1}\circ\psi_{1}^{-1}(C)}(z)$.
Consequently $S_{1}:=\big\{ z\in\mathsf{Z}\colon\mathbf{1}_{C\cap\psi_{1}^{-1}(C)}(z)>0\big\}$
and $S_{2}=\big\{ z\in\mathsf{Z}\colon\mathbf{1}_{C\cap\psi_{1}^{-1}(C^{\complement})}(z)\mathbf{1}_{\psi_{2}^{-1}(C)\cap\psi_{2}^{-1}\circ\psi_{1}^{-1}(C^{\complement})}(z)>0\big\}$
and as a result
\begin{align*}
\alpha_{1}(z) & :=\mathbf{1}_{A\cap\psi_{1}^{-1}(C)}(z)\\
\alpha_{2}(z) & :=\mathbf{1}_{A\cap\psi_{1}^{-1}(C^{\complement})\cap\psi_{2}^{-1}(C)\cap\psi_{2}^{-1}\circ\psi_{1}^{-1}(C^{\complement})}(z)\,.
\end{align*}
The corresponding kernel $M^{\otimes T}$ is described algorithmically
in Alg.~\ref{alg:deterministic-DR-uniform-target}. In the situation
where $C:=\{x\in\mathsf{X}\colon c(x)=0\}$ for a continuously differentiable
function $c\colon\mathsf{X}\rightarrow\mathbb{R}$, the bounces described
in (\ref{eq:def-bounce}) can be defined in terms of the field $x\mapsto n(x)$such
that 
\begin{align*}
n(x) & :=\begin{cases}
\nabla c(x)/|\nabla c(x)| & \text{for}\,\nabla c(x)\neq0\\
0 & \text{otherwise.}
\end{cases}
\end{align*}
This justifies the ideas of \cite{betancourt2011nested}, where a
process of the type given in Alg.~\ref{alg:deterministic-DR-uniform-target}
is used as a proposal within a MH update, although the possibility
of a rejection after the second stage seems to have been overlooked
in that reference.
\end{example}

\begin{algorithm}
Given $z_{0}=z\in C\subset\mathsf{Z}$

\For{$k=1,\ldots,T$}{

\uIf{$\psi_{1}(z_{k-1})\in C$}{

$z_{k}=\psi_{1}(z_{k-1})$

}\uElseIf{$\psi_{2}(z_{k-1})\in C$ $\text{and}$ $\psi_{1}\circ\psi_{2}(z_{k-1})\notin C$}{

$z_{k}=\psi_{2}(z_{k-1})$

}\Else{

$z_{k}=z_{k-1}$.

}}

\caption{$M^{\otimes T}$ for $M$ the delayed rejection algorithm targetting
the uniform distribution on $C$ \protect\label{alg:deterministic-DR-uniform-target}}
\end{algorithm}
Naturally a rejection of both transformations $\psi_{1}$ and $\psi_{2}$
of the current state means that the algorithm gets stuck. We note
that it is also possible to replace the third update case with a full
refreshment of the velocity, which can be interpreted as a third delayed
rejection update, of acceptance probability one. 

\section{Early experimental results: HMC-IS \protect\label{app:early-experimental-HMC}}

\subsection{Numerical illustration: logistic regression \protect\label{subsec:first-example-simulations}}

In this section, we consider sampling from the posterior distribution
of a logistic regression model, focusing on the compution of the normalising
constant. We follow \cite{dau2020waste} and consider the sonar dataset,
previously used in \cite{chopin2017leave}. With intercept terms,
the dataset has responses $y_{i}\in\{-1,1\}$ and covariates $z_{i}\in\mathbb{R}^{p}$,
where $p=61$. The likelihood of the parameters $x\in\mathsf{X}:=\mathbb{R}^{p}$
is then given by 
\begin{equation}
L(x)=\prod_{i}^{n}F(z_{i}^{\top}x\cdot y_{i}),\label{Logistic=000020Regression=000020Llk}
\end{equation}
where $F(x):=1/(1+\exp(-x))$. We ascribe $x$ a product of independent
normal distributions of zero mean as a prior, with standard deviation
equal to 20 for the intercept and 5 for the other parameters. Denote
$p(dx)$ the prior distribution of $x$, we define a sequence of tempered
distributions of densities of the form $\pi_{n}(x)\propto p(dx)L(x)^{\lambda_{n}}$
for $\lambda_{n}\colon\llbracket0,P\rrbracket\rightarrow[0,1]$ non-decreasing
and such that $\lambda_{0}=0$ and $\lambda_{P}=1$. We apply both
Hamiltonian Snippet-SMC and the implementation of waste-free SMC of
\cite{dau2020waste} and compare their performance.

For both algorithms, we set the total number of particles at each
SMC step to be $N(T+1)=10,000$. For the waste-free SMC, the Markov
kernel is chosen to be a random-walk Metropolis-Hastings kernel with
covariances adaptively computed as $2.38^{2}/d\ \hat{\Sigma}$, where
$\hat{\Sigma}$ is the empirical covariance matrix obtained from the
particles in the previous SMC step. For the Hamiltonian Snippet-SMC
algorithm, we set $\psi_{n}$ to be the one-step leap-frog integrator
with stepsize $\epsilon$, $U_{n}(x)=-\log(\pi_{n}(x))$ and $\varpi$
a $\mathcal{N}(0,\mathrm{Id})$. To investigate the stability of our
algorithm, we ran Hamiltonian Snippet SMC with $\epsilon=0.05,0.1,0.2$
and $0.3$. For both algorithms, the temperatures $\lambda_{n}$ are
adaptively chosen so that the effective sample size (ESS) of the current
SMC step will be $\alpha ESS_{max}$, where $ESS_{max}$ is the maximum
ESS achievable at the current step. In our experiments, we have chosen
$\alpha=0.3,0.5$ and $0.7$ for both algorithms.

\subsubsection*{Performance comparison}

Figure~\ref{Sonar_LogNC_alpha50} shows the boxplots of estimates
of the logarithm of the normalising constant obtained from both algorithms,
for different choices of $N$ and $\epsilon$ for the Hamiltonian
Snippet SMC algorithm. The boxplots are obtained by running both algorithms
100 times for different of the algorithm parameters, with $\alpha=0.5$
in all setups. Several points are worth observing. For a suitably
choice of $\epsilon$, the Hamiltonian Snippet SMC algorithm can produce
stable and consistent estimates of the normalising constant with $10,000$
particles at each iteration. On the other hand, however, the waste-free
SMC algorithm fails to produce accurate results for the same computational
budget. It is also clear that with larger values of $N$ (meaning
smaller value of $T$ and hence shorter snippets), the waste-free
SMC algorithm produces results with larger biases and variability.
For Hamiltonian Snippet SMC algorithm, the results are stable both
for short and long snippets when $\epsilon$ is equal to $0.1$ or
$0.2$. Another point is that when $\epsilon=0.05$ or $0.3$, the
Hamiltonian Snippet SMC algorithm becomes unstable with short (i.e.
$\epsilon=0.05$) and long (i.e. $\epsilon=0.3$). Possible reasons
are for too small a stepsize the algorithm is not able to explore
the target distribution efficiently, resulting in unstable performances.
On the other hand, when the stepsize is too large, the leapfrog integrator
becomes unstable, therefore affecting the variability of the weights;
this is a common problem of HMC algorithms. Hence, a long trajectory
will result in detoriate estimations. Hence, to obtain the best performance,
one should find a way of tuning the stepsize and trajectory length
along with the SMC steps. 
\begin{figure}[htb!]
\includegraphics[width=1\textwidth]{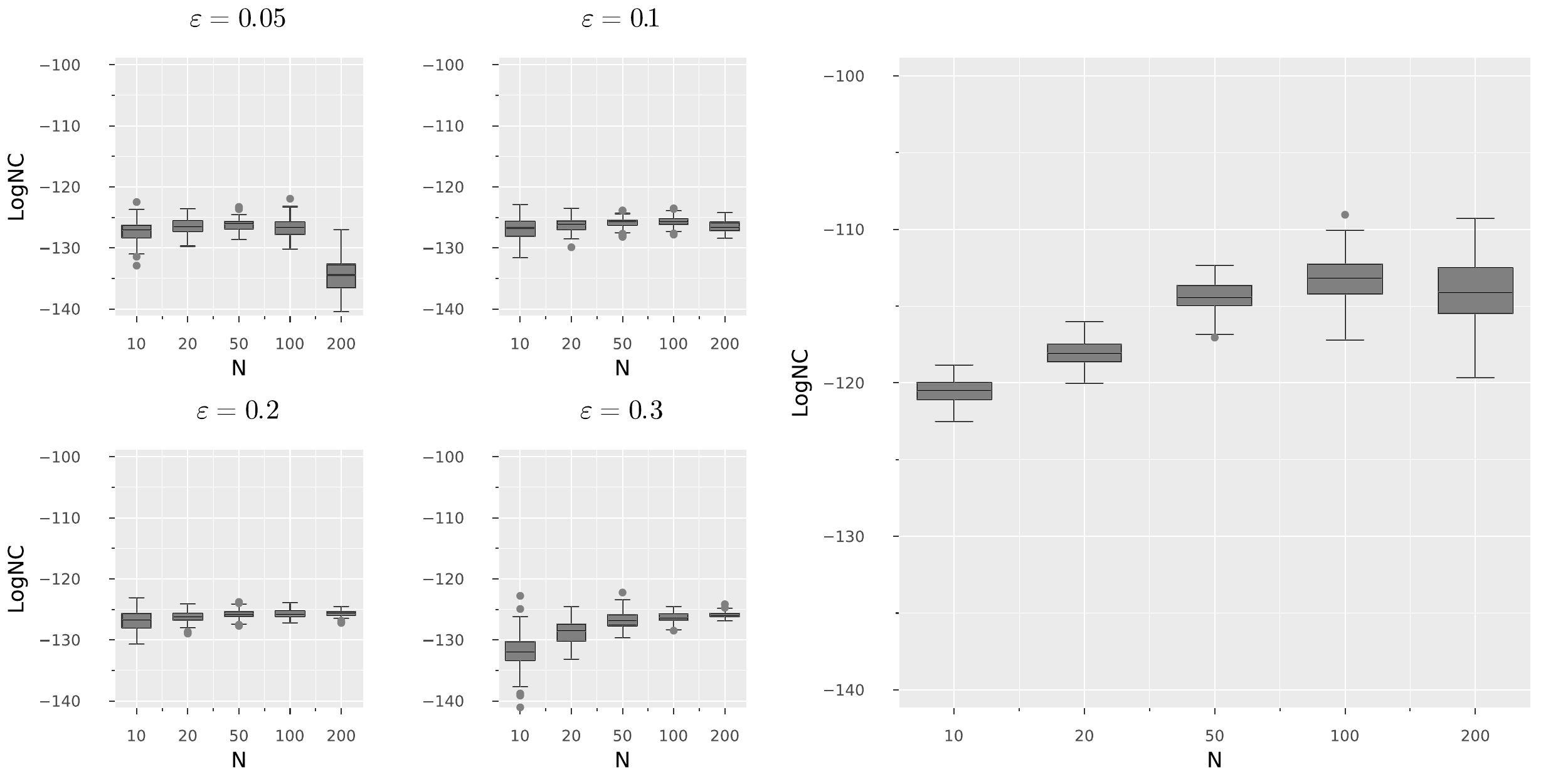} \caption{Estimates of the normalising constant (in log scale) obtained from
both Hamiltonian Snippet SMC and waste-free SMC algorithm. Left: Estimate
obtained from Hamiltonian Snippet SMC algorithm with different values
of $\epsilon$. Right: Estimates obtained from waste-free SMC algorithm.}
\label{Sonar_LogNC_alpha50}
\end{figure}

In Figure~\ref{Sonar_MM_alpha50} we display boxplots of the estimates
of the posterior expectations of the mean of all coefficients, i.e.
$\mathbb{E}{}_{\pi_{P}}(d^{-1}\sum_{i=1}^{d}x_{i})$. This quantity
is referred to as the \textit{mean of marginals} in \cite{dau2020waste}
and we use this terminology. One can see that the same properties
can be seen from the estimations of the mean of marginals, with the
unstability problems exacerbated with small and large stepsizes.

\begin{figure}[htb!]
\includegraphics[width=1\textwidth]{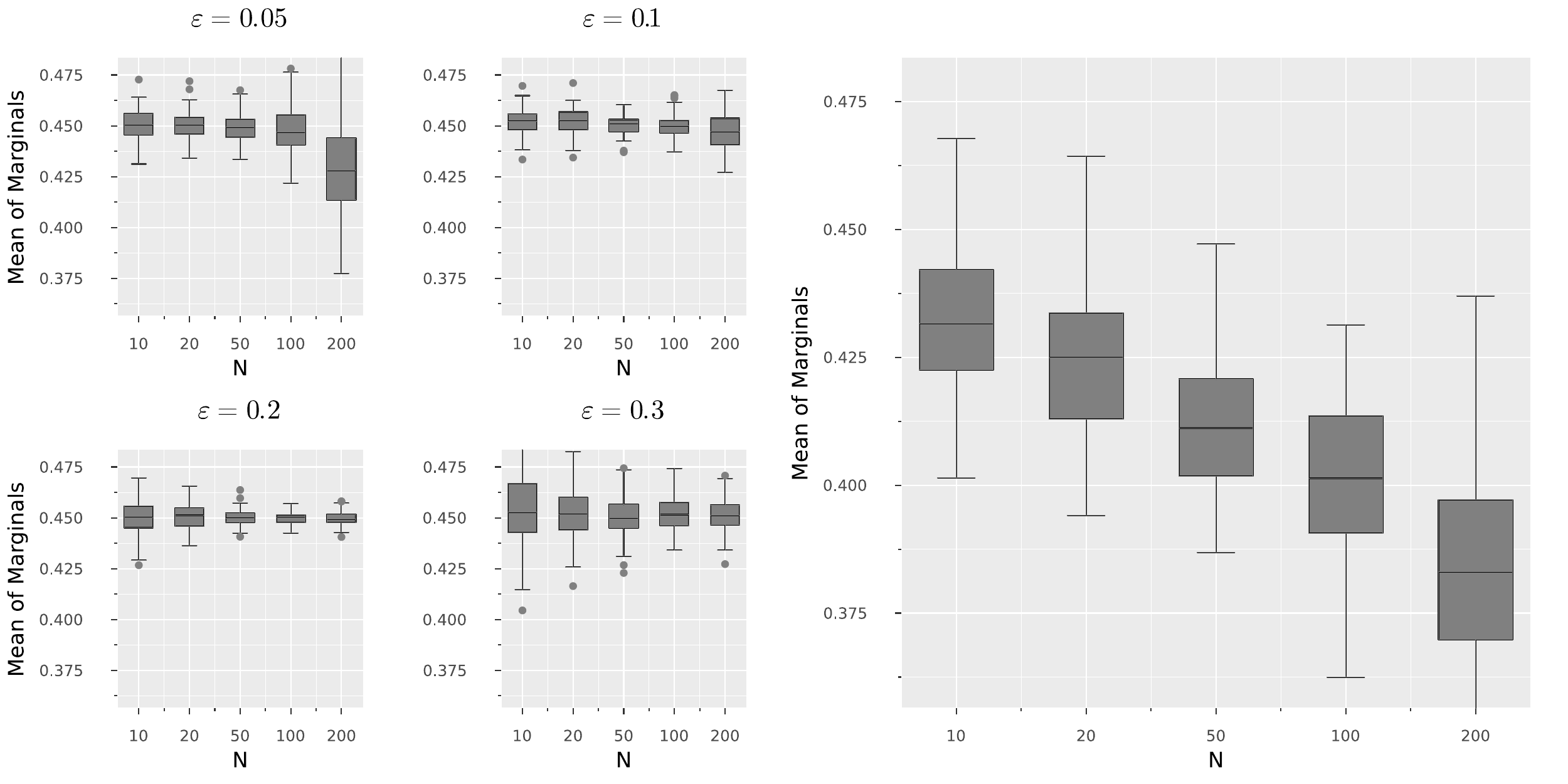} \caption{Estimates of the mean of marginals obtained from both Hamiltonian
Snippet SMC and the waste-free SMC algorithms. Left: Estimate obtained
from the Hamiltonian Snippet SMC algorithm with difference values
of $\epsilon$. Right: Estimates obtained from the waste-free SMC
algorithm.}
\label{Sonar_MM_alpha50}
\end{figure}

\subsubsection*{Computational Cost}

In this section, we compare the running time of both algorithms. Since
the calculations of the potential energy and its gradient often share
common intermediate steps, we can recycle these to save computational
cost. As the waste-free SMC also requires density evaluations, the
Hamiltonian Snippet SMC algorithm will not require significant additional
computations. Figure~\ref{Sonar_Simtime} shows boxplots of the simulation
time of both algorithms from 100 runs. The simulations were run on
an Apple M1-Pro CPU with 16G of memory. One can see that in comparison
to the waste-free SMC the additional computational time is only marginal
for the Hamiltonian Snippet SMC algorithm and mostly due to the larger
memory needed to store the intermediate values. 
\begin{figure}[htb!]
\includegraphics[width=1\textwidth]{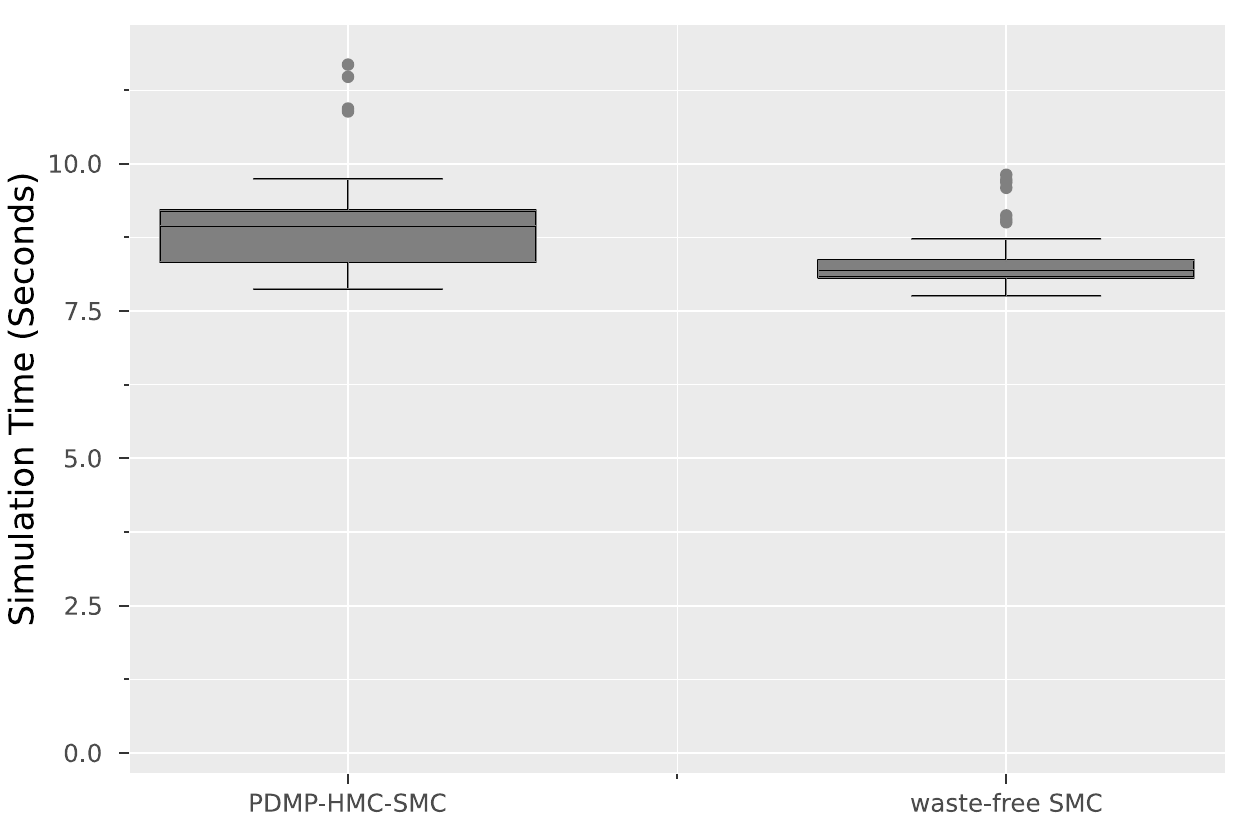} \caption{Simulation times for both algorithms. Left: Simulation time for Hamiltonian
Snippet SMC algorithm, with $N=100,\epsilon=0.2,\alpha=0.5$. Right:Simulation
times for the waste-free SMC with $N=100$}
\label{Sonar_Simtime}
\end{figure}

\subsection{Numerical illustration: orthant probabilities \protect\label{subsec:Numerical-illustration:-orthant}}

In this section, we consider the problem of calculating the Gaussian
orthant probabilities, which is given by 
\[
p(\mathbf{a},\mathbf{b},\Sigma):=\mathbb{P}(\mathbf{a}\leq\mathbf{X}\leq\mathbf{b}),
\]
where $\mathbf{a},\mathbf{b}\in\mathbb{R}^{d}$ are known vectors
of dimension $d$ and $\mathbf{X}\sim\mathcal{N}_{d}(\mathbf{0},\Sigma)$
with $\Sigma$ a covariance matrix of size $d\times d$. Consider
the Cholesky decomposition of $\Sigma$ which is given by $\Sigma:=LL^{\top}$,
where $L:=(l_{ij},i,j\in\llbracket d\rrbracket)$ is a lower triangular
matrix with positive diagonal entries. It is clear that $\mathbf{X}$
can be viewed as $\mathbf{X}:=L\mathbf{\eta}$, where $\mathbf{\eta}\sim\mathcal{N}_{d}(\mathbf{0},{\rm Id}_{d})$.
Consequently, one can instead rewrite $p(\mathbf{a},\mathbf{b},\Sigma)$
as a product of $d$ probabilities given by
\begin{align}
p_{1}=\mathbb{P}(a_{1}\leq l_{11}\eta_{1}\leq b_{1})=\mathbb{P}\left(a_{1}/l_{11}\leq\eta_{1}\leq b_{1}/l_{11}\right)\,,
\end{align}
and
\begin{align}
p_{n}=\mathbb{P}\left(a_{t}\leq{\textstyle \sum}_{j=1}^{n}l_{nj}\eta_{j}\leq b_{t}\right)=\mathbb{P}\left(\nicefrac{(a_{t}-\sum_{j=1}^{n-1}l_{nj}\eta_{j})}{l_{nn}}\leq\eta_{t}\leq\nicefrac{(b_{t}-\sum_{j=1}^{n-1}l_{nj}\eta_{j})}{l_{nn}}\right)\,,
\end{align}
for $n=2,\ldots,d$. For notational simplicity, we let $\mathcal{B}_{n}(\eta_{1:n-1})$
denote the interval $\left[\nicefrac{(a_{t}-\sum_{j=1}^{n-1}l_{nj}\eta_{j})}{l_{nn}},\nicefrac{(b_{t}-\sum_{j=1}^{n-1}l_{nj}\eta_{j})}{l_{nn}}\right],$
with the convention $\mathcal{B}_{1}(\eta_{1:0}):=[a_{1}/l_{11},b_{1}/l_{11}]$.
Then, $p(\mathbf{a},\mathbf{b},\Sigma)$ can be written as the product
of $p_{n}\,s$ for $n=1,2,\ldots,d$. Moreover, one could see that
$p_{n}$ is also the normalising constant of the conditional distribution
of $\eta_{n}$ given $\eta_{1:n-1}$. To calculate the orthant probability,
\cite{ridgway2016computation} have proposed an SMC algorithm targetting
the sequence of distributions $\pi_{n}(\eta_{1:n}):=\pi_{1}(\eta_{n})\prod_{k=2}^{n}\pi_{n}(\eta_{k}|\eta_{1:k-1})$
for $n\in\llbracket1,d\rrbracket$, given by 
\begin{align}
 & \pi_{1}(\eta_{1})\propto\phi(\eta_{1})\mathbf{1}\{\eta_{1}\in\mathcal{B}_{1}\}=\gamma_{1}(\eta_{1})\\
 & \pi_{n}(\eta_{n}|\eta_{1:n-1})\propto\phi(\eta_{n})\mathbf{1}\{\eta_{n}\in\mathcal{B}_{n}(\eta_{1:n-1})\}=\gamma_{n}(\eta_{n}|\eta_{1:n-1}).
\end{align}
where $\phi$ denotes the probability density of a $\mathcal{N}(0,1)$.
One could also note that 
\[
\pi_{n}(\eta_{n}|\eta_{1:n-1})=1/\Phi(\mathcal{B}_{n}(\eta_{1:n-1}))\phi(\eta_{n})\mathbf{1}\{\eta_{n}\in\mathcal{B}_{n}(\eta_{1:n-1})\}
\]
 and $\gamma_{n}(\eta_{1:n})=\phi(\eta_{1:n})\prod_{k=1}^{n}\mathbf{1}\{\eta_{k}\in\mathcal{B}_{k}(\eta_{1:k-1})\}$,
where $\Phi(\mathcal{B}_{n}(\eta_{1:n-1}))$ represents the probability
of a standard Normal random variable being in the region $\mathcal{B}_{n}(\eta_{1:n-1})$.
Therefore, the SMC algorithm proposed by \cite{ridgway2016computation}
then proceeds as follows. (1) At time $t$, particles $\eta_{1:t-1}^{n}$
are extended by sampling $\eta_{n}^{n}\sim\pi_{n}({\rm d}\eta_{t}|\eta_{1:t-1}^{(i)})$.
(2) Particles $\eta_{1:t}^{(i)}$ are then reweighted by multiplying
the incremental weights $\Phi(\mathcal{B}_{n}(\eta_{1:n-1}^{(i)}))$
to $w_{n-1}^{(i)}$. (3) If the ESS is below a certain threshold,
resample the particles and move them through an MCMC kernel that leaves
$\pi_{t}$ invariant for $k$ iterations. For the MCMC kernel, \cite{ridgway2016computation}
recommended using Gibbs sampler that leaves $\pi_{t}$ invariant to
move the particles at step (3). The orthant probability we are interested
in can then be viewed as the normalising constant of $\pi_{d}$ and
this can be estimated as a by-product of the SMC algorithm.

 Since we are trying to sample from the constrained Gaussian distributions,
the Hamiltonian equation can be solved exactly and $w_{n,k}$ is always
$1$. As a result, the incremental weights for the trajectories simplify
to $\Phi(\mathcal{B}_{n}(u_{n-1}))$ and each particle on the trajectory
starting from $z_{t}$ will have an incremental weight proportional
to $\Phi(\mathcal{B}_{n}(u_{n-1}))$. To obtain a trajectory, we follow
\cite{pakman2014exact} who perform HMC with reflections to sample.
As the dimension increases, the number of reflections performed under
$\psi_{n}$ will also increase given a fixed integration time. We
adaptively tuned the integrating time $\epsilon$ to ensure that the
everage number of reflections at each SMC step does not exceed a given
threshold. In our experiment we set this threshold to be $5$. To
show that the waste-recycling RSMC algorithm scales well in high dimension,
we set $d=150$, $a=(1.5,1.5,...)$ and $b=(\infty,\infty,...)$.
Also, we use the same covariance matrix in \cite{dau2020waste} and
perform variable re-ordering as suggested in \cite{ridgway2016computation}
before the simulation. 
\begin{figure}[htb!]
\includegraphics[width=1\textwidth]{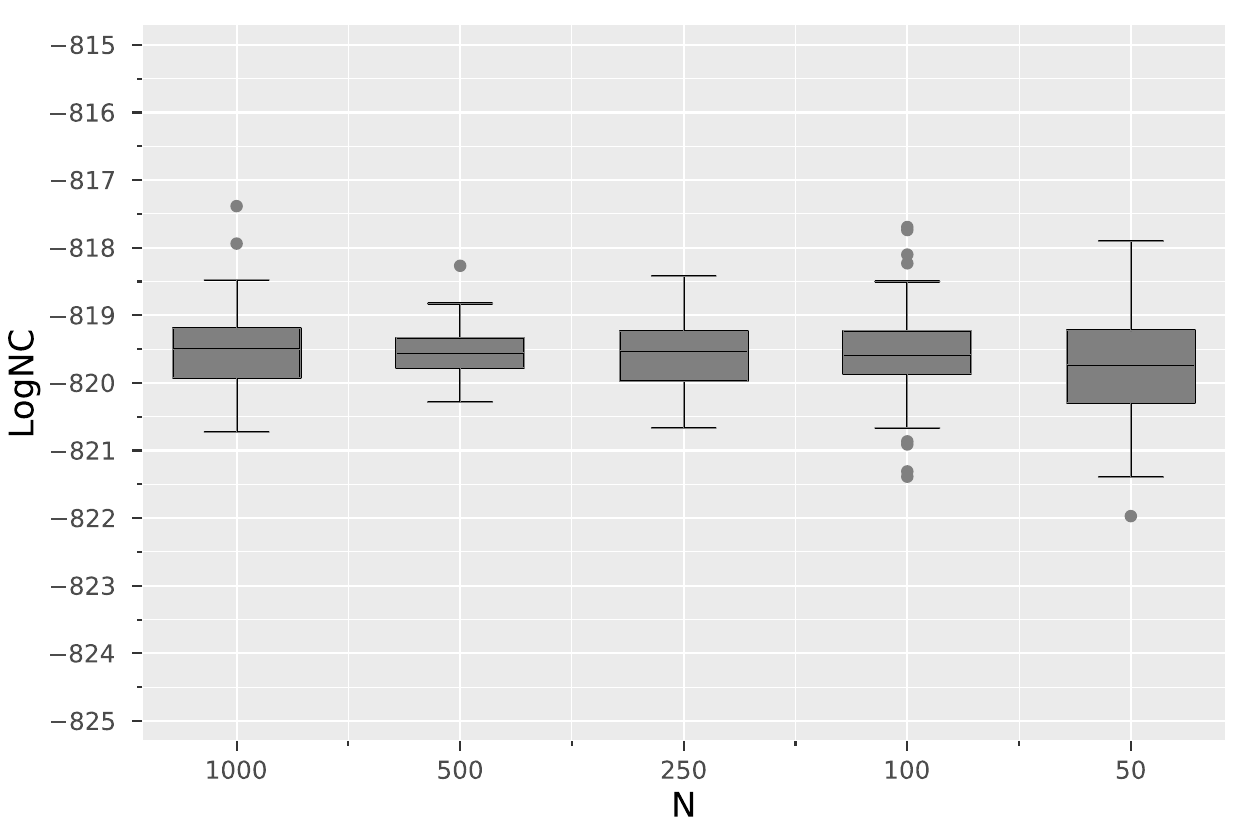} \caption{Orthant probability example: estimates of the normalising constant
(i.e. the orthant probability) obtained from the waste-recycling HSMC
algorithm with $N=50,000$.}
\label{orthant=000020lognc}
\end{figure}

\begin{figure}[htb!]
\includegraphics[width=1\textwidth]{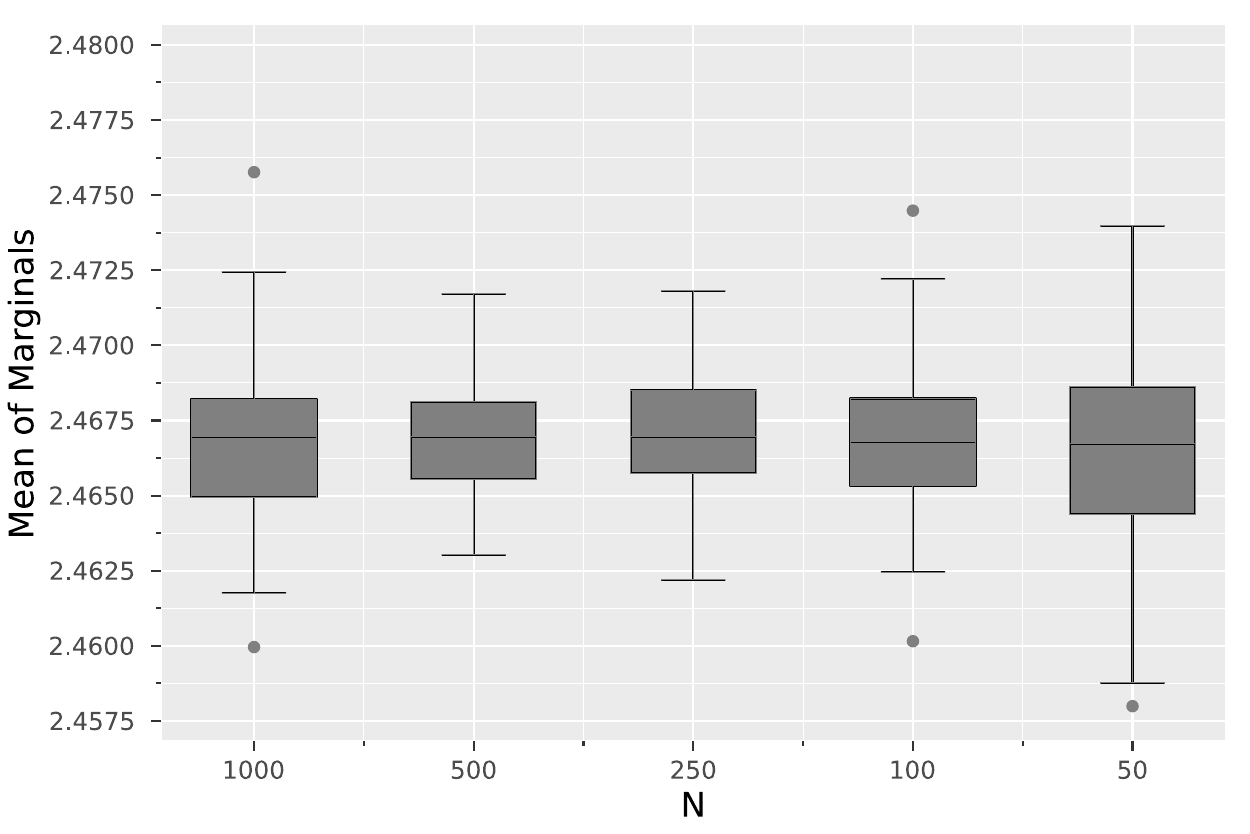} \caption{Orthant probability example: estimates of the mean of marginals obtained
from the integrator snippet SMC algorithm with $N=50,000$.}
\label{orthant=000020mm}
\end{figure}

Figures~\ref{orthant=000020lognc} and \ref{orthant=000020mm} show
the results obtained with $N\times(T+1)=50,000$ and various values
for $N$. With a quarter of the number of particles used in \cite{dau2020waste},
the waste-recycling HSMC algorithm achieves comparable performance
when estimating the normalising constant (i.e. the orthant probability).
Moreover, the estimates are stable for different choices of $N$ values,
although one observes that the algorithm achieves best performance
when $N=500$ (i.e. each trajectory contains $100$ particles). This
also suggests that the integrating time should be adaptively tuned
in a different way to achieve the best performance given a fixed computational
budget. Estimates of the function $\varphi(x_{0:d})=\mathbb{E}(1/d\sum_{i=1}^{d}x_{i})$
with respect to the Gaussian distribution $\mathcal{N}_{d}(\mathbf{0},\Sigma)$
truncated between $\mathbf{a}$ and $\mathbf{b}$ are also stable
for different choices of $N$, although they are more variable than
those obtained in \cite{dau2020waste}. This indicates that the waste-recycling
HSMC algorithm does scale well in high dimension. We note that this
higher variance compared to the waste-free SMC of \cite{dau2020waste},
is obtained in a scenario where they are able to exploit the particular
structure of the problem and implement an exact Gibbs sampler to move
the particles. The integrator snippet SMC we propose is however more
general and applicable to scenarios where such a structure is not
present.

\section{Early experimental results: filamentary distributions \protect\label{sec:Early-experimental-results-dilamentary}}

\subsection{Numerical illustration: simulating from filamentary distributions
\protect\label{subsec:example-filamentary-distributions}}

We now illustrate the interest of integrator snippets in a scenario
where the target distribution possesses specific geometric features.
Specifically, we focus here on distributions concentrated around a
manifold $\mathcal{M}\subset\mathsf{X}=\mathbb{R}^{d}$ defined as
the zero level set of a smooth function $\ell:\mathbb{R}^{d}\to\mathbb{R}^{m}$
\[
\mathcal{M}:=\left\{ x\in\mathsf{X}:\ell(x)=0\right\} ,
\]
sometimes referred to as filamentary distributions \cite{chang1997conditioning,maire2018markov,maire2022markov}.
Such distributions arise naturally in various scenarios, including
inverse problems or generalisations of the Gibbs sampler \cite{chang1997conditioning,maire2018markov,maire2022markov}
or as a relaxation of problems where the support of the distribution
of interest is included in $\mathcal{M}$ or for example generalisation
of the Gibbs sampler through disintegration \cite{chang1997conditioning}.
Such is the case for ABC methods in Bayesian statistics \cite{beaumont2010approximate}.
Assume that $\pi$ is a probability density with respect to the Lebesgue
measure defined on $\mathsf{X}$ and for $\epsilon>0$ consider a
``kernel'' function $k_{\epsilon}(u):=\epsilon^{-m}k(u/\epsilon)$
where $k:\mathbb{R}^{m}\to\mathbb{R}_{+}$ and define the probability
density 
\[
\pi_{\epsilon}(x)\propto k_{\epsilon}\circ\ell(x)\,\pi(x)\,.
\]
The corresponding probability distribution can be thought of as an
approximation of the probability distribution of density $\pi_{0}(x)\propto\pi(x)\mathbf{1}\{x\in\mathcal{M}\}$
with respect to the Hausdorff measure on $\mathcal{M}$. Typical choices
for the kernel are $k(u)=\mathbf{1}\{\|u\|\leq1\}$ or $k(u)=\mathcal{N}(u;0,\mathbf{I}_{m})$.
Strong anisotropy may result from such constraints and make exploration
of the support of such distributions awkward for standard Monte Carlo
algorithms. This is illustrated in Fig.~\ref{fig:hmc_fails_on_filamentary}
where a standard MH-HMC algorithms is used to sample from $\pi_{\epsilon}$
defined on $\mathbb{R}^{2}$, for three values $\epsilon=0.5,0.1,0.05$,
and performance is observed to deteriorate as $\epsilon$ decreases.
The samples produced are displayed in blue: for $\epsilon=0.5$ HMC-MH
mixes well and explores the support rapidly, but for $\epsilon=0.1$
the chain gets stuck in a narrow region of $\pi_{\epsilon}$ at the
bottom, near initialisation, while for $\epsilon=0.05$, no proposal
is ever accepted in the experiment. 

\begin{figure}
\begin{centering}
\includegraphics[width=1\textwidth]{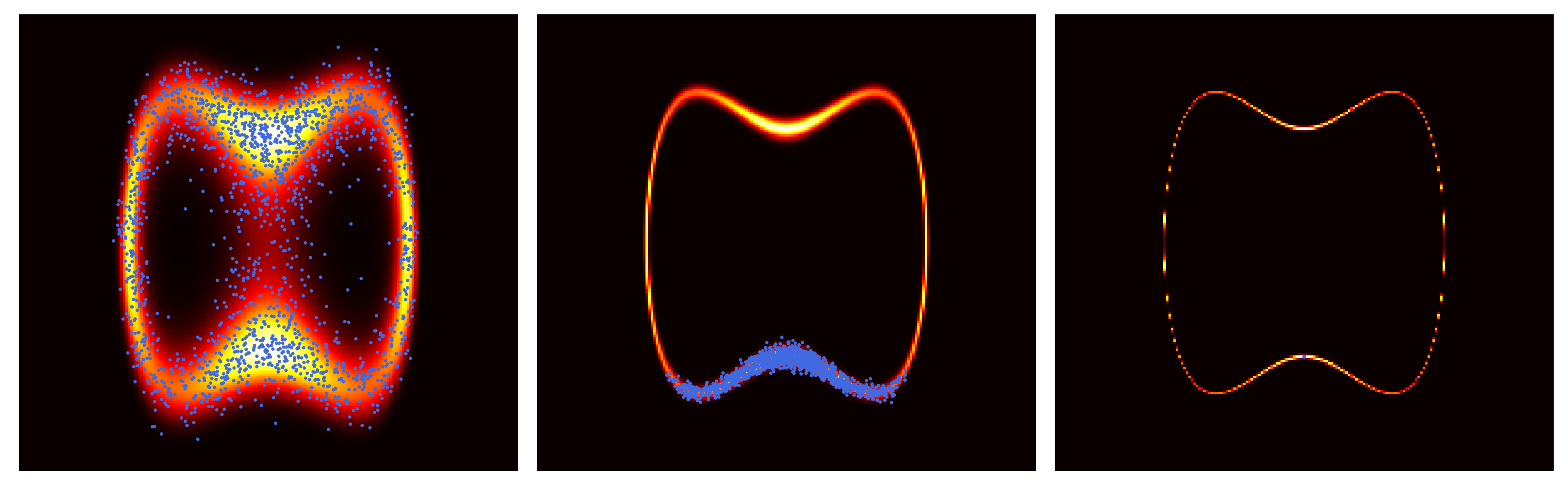}\caption{Heat map of $\pi_{\epsilon}(x)\propto\mathcal{N}(y\mid F(x),\epsilon^{2})\mathcal{N}(x\mid0,\mathbf{I}_{2})$
with $F(x_{1},x_{2}):=x_{2}^{2}+x_{1}^{2}\left(x_{1}^{2}-\frac{1}{2}\right)$,
defined on $\mathbb{R}^{2}$, for $y\in\mathbb{R}$ and $\epsilon=0.5,0.1,0.05$
(left to right) superimposed with $2000$ samples (in blue) obtained
from an MH-HMC with step size $0.05$ and involving $T=20$ leapfrog
steps.}
\par\end{centering}
\label{fig:hmc_fails_on_filamentary}
\end{figure}
To illustrate the properties of integrator snippet techniques we consider
the following toy example. For $\epsilon>0$ and $d\in\mathbb{N}_{*}$
let 
\[
\pi_{\epsilon}(x)\propto\frac{1}{\epsilon^{m}}\mathbf{1}\{\left\Vert \ell(x)\right\Vert \leq\epsilon\}\,\mathcal{N}(x;0,\mathrm{\mathbf{I}}_{d})\quad\text{with}\quad\ell(x)=x^{\top}\Sigma^{-1}x-c\,,
\]
for $\Sigma$ a $d\times d$ symmetric positive matrix and $c\in\mathbb{R}$,
that is we consider the restriction of a standard normal distribution
around an ellipsoid defined by $\ell(x)=0$.

\newcommand{\parallelsum}{\mathbin{\scalebox{0.4}{$\!/\mkern-5mu/\!$}}}

\newcommand*{\para}{\rotatebox[origin=c]{-45}{$\scalebox{0.8}{$\shortparallel$}$}}

\global\long\def\smallparallel{\parallelsum}%

\newcommand*{\ortho}{\rotatebox[origin=c]{180}{$\scalebox{0.45}{$\top$}$}}

\global\long\def\smallperp{\ortho}%

In order to explore the support of the target distribution we use
two mechanisms based on reflections either through tangent hyperplanes
of equicontours of $\ell(x)$ or through the corresponding orthogonal
complements. More specifically for $x\in\mathsf{X}$ such that $\nabla\ell(x)\neq0$
let $n(x):=\nabla\ell(x)/\|\nabla\ell(x)\|$ and define the tangential
HUG (THUG) and symmetrically-normal HUG (SNUG) as $\psi_{\smallparallel}:={}_{\textsc{a}}\psi\circ{\rm b}\circ{}_{\textsc{a}}\psi$,
, and $\psi_{\smallperp}:={}_{\textsc{a}}\psi\circ(-{\rm b})\circ{}_{\textsc{a}}\psi$
respectively, with $_{\textsc{a}}\psi$ and $\mathrm{b}$ as in Example~\ref{exa:verlet-HMC}
and (\ref{eq:def-bounce}). Both updates can be understood as discretisations
of ODEs and are volume preserving. Intuitively for $(x,v)\in\mathsf{Z}$
with $v_{\smallparallel}=v_{\smallparallel}(x+\epsilon v):=n(x+\epsilon v)n(x+\epsilon v)^{\top}v$
for $\epsilon>0$ and $v_{\smallperp}=v-v_{\smallparallel}$ we have
$\psi_{\smallparallel}(x,v)=(x+2\epsilon v_{\smallparallel},v_{\smallparallel}-v_{\smallperp})$
and $\psi_{\smallperp}(x,v)=(x+2\epsilon v_{\smallperp},v_{\smallperp}-v_{\smallparallel})$.
This is illustrated in Fig.~\ref{fig:fig:thug_snug} for $d=2$.
Further, for an initial state $z_{0}=(x_{0},v_{0})\in\mathsf{Z}$,
trajectories of the first component of $k\mapsto\psi_{\smallparallel}^{k}(z_{0})$
remain close to $\ell^{-1}(\{\ell(x_{0})\})$ while $k\mapsto\psi_{\smallperp}^{k}(z_{0})$
follows the gradient field $x\mapsto\nabla\ell(x)$ and leads to hops
across equicontours. 

\begin{figure}
\begin{centering}
\includegraphics[width=1\textwidth]{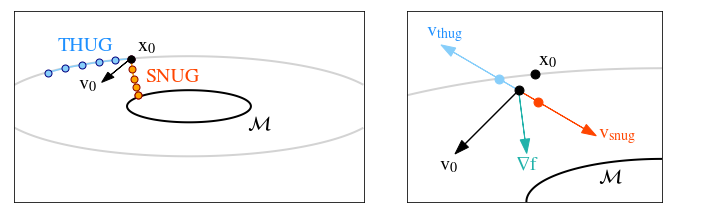}
\par\end{centering}
\caption{Black line: manifold of interest $\mathcal{M}=\ell^{-1}(0)$, grey
line: level set of $\ell^{-1}(\{\ell(x_{0})\})$ the $\ell$-level
set $x_{0}$ belongs to. Left: THUG trajectory (blue) remains close
to $\ell^{-1}(\{\ell(x_{0})\})$, SNUG trajectory (orange) explores
different contours of $\ell(x)$. Right: $v_{{\rm thug}}=v_{\protect\smallparallel}-v_{\protect\smallperp}$
and $v_{{\rm snug}}=-(v_{\protect\smallparallel}-v_{\protect\smallperp})$.}

\label{fig:fig:thug_snug}
\end{figure}
The sequence of target distributions on $(\mathsf{Z},\mathcal{Z})$
we consider is of the form
\[
\mu_{n}({\rm d}z)\propto\pi_{\epsilon_{n}}({\rm d}x)\mathcal{N}({\rm d}v;0,\mathbf{I}_{d}),\qquad\qquad\epsilon_{n}>0,n\in\llbracket0,P\rrbracket\,,
\]
and the Integrator Snippet SMC is defined through the mixture, for
$\alpha\in[0,1]$, 
\[
\bar{\mu}_{n}({\rm d}z)=\frac{\alpha}{T+1}\sum_{k=0}^{T}\mu_{n}^{\psi_{\smallparallel}^{-k}}({\rm d}z)+\frac{1-\alpha}{T+1}\sum_{k=0}^{T}\mu_{n}^{\psi_{\smallperp}^{-k}}({\rm d}z)\,.
\]
We compare performance of integrator snippet with an SMC Sampler relying
on a mutation kernel $M_{n}$ consisting of a mixture of two updates
targetting $\mu_{n-1}$ each iterating $T$ times a MH kernel applying
integrator $\psi_{\shortparallel}$ or $\psi_{\smallperp}$ once after
refreshing the velocity; the backward kernels are chosen to correspond
to the default choices discussed earlier. In the experiments below
we set $d=50$, $c=12$, and $\Sigma$ is the diagonal matrix alternating
$1$'s and $0.1$'s along the diagonal. We used $N=5000$ particles,
sampled from $\mathcal{N}(0,\mathrm{\mathbf{I}}_{d})$ at time zero
and $\epsilon_{0}$ is set to the maximum distance of these particles
in the $\ell$ domain from $0$ and $\alpha=0.8$ for both algorithms.
We compared the two samplers across three metrics; all results are
averaged over $20$ runs.

\textbf{Robustness and Accuracy}: we fix the step size for SNUG to
$\epsilon_{\smallperp}=0.1$ and run both algorithms for a grid of
values of $T$ and $\epsilon_{\smallparallel}$ using a standard adaptive
scheme based on ESS to determine $\{\epsilon_{n},n\in\llbracket P\rrbracket\}$
until a criterion described below is satisfied and retain the final
tolerance achieved, recorded in Fig.~\ref{fig:eps-grid}. As a result
the terminal value $\epsilon_{P}$ and computational costs are different
for both algorithms: the point of this experiment is only to demonstrate
robustness and potential accuracy of Integrator Snippets. Both the
SMC sampler and Integrator Snippet are stopped when the average probability
of leaving the seed particle drops below $0.01$. 

 Our proposed algorithm consistently achieves a two order magnitude
smaller final value of $\epsilon$ and is more robust to the choice
of stepsize.

\begin{figure}
\begin{centering}
\includegraphics[width=1\textwidth]{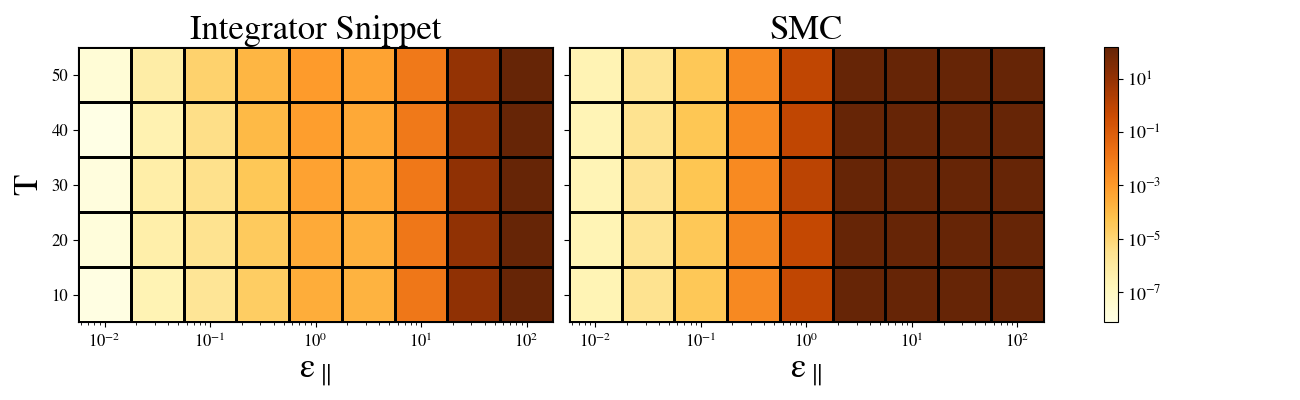}
\par\end{centering}
\caption{Final tolerances achieved by the two algorithms, averaged over $20$
runs.}
\label{fig:eps-grid}
\end{figure}

\textbf{Variance}: for this experiment we set $T=50$ steps, $\epsilon_{\smallparallel}=0.01$
and determine and compare the variances of the estimates of the mean
of $\pi_{\epsilon_{P}}$ for the final SMC step, for which $\epsilon_{P}=1\times10^{-7}$.
To improve comparison, and in particular ensure comparable computational
costs, both algorithms share the same schedule $\{\epsilon_{n},n\in\llbracket P\rrbracket\}$,
determined adaptively by the SMC algorithm in a separate pre-run.

The results are reported as componentwise boxplots in Fig.~\ref{fig:var-box-plot}
where we observe a significant variance reduction for comparable computational
cost.

\begin{figure}
\begin{centering}
\includegraphics[width=0.8\textwidth]{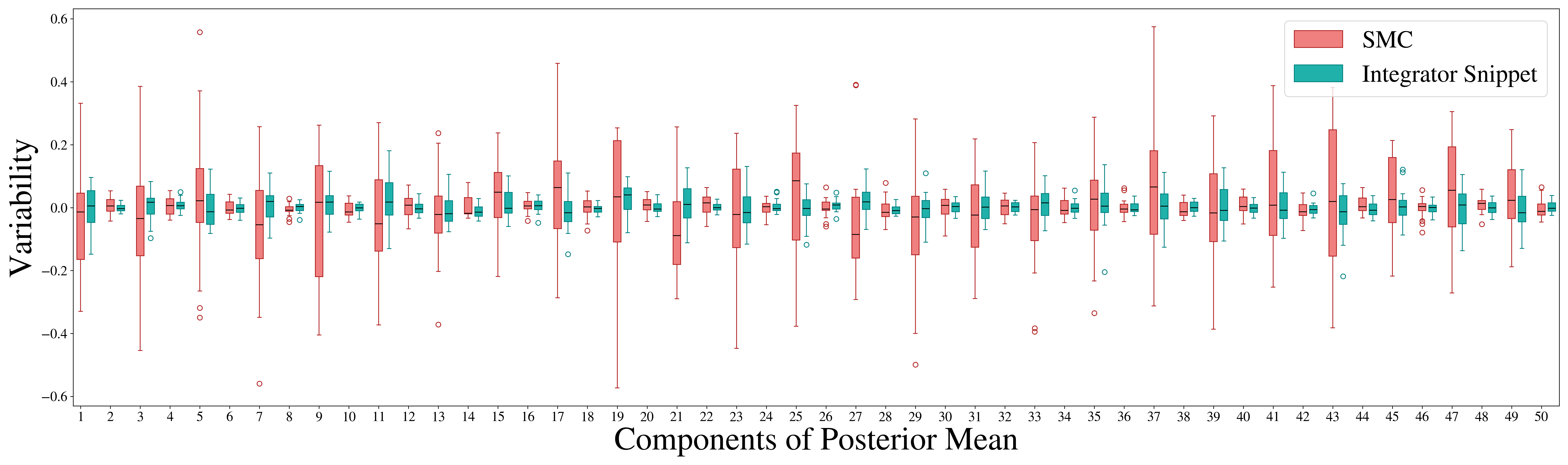}\caption{Variability of the estimates of the mean of $\pi_{\epsilon}$ for
the SMC sampler and the Integrator Snippet, using estimator $\hat{\mu}_{P}(f)$,
for the test function $f(x)=x$. Integrator snippet is able uniformly
achieve smaller variance, in particular on the odd components corresponding
to larger variances in $\Sigma$.}
\par\end{centering}
\label{fig:var-box-plot}
\end{figure}

\textbf{Efficiency}: we report the Expected Squared Jump Distance
(ESJD) as a proxy of distance travelled by the two algorithms. For
Integrator Snippets, it is possible to estimate this quantity as follows
for a function $f:\mathsf{Z}\to\mathbb{R}$
\[
\text{ESJD}_{n}(f)\approx\sum_{i=1}^{N}\sum_{k=0}^{T}\sum_{l=k+1}^{T}\left(f(Z_{n-1,l}^{(i)})\bar{W}_{n,l}^{(i)}-f(Z_{n-1,k}^{(i)})\bar{W}_{n,k}^{(i)}\right)^{2},\qquad\bar{W}_{n,k}^{(i)}=\frac{\bar{w}_{n,k}(Z_{n-1}^{(i)})}{{\displaystyle \sum_{l=0}^{T}\bar{w}_{n,l}(Z_{n-1,l}^{(i)})}}\,.
\]

We report the average of this metric in Table~\ref{table:esjd_table}
for the functions $f_{i}(x)=x_{i}$ $i\in\llbracket d\rrbracket$,
normalised by total runtime in seconds for all particles (first row),
for the particles using the THUG update (second row) and those using
the SNUG update (third row), with standard deviation shown in parenthesis.
Our proposed algorithm is several orders of magnitude more efficient
in the exploration of $\pi_{\epsilon}$ than its SMC counterpart which,
thanks to its ability to take full advantage of all the intermediary
states of the snippet. This is in contrast with the SMC sampler which
creates trajectories of random length. 

\begin{table}
\centering %
\begin{tabular}{c|c|c}
 & Integrator Snippet & SMC\tabularnewline
\hline 
 &  & \tabularnewline
$\text{ESJD}/s$ & \,\,$\mathbf{5.3\times10^{-5}}\,\,(\pm5.9\times10^{-6})$ \,\, & \,\,$1.2\times10^{-7}\,\,(\pm3.7\times10^{-9})$ \,\,\tabularnewline
$\text{ESJD-THUG}/s$ & \,\,$\mathbf{6.6\times10^{-5}}\,\,(\pm7.5\times10^{-6})$ \,\, & \,\,$1.7\times10^{-7}\,\,(\pm4.5\times10^{-9})$ \,\,\tabularnewline
$\text{ESJD-SNUG}/s$ & \,\,$\mathbf{2.7\times10^{-4}}\,\,(\pm3.2\times10^{-5})$ \,\, & \,\,$1.6\times10^{-20}\,\,(\pm6.0\times10^{-21})$ \,\,\tabularnewline
\end{tabular}\caption{$d^{-1}\sum_{i=1}^{d}\text{ESJD}_{n}(f_{i})$ normalised by time for
Integrator Snippet and an SMC sampler.}
\label{table:esjd_table}
\end{table}

In Fig.~\ref{fig:hug-scatter-plots-support} we illustrate robustness
of integrator snippet SMC to the choice of $\epsilon$ and $T$ through
the support covered; note that the computational cost of the three
algorithms considered is comparable.

\begin{figure}
\centering
\includegraphics[width=0.9\textwidth]{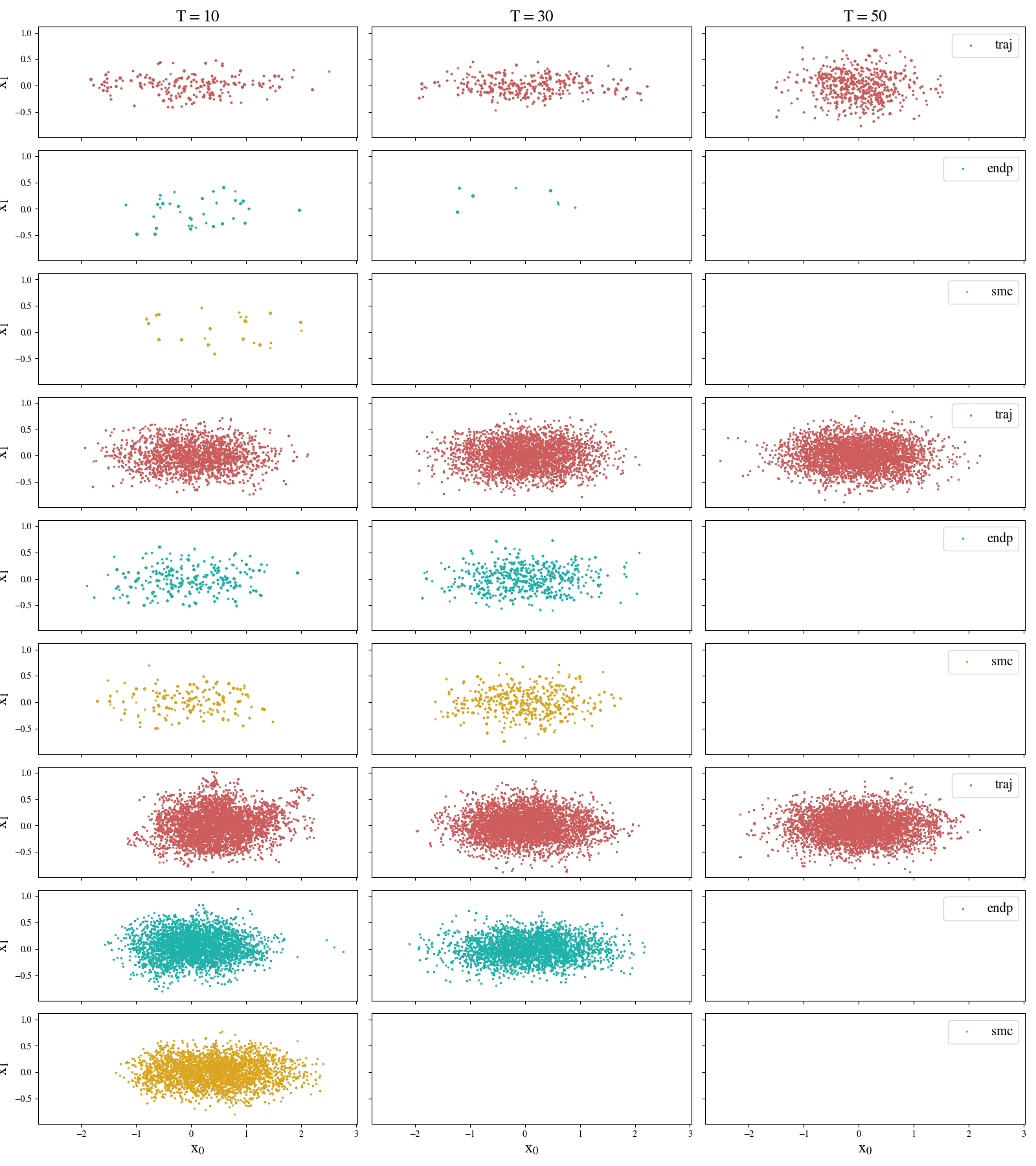}

\caption{Scatter plot of the first two coordinates of $\{x_{n,k}^{(i)},i\in\llbracket N\rrbracket,k\in\llbracket0,T\rrbracket\}$
for three SMC algorithms: Snippet SMC (traj), Snippet SMC using the
snippets endpoints only (endp) and the SMC sampler (smc). The stepsize
is decreased from a high value for rows 1-3 $(\epsilon=1.0)$, to
lower values for rows 4-6 ($\epsilon=0.1$) and rows 7-9 $(\epsilon=0.01)$.
The three columns correspond to different snippet length, $T=10,30,50$. }

\label{fig:hug-scatter-plots-support}
\end{figure}

\subsection{Going round the bend: bananas and pears}

In this section we illustrate how ODEs and associated integrators
can be used in order to update states involved in simulation algorithm
for probability distributions involving geometrically awkward densities.
The banana shaped distribution introduced by \cite{haario1999adaptive}
in the context of adaptive MCMC algorithm is a prototypical toy example
of dependencies encountered in some practical situation e.g. inverse
problems. The banana shaped distribution defined on $\mathsf{X}=\mathbb{R}^{d},$
has density with respect to the Lebesgue measure $\pi(x)\propto\mathcal{N}\big(\phi(x);0,\text{\ensuremath{\Sigma}}\big)$
where $\phi(x)=(x_{1},x_{2}+0.03(x_{1}^{2}-100),x_{2},\ldots,x_{n})$
and $\Sigma={\rm diag}(100,1,\ldots,1)$ -- see its contours for
$d=2$ in Fig.~\ref{fig:banana}. A first suggestion could be to
follow the contours of $\pi$ using the update described in the previous
section. Here we suggest the following flow, for $v\in\{-1,1\}$,
\begin{align*}
t & \mapsto\begin{cases}
x_{1}(t)=vt\\
x_{2}(t)=-0.03(x_{1}^{2}(t)-100)+x_{2}(0)+0.03(x_{1}^{2}(0)-100)\\
x_{i}(t)=x_{i}(0) & \text{for }i>2
\end{cases}
\end{align*}
Motivation for this is that $t\mapsto x_{2}(t)+0.03(x_{1}^{2}(t)-100)=x_{2}(0)+0.03(x_{1}^{2}(0)-100)$
is a constant, effectively unbending the banana and effectively applying
a standard Gibbs sampler \cite{doi:10.1137/1116083} on a Gaussian
in a different set of coordinates. This suggests a practical implementation
of generalized forms of the Gibbs sampler, which can be traced back
to \cite{https://doi.org/10.1111/1467-9574.00056} and were recently
revived in \cite{qin2023spectral}. Naturally in practice such updates
must be completed with other updates to ensure irreducibility.

In more general situations where $\phi(x)=(x_{1},u(x_{1},x_{2}),x_{3},\ldots)$
for some $u\colon\mathsf{X}^{2}\rightarrow\mathbb{R}$ such that $\pi$
is indeed a probability density, the idea can be generalized with
the ideal aim of keeping $t\mapsto u\big(x_{1}(t),x_{2}(t)\big)$
or $t\mapsto x_{1}(t)^{2}+u\big(x_{1}(t),x_{2}(t)\big)^{2}$ constant,
either using THUG or even Hamilton's equation as suggested by the
last contraint. Note however that non-separability may cause additional
numerical difficulties in the latter scenario. 

\begin{figure}
\centering{}\includegraphics[width=0.7\textwidth]{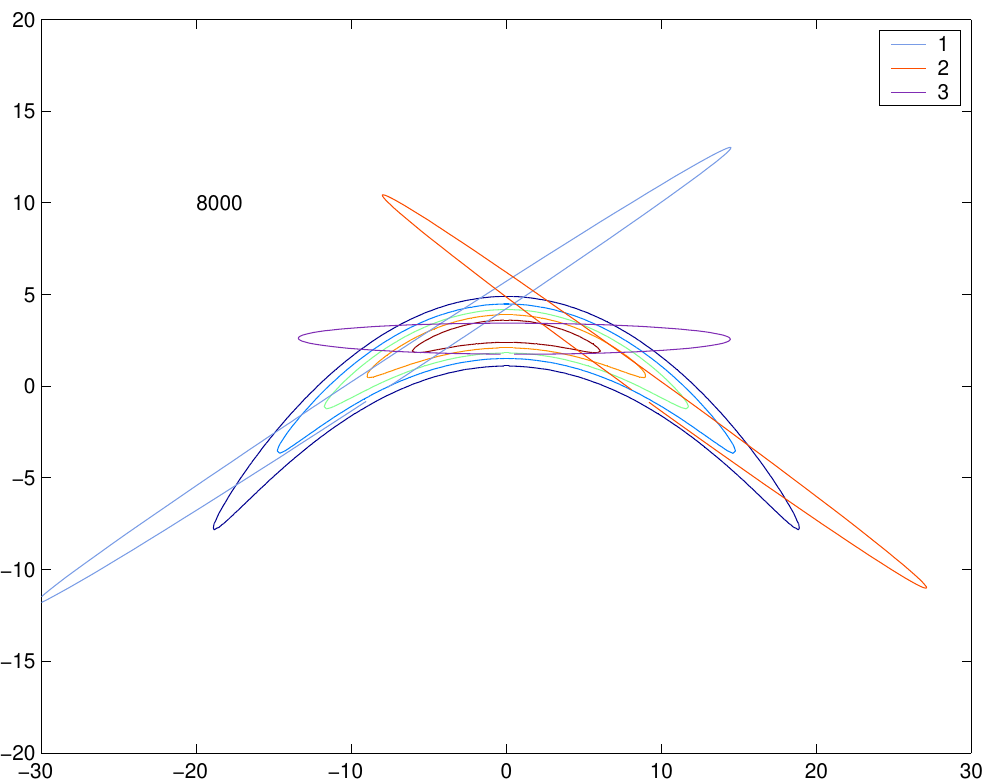}\caption{The banana distribution, $d=2$}
\label{fig:banana}
\end{figure}
Another popular illustrative toy example is Neal's funnel on which
numerous standard algorithms may struggle to converge. The funnel
consists of a mixture of normal distributions
\[
\pi(x,y)=\mathcal{N}(y;0,3)\mathcal{N}(x;0,\exp(y/2))\,,
\]
defining a probability density on $\mathbb{R}^{2}$ of contours shown
in Fig.~\ref{fig:funnel}--multiple $x$'s can be incorporated therefore
taking the form of a hierchical model. It exemplifies some of the
difficulties which may be encountered when sampling from hierarchical
models.

\begin{figure}
\centering{}\includegraphics[width=0.7\textwidth]{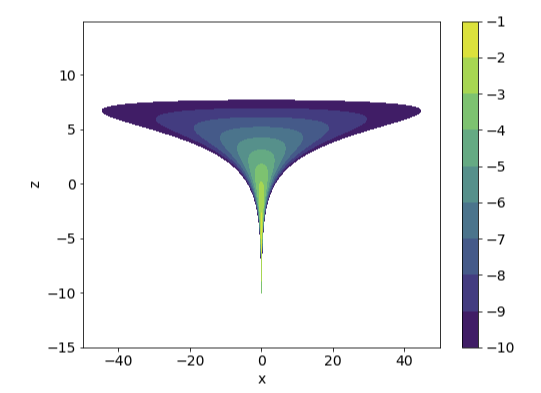}\caption{Neal's funnel}
\label{fig:funnel}
\end{figure}
The narrow funnel is the cause of failure of traditional algorithms.
A simple idea here consists of identifying a flow such that $t\mapsto\mathcal{N}(x_{t};0,\exp(y_{t}/2))\equiv\mathcal{N}(x_{0};0,\exp(y_{0}/2))$
for given $(x_{0},y_{0})$. A straightforward choice imposes $\dot{y}_{t}=v$
with $v\in\{-1,1\}$ and $x_{t}/\exp(y_{t}/4)=x_{0}/\exp(y_{0}/4)$
for $t>0$, therefore ensuring that $x_{t}\sim\mathcal{N}(0,\exp(y_{t}/2))$
for $t>0$. This can be solved analytically
\begin{align*}
x_{t} & =x_{0}\exp\left(y_{0}/4+y_{t}/4\right)=x_{0}\exp\left((y_{0}+vt)/4\right)\\
y_{t} & =vt\,.
\end{align*}
This corresponds to a simple reparametrizations, standard in the literature.
However our suggested use differs in that both components are updated
simultaneously within an MCMC kernel, and illustrate the potential
of this type of ideas. Such a dynamic should naturally be mixed with
other updates ensuring irreducibility.

\section{Adaptation: implementational details \protect\label{app-sec:details-imp-adapt-epsilon}}

\subsection{About the structure of the mother distribution \protect\label{appsubsec:structure-mother-dist}}

It is useful to provide details about the (hierchical) structure of
the target distributions useful for implementational purposes. The
mother distribution in (\ref{eq:target-mother-adaptive}) is given
by 
\[
\bar{\mu}\big(k,{\rm d}(\epsilon,z);\gamma,\theta\big):=\frac{1}{T+1}\mu^{\psi_{\epsilon}^{-k}}({\rm d}z;\gamma)\nu_{\theta}({\rm d}\epsilon)
\]
from which we notice that
\begin{align*}
\bar{\mu}\big(k,{\rm d}(\epsilon,z);\theta,\gamma\big) & =\bar{\mu}\big(k,{\rm d}z\mid\epsilon,\gamma\big)\nu_{\theta}({\rm d}\epsilon)\\
 & =\bar{\mu}\big(k\mid\epsilon,z;\gamma\big)\bar{\mu}\big({\rm d}z\mid\epsilon;\gamma\big)\nu_{\theta}({\rm d}\epsilon)
\end{align*}
that is the first conditionals do not depend on $\theta\in\Theta$.
Indeed, marginalising yields $\bar{\mu}\big({\rm d}\epsilon,z;\theta,\gamma\big)=\nu_{\theta}({\rm d}\epsilon)$,
implying

\begin{align*}
\bar{\mu}\big(k,{\rm d}z\mid\epsilon;\gamma,\theta\big) & =\frac{1}{T+1}\mu^{\psi_{\epsilon}^{-k}}({\rm d}z;\gamma)
\end{align*}
and justifying the notational simplifications, leading to
\begin{align*}
\bar{\mu}\big({\rm d}z\mid\epsilon;\gamma,\theta\big)=\bar{\mu}\big({\rm d}z\mid\epsilon;\gamma\big)=\frac{1}{T+1} & \sum_{l=0}^{T}\mu^{\psi_{\epsilon}^{-l}}({\rm d}z;\gamma)\\
\bar{\mu}\big(k\mid\epsilon,z;\gamma,\theta\big)=\bar{\mu}\big(k\mid\epsilon,z;\gamma\big)= & \frac{{\rm d}\mu^{\psi_{\epsilon}^{-k}}}{{\rm d}\big(\sum_{l=0}^{T}\mu^{\psi_{\epsilon}^{-l}}\big)}(z;\gamma)\,.
\end{align*}
With $\bar{f}=\bar{f}(\epsilon,K,Z)$ for $\bar{f}(\epsilon,k,z):=f\circ\psi_{\epsilon}^{k}(z)$
for $f\colon\mathsf{Z}\rightarrow\mathbb{R}^{m}$ for $m\in\mathbb{N}$,
\[
\bar{\mu}\big(\bar{f}\mid\epsilon,z;\gamma\big):=\sum_{k=0}^{T}\bar{\mu}\big(k\mid\epsilon,z;\gamma\big)f\circ\psi_{\epsilon}^{k}(z)\,.
\]

\subsection{Kullback-Leibler minimisation}

To minimize the KL divergence to determine $\theta_{n}$ we first
note that, with $C$ independent of $\theta$ and possibly changing
on each appearance,
\begin{align*}
\theta\mapsto{\rm {\rm KL}}\big(\tilde{\mu}(\cdot;\gamma_{n},\theta_{n}),\bar{\mu}(\cdot\mid\gamma_{n},\theta)\big)\\
=C+\int\log & \left(\frac{v_{\gamma_{n-1}}(\epsilon,z)\nu_{\theta_{n-1}}({\rm d}\epsilon)\bar{\mu}\big({\rm d}z\mid\epsilon;\gamma_{n1}\big)}{\nu_{\theta}({\rm d}\epsilon)\bar{\mu}\big({\rm d}z\mid\epsilon;\gamma_{n}\big)}\right)\tilde{\mu}\big({\rm d}(\epsilon,z);\gamma_{n},\theta_{n}\big)\\
=C-\int\log & \left(\nu_{\theta}(\epsilon)\right)\tilde{\mu}\big({\rm d}(\epsilon,z);\gamma_{n},\theta_{n}\big)\\
=C-\int\log & \left(\nu_{\theta}(\epsilon)\right)\,\tilde{\mu}\big({\rm d}\epsilon;\gamma_{n},\theta_{n}\big)\,.
\end{align*}
We remark that minimization of the above is easily achieved in scenarios
where $\nu_{\theta}$ is chosen in the exponential family, e.g. $\nu_{\theta}$
has density of the form $\nu_{\theta}(\epsilon)=\exp\big(\langle\theta,T(\epsilon)\rangle-A(\theta)\big)$,
in which case minimisation of (\ref{eq:def-chi-theta}) amounts to
solving 
\[
\theta\mapsto\mathbb{E}_{\nu_{\theta}}\left(T(\epsilon)\right)=\mathbb{E}_{\tilde{\mu}\big(\cdot;\gamma_{n},\theta_{n}\big)}\left(T(\epsilon)\right)\,.
\]
In practice the expectation with respect to $\tilde{\mu}\big(\cdot\gamma_{n},\theta_{n}\big)$
is obtained using self-normalised importance sampling of the seed
particles $\{(\check{z}_{n}^{(i)},\check{\epsilon}_{n}^{(i)}),i\in\llbracket N\rrbracket\}$
representing $\bar{\mu}_{n}$ or weighted samples $\{(z_{n-1}^{(i)},\epsilon_{n-1}^{(i)}),i\in\llbracket N\rrbracket\}$
(see Alg.~(\ref{alg:Folded-PDMP-SMC-1}) and surrounding comments),
leading in the exponential family scenario to solving
\[
\theta\mapsto\mathbb{E}_{\nu_{\theta}}\left(T(\epsilon)\right)=\frac{\sum_{i=1}^{N}T(\check{\epsilon}_{n}^{(i)})v(\check{\epsilon}_{n}^{(i)}),\check{z}_{n}^{(i)})}{\sum_{i=1}^{N}v(\check{\epsilon}_{n}^{(i)},\check{z}_{n}^{(i)})}\,.
\]
In the general scenario minimization is equivalent to maximum likehood
estimation with samples $\{(\check{z}_{n}^{(i)},\check{\epsilon}_{n}^{(i)}),i\in\llbracket N\rrbracket\}$
as data.

Now we derive a more expression for $v_{\gamma}(\epsilon,z):=\mathrm{Tr}\big[{\rm var}\big(\bar{f}\mid\epsilon,z;\gamma\big)\big]$
useful for implementation.
\begin{align}
v_{\gamma}(\epsilon,z) & =\mathrm{Tr}\{\mathbb{E}\big[\bar{(f}-\mathbb{E}(\bar{f}\mid\epsilon,z))\bar{(f}-\mathbb{E}(\bar{f}\mid\epsilon,z))^{\top}\mid\epsilon,z;\gamma\big]\}\nonumber \\
 & =\mathbb{E}\{{\rm Tr}\big[\bar{(f}-\mathbb{E}(\bar{f}\mid\epsilon,z))\bar{(f}-\mathbb{E}(\bar{f}\mid\epsilon,z))^{\top}\big]\mid\epsilon,z;\gamma\}\nonumber \\
 & =\mathbb{E}\{\|\bar{f}-\mathbb{E}(\bar{f}\mid\epsilon,z)\|^{2}\mid\epsilon,z;\gamma\}\nonumber \\
 & =\frac{1}{T+1}\sum_{k=0}^{T}\Bigl\Vert f\circ\psi_{\epsilon}^{k}(z)-\bar{\mu}(f\mid\epsilon,z;\gamma,\theta)\Bigr\Vert^{2}\bar{\mu}\big(k\mid\epsilon,z;\gamma,\theta\big)\,.\label{eq:explicit-express-v_epsilon}
\end{align}

\subsection{Calculations in the inverse Gaussian scenario\protect\label{app:inverse-gaussian-calculations}}

The Inverse Gaussian with mean $\theta>0$ and shape parameter $\lambda>0$
has skewness $s=3\sqrt{\theta/\lambda}$. Therefore for a fixed skewness
$s>0$ and mean parameter $\theta>0$ the IG distribution has density
\[
\nu_{\theta}(\epsilon;s)=\sqrt{\frac{9\theta}{2\pi\epsilon^{3}s^{2}}}\exp\left(-\frac{9\theta(\epsilon-\theta)^{2}}{2\theta^{2}\epsilon s^{2}}\right)\,,
\]
 w.r.t. the Lebesgue measure. We aim to minimise the Kullback-Leibler
divergence ${\rm \theta\mapsto KL}(\nu,\nu_{\theta})$ for a given
probability distribution. The corresponding log-likelihood is
\begin{align*}
\log\nu_{\theta}(\epsilon;s)=\log(3)+\frac{1}{2}\log(\theta)-\frac{1}{2}\log(2\pi)-\frac{3}{2}\log(\epsilon)-\log(s)-\frac{9(\epsilon^{2}+\theta^{2}-2\epsilon\theta)}{2\theta\epsilon s^{2}}\,,
\end{align*}
and its derivative w.r.t. $\theta>0$ is
\begin{align*}
\frac{{\rm d}}{{\rm d}\theta}\log\nu_{\theta}(\epsilon;s) & =\frac{1}{2\theta}-\frac{9}{2\epsilon s^{2}}\left[\frac{{\rm d}}{{\rm d}\theta}\left(\frac{\epsilon^{2}}{\theta}+\theta-2\epsilon\right)\right]\\
 & =\frac{1}{2\theta}-\frac{9}{2\epsilon s^{2}}\left[-\frac{\epsilon^{2}}{\theta^{2}}+1\right]\\
 & =\frac{1}{2\theta}-\frac{9(\theta^{2}-\epsilon^{2})}{2\epsilon s^{2}\theta^{2}}\,.
\end{align*}
Therefore
\[
\frac{{\rm d}}{{\rm d}\theta}{\rm KL}(\nu,\nu_{\theta})=\frac{1}{2\theta^{2}}\frac{9}{s^{2}}\left\{ \frac{s^{2}}{9}\theta-\theta^{2}\mathbb{E}_{\nu}[\epsilon^{-1}]+\mathbb{E}_{\nu}[\epsilon]\right\} =0\,,
\]
which has unique positive solution

\[
\hat{\theta}=\frac{{\displaystyle \frac{s^{2}}{9}+\sqrt{\frac{s^{4}}{81}+4\mathbb{E}_{\nu}[\epsilon^{-1}]\mathbb{E}_{\nu}[\epsilon]}}}{2\mathbb{E}_{\nu}[\epsilon^{-1}]}
\]
Using that $\mathbb{E}_{\nu}[\epsilon^{-1}]=\mathbb{E}_{\nu}[\epsilon]^{-1}+\lambda^{-1}$,
$\lambda=9\mathbb{E}_{\nu}[\epsilon]/s^{2}$ we deduce $\mathbb{E}[\epsilon^{-1}]\mathbb{E}[\epsilon]=(9+s^{2})/9$
yielding the simpler expression
\begin{align*}
\hat{\theta}=9\mathbb{E}_{\nu}[\epsilon]\left[\frac{{\displaystyle \frac{s^{2}}{9}+\sqrt{\frac{s^{4}}{81}+4\frac{(9+s^{2})}{9}}}}{2(9+s^{2})}\right]\\
=\mathbb{E}_{\nu}[\epsilon]\frac{s^{2}+\sqrt{(s^{2}+18)^{2}}}{2(9+s^{2})}\,.
\end{align*}
The constrained minimisation problem is therefore akin to a moment
matching problem.

\section{Additional theoretical properties \protect\label{sec:Preliminary-theoretical-characte}}

In this section we omit the index $n\in\llbracket0,P\rrbracket$ and
provide precise elements to understand the properties of the algorithms
and estimators considered in this manuscript.

\subsection{Variance using folded mixture samples}

 The following provides us with a notion of effective sample size
for estimators of the form $\check{\mu}(f)$ as defined in Proposition~\ref{prop:variance-reduc-barmu}.
\begin{thm}
\label{thm:ESS-mu-check}Assume that $\upsilon\gg\mu^{\psi^{-k}}$
for $k\in\llbracket T\rrbracket$, then with $\check{Z}\sim\bar{\mu}$,
\begin{enumerate}
\item for $f\colon\mathsf{Z}\rightarrow\mathbb{R},$
\[
{\rm var}\left(\check{\mu}(f)\right)\leq\frac{2}{N}\frac{{\rm var}\big({\textstyle \sum}_{k=0}^{T}\mu\circ\psi^{k}(\check{Z})f\circ\psi^{k}(\check{Z})\big)+\|f\|_{\infty}^{2}{\rm var}\big({\textstyle \sum}_{k=0}^{T}\mu\circ\psi^{k}(\check{Z})\big)}{\mathbb{E}_{\bar{\mu}}\big({\textstyle \sum}_{k=0}^{T}\mu\circ\psi^{k}(\check{Z})\big)^{2}}\,,
\]
\item further
\[
{\rm var}\left(\check{\mu}(f)\right)\leq\frac{2\|f\|_{\infty}^{2}}{N}\left\{ \frac{2\mathbb{E}_{\bar{\mu}}\big(({\textstyle \sum}_{k=0}^{T}\mu\circ\psi^{k}(\check{Z}))^{2}\big)}{\mathbb{E}_{\bar{\mu}}\big({\textstyle \sum}_{k=0}^{T}\mu\circ\psi^{k}(\check{Z})\big)^{2}}-1\right\} \,.
\]
\end{enumerate}
\end{thm}

\begin{rem}
Multiplying $\mu$ by any constant $C_{\mu}>0$ does not alter the
values of the upper bounds, making the results relevant to the scenario
where $\mu$ is known up to a constant only. From the second statement
we can define a notion of effective sample size (ESS)
\[
\frac{N}{2}\frac{\mathbb{E}_{\bar{\mu}}\big({\textstyle \sum}_{k=0}^{T}\mu\circ\psi^{k}(\check{Z})\big)^{2}}{2\mathbb{E}_{\bar{\mu}}\big(({\textstyle \sum}_{k=0}^{T}\mu\circ\psi^{k}(z))^{2}\big)-\mathbb{E}_{\bar{\mu}}\big({\textstyle \sum}_{k=0}^{T}\mu\circ\psi^{k}(\check{Z})\big)^{2}}\,,
\]
which could be compared to $N(T+1)$ or $N$.
\end{rem}

\begin{proof}
We have
\[
{\rm var}\left(\check{\mu}(f)\right)=\frac{1}{N}{\rm var}_{\bar{\mu}}\left(\frac{\sum_{k=0}^{T}\mu\circ\psi^{k}(\check{Z})f\circ\psi^{k}(\check{Z})}{\sum_{l=0}^{T}\mu\circ\psi^{l}(\check{Z})}\right)
\]
and we apply Lemma~\ref{lem:NEO-IS-bounds} 
\begin{align*}
{\rm var}_{\bar{\mu}}\left(\frac{\sum_{k=0}^{T}\mu\circ\psi^{k}(\check{Z})f\circ\psi^{k}(\check{Z})}{\sum_{l=0}^{T}\mu\circ\psi^{l}(\check{Z})}\right) & \leq2\frac{{\rm var}\big({\textstyle \sum}_{k=0}^{T}\mu\circ\psi^{k}(\check{Z})f\circ\psi^{k}(\check{Z})\big)+\|f\|_{\infty}^{2}{\rm var}\big({\textstyle \sum}_{k=0}^{T}\mu\circ\psi^{k}(\check{Z})\big)}{\mathbb{E}_{\bar{\mu}}\big({\textstyle \sum}_{k=0}^{T}\mu\circ\psi^{k}(z)\big)^{2}}\\
 & \leq2\|f\|_{\infty}^{2}\left\{ \frac{2\mathbb{E}_{\bar{\mu}}\big(({\textstyle \sum}_{k=0}^{T}\mu\circ\psi^{k}(\check{Z}))^{2}\big)}{\mathbb{E}_{\bar{\mu}}\big({\textstyle \sum}_{k=0}^{T}\mu\circ\psi^{k}(\check{Z})\big)^{2}}-1\right\} \,.
\end{align*}
\end{proof}

\subsection{Variance using unfolded mixture samples \protect\label{subsec:PISA}}

For $\psi\colon\mathsf{Z}\rightarrow\mathbb{R}$ invertible, $\nu$
a probability distribution such that $\nu\gg\mu^{\psi^{-1}}$ and
$f\colon\mathsf{Z}\rightarrow\mathbb{R}$ such that $\mu(f)$ exists
we have the identity
\[
\mu(f)=\int f\circ\psi(z)\mu^{\psi^{-1}}({\rm d}z)=\int f\circ\psi(z)\frac{{\rm d}\mu^{\psi^{-1}}}{{\rm d}\nu}(z)\nu({\rm d}z)\,.
\]
As a result for $\{\psi_{k}\colon\mathsf{Z}\rightarrow\mathbb{R},k\in\llbracket K\rrbracket\}$,
all invertible and such that $\nu\gg\mu^{\psi_{k}^{-1}}$ for $k\in\llbracket K\rrbracket$,
we have
\begin{equation}
\mu(f)=\int\Bigl\{\frac{1}{K}\sum_{k\in\llbracket K\rrbracket}f\circ\psi_{k}(z)\frac{{\rm d}\mu^{\psi_{k}^{-1}}}{{\rm d}\nu}(z)\Bigr\}\nu({\rm d}z)\,,\label{eq:PISA-fundamental-identity-theory}
\end{equation}
which suggests the Pushforward Importance Sampling (PISA) estimator,
for $Z^{i}\overset{{\rm iid}}{\sim}\nu$, $i\in\llbracket N\rrbracket$
and with $\bar{\mu}_{/\nu}(f\mid z):=K^{-1}\sum_{k\in\llbracket K\rrbracket}f\circ\psi_{k}(z)\frac{{\rm d}\mu^{\psi_{k}^{-1}}}{{\rm d}\nu}(z)$
\begin{align}
\hat{\mu}(f) & =\sum_{i=1}^{N}\frac{\bar{\mu}_{/\nu}(f\mid Z^{i})}{\sum_{j=1}^{N}\bar{\mu}_{/\nu}(\mathrm{1}\mid Z^{j})}=\frac{\frac{1}{N}\sum_{i=1}^{N}\bar{\mu}_{/\nu}(f\mid Z^{i})}{\frac{1}{N}\sum_{j=1}^{N}\bar{\mu}_{/\nu}(\mathrm{1}\mid Z^{j})}\,,\label{eq:PISA-unfolded-estimator}
\end{align}
This is the estimator in (\ref{eq:snippet-estimator-I}) when $\nu=\mu_{n-1}$
and $\mu=\mu_{n}$.

\subsubsection{Relative efficiency for unfolded estimators}

In order to define the notion of relative efficiency for the estimator
$\hat{\mu}(f)$ in (\ref{eq:PISA-unfolded-estimator}) we first establish
the following bounds.
\begin{thm}
\label{thm:PISA-ESS} With the notation above and $Z\sim\nu$ throughout.
For any $f\colon\mathsf{Z}\rightarrow\mathbb{R},$
\begin{enumerate}
\item 
\begin{align*}
{\rm var}\big(\hat{\mu}(f)\big)\leq\mathbb{E}\Bigl(|\hat{\mu}(f)-\mu(f)|^{2}\Bigr) & \leq\frac{2}{N}\big\{{\rm var}\big(\bar{\mu}_{/\nu}(f\mid Z)\big)+\|f\|_{\infty}^{2}{\rm var}\big(\bar{\mu}_{/\nu}(1\mid Z)\big)\big\}\,,
\end{align*}
\item more simply
\begin{equation}
\mathbb{E}\Bigl(|\hat{\mu}(f)-\mu(f)|^{2}\Bigr)\leq\frac{2\|f\|_{\infty}^{2}}{KN}\Bigl\{\frac{2\mathbb{E}\Bigl[\bigl(\sum_{k\in\llbracket K\rrbracket}{\rm d}\mu^{\psi_{k}^{-1}}/{\rm d}\nu(Z)\bigr)^{2}\Bigr]}{K}-K\Bigr\}\,,\label{eq:PISA-ESS-bound}
\end{equation}
\item and
\[
|\mathbb{E}[\hat{\mu}(f)]-\mu(f)|\leq\frac{2\|f\|_{\infty}^{2}}{KN}\frac{\mathbb{E}\Bigl[\bigl(\sum_{k\in\llbracket K\rrbracket}{\rm d}\mu^{\psi_{k}^{-1}}/{\rm d}\nu(Z)\bigr)^{2}\Bigr]}{K}\,.
\]
\end{enumerate}
\end{thm}

\begin{rem}
The upper bound in the first statement confirms the control variate
nature of integrator snippets, even when using the unfolded perspective,
a property missed by the rougher bounds of the last two statements.
\end{rem}

\begin{rem}[ESS for PISA]
 \label{rem:def-ESS-unfolded} The notion of efficiency is usually
defined relative to the ``perfect'' Monte Carlo scenario that is
the standard estimator $\hat{\mu}_{0}$ of $\mu(f)$ relying on $KN$
iid samples from $\mu$ for which we have
\begin{equation}
{\rm var}\Bigl(\hat{\mu}_{0}(f)\Bigr)=\frac{{\rm var}_{\mu}(f)}{KN}\leq\frac{\|f\|_{\infty}^{2}}{KN}\,.\label{eq:PISA-iid-mu-standard}
\end{equation}
The $\mathrm{RE}_{0}$, is determined by the ratio of the upper bound
in (\ref{eq:PISA-ESS-bound}) by (\ref{eq:PISA-iid-mu-standard}).
Our point below is that the notion of efficiency can be defined relative
to any competing algorithm, virtual or not, in order to characterize
particular properties. For example we can compute the efficiency relative
to that of the ``ideal'' PISA estimator i.e. for which $\bar{\mu}_{/\nu}(\bar{f}\mid Z^{i})$
is replaced with $K^{-1}\sum_{k\in\llbracket K\rrbracket}f\circ\psi_{k}(Z^{i,k})\frac{{\rm d}\mu^{\psi_{k}^{-1}}}{{\rm d}\nu}(Z^{i,k})$,
$Z^{i,k}\overset{{\rm iid}}{\sim}\nu$ and 
\begin{equation}
{\rm var}\Bigl(\hat{\mu}_{1}(f)\Bigr)\leq\frac{2\|f\|_{\infty}^{2}}{KN}\left\{ \frac{2\sum_{k\in\llbracket K\rrbracket}\mathbb{E}\Bigl[\bigl({\rm d}\mu^{\psi_{k}^{-1}}/{\rm d}\nu(Z)\bigr)^{2}\Bigr]}{K}-K\right\} \,.\label{eq:variance-PISA-IID}
\end{equation}
The corresponding $\mathrm{RE}_{1}$ captures the loss incurred because
of dependence along a snippet. However, given our initial motivation
of recycling the computation of a standard HMC based SMC algorithm
we opt to define the $\mathrm{RE}_{2}$ relative to the estimator
relying on both ends of the snippet only, i.e. 
\[
\hat{\mu}_{2}(f)=\frac{1}{N}\sum_{i=1}^{N}\frac{\frac{1}{2}\frac{{\rm d}\mu}{{\rm d}\nu}(z^{i})f(Z^{i})+\frac{1}{2}\frac{{\rm d}\mu^{\psi_{K}^{-1}}}{{\rm d}\nu}(Z^{i})f\circ\psi_{K}(Z^{i})}{\sum_{j=1}^{N}\frac{1}{2}\frac{{\rm d}\mu}{{\rm d}\nu}(Z^{j})+\frac{1}{2}\frac{{\rm d}\mu^{\psi_{K}^{-1}}}{{\rm d}\nu}(Z^{j})}\,.
\]

In the SMC scenario considered in this manuscript (see Section~\ref{sec:Overview-and-motivation:})
the above can be thought of as a proxy for estimators obtained by
a ``Rao-Blackwellized'' SMC algorithm using $P_{n-1,T}$ in (\ref{eq:def-Pnk-intro}),
where $N$ particles $\{z_{n-1}^{(i)},i\in\llbracket N\rrbracket\}$
in Alg.~\ref{alg:generic-SMC} give rise to $2N$ weighted particles
\[
\{\big(z^{i},\bar{\alpha}_{n-1}(z_{i};T)\tfrac{\mu_{n}}{\mu_{n-1}}(z_{i})\big);\big(z^{i},\alpha_{n-1}(z_{i};T)\tfrac{\mu_{n}}{\mu_{n-1}}\circ\psi_{n}^{T}(z_{i})\big),i\in\llbracket N\rrbracket\}\:,
\]
with $\alpha_{n-1}(\cdot;T)$ defined in (\ref{eq:def-alpha-T}).
Resampling with these weights is then applied to obtain $N$ particles
and followed by an update of velocities to yield $\{z_{n}^{(i)},i\in\llbracket N\rrbracket\}$.
Now, we observe the similarity between
\begin{equation}
\alpha_{n-1}(z;T)\tfrac{\mu_{n}}{\mu_{n-1}}\circ\psi_{n}^{T}(z)=\min\left\{ 1,\frac{\mu_{n-1}\circ\psi_{n-1}^{T}}{\mu_{n-1}}(z)\right\} \times\frac{\mu_{n}\circ\psi_{n}^{T}}{\mu_{n-1}\circ\psi_{n}^{T}}(z)\,,\label{eq:rao-blackwell-SMC}
\end{equation}
and the corresponding weight in Alg.~\ref{alg:Unfolded-PDMP-SMC-1},
in particular when $\mu_{n-1}$ and $\mu_{n}$ are similar, and hence
$\psi_{n-1}$ and $\psi_{n}$. This motivates our choice of reference
to define $\mathrm{ESS}_{2}$ which has a clear computational advantage
since it involves ignoring $T-1$ terms only. In the present scenario,
following Lemma~\ref{lem:NEO-IS-bounds}, we have
\begin{align*}
{\rm var}\big(\hat{\mu}_{2}(f)\big) & \leq\frac{\|f\|_{\infty}^{2}}{N}\left\{ \mathbb{E}\left[\big(\frac{{\rm d}\mu}{{\rm d}\nu}(Z)+\frac{{\rm d}\mu^{\psi_{K}^{-1}}}{{\rm d}\nu}(Z)\big)^{2}\right]-\frac{1}{2}\right\} \,,
\end{align*}
which leads to the relative efficiency for PISA, 
\[
{\rm RE}_{2}=\frac{\frac{4\mathbb{E}\Bigl[\bigl(\sum_{k\in\llbracket K\rrbracket}{\rm d}\mu^{\psi_{k}^{-1}}/{\rm d}\nu(Z)\bigr)^{2}\Bigr]}{K^{2}}-2}{\mathbb{E}\left[\big(\frac{{\rm d}\mu}{{\rm d}\nu}(Z)+\frac{{\rm d}\mu^{\psi_{K}^{-1}}}{{\rm d}\nu}(Z)\big)^{2}\right]-2}\,,
\]
which can be estimated using empirical averages. 
\end{rem}

\begin{proof}[Proof of Theorem~\ref{thm:PISA-ESS}]
 We apply Lemma~\ref{lem:NEO-IS-bounds} with 
\[
A(Z^{1},\ldots,Z^{N})=\frac{1}{N}\sum_{i=1}^{N}\bar{\mu}_{/\nu}(f\mid Z^{i})\quad B(Z^{1},\ldots,Z^{N})=\frac{1}{N}\sum_{j=1}^{N}\bar{\mu}_{/\nu}(\mathrm{1}\mid Z^{j})\,.
\]
With $Z^{i}\overset{{\rm iid}}{\sim}\nu$, we have directly
\[
a=\mathbb{E}\big(A\big)=\mathbb{E}\big(\bar{\mu}_{/\nu}(f\mid Z)\big)=\mu(f)\quad b=\mathbb{E}\big(B\big)=1\,,
\]
and
\begin{align*}
{\rm var}\big(B\big) & =\frac{1}{N}{\rm var}\big(\bar{\mu}_{/\nu}(1\mid Z^{j})\big)\\
 & =\frac{1}{K^{2}N}\left\{ \mathbb{E}\Bigl[\bigl({\textstyle \sum}_{k\in\llbracket K\rrbracket}{\rm d}\mu^{\psi_{k}^{-1}}/{\rm d}\nu(Z)\bigr)^{2}\Bigr]-K^{2}\right\} \,.
\end{align*}
Now with $\|f\|_{\infty}<\infty$ 
\begin{align*}
{\rm var}\big(A\big) & =\frac{1}{N}{\rm var}\big(\bar{\mu}_{/\nu}(f\mid Z)\big)\\
 & =\frac{1}{K^{2}N}\Bigl\{\|f\|_{\infty}^{2}\mathbb{E}\Bigl[\bigl({\textstyle \sum}_{k\in\llbracket K\rrbracket}{\rm d}\mu^{\psi_{k}^{-1}}/{\rm d}\nu(Z)\bigr)^{2}\Bigr]-K^{2}\mu(f)^{2}\Bigr\}\\
 & \leq\frac{\|f\|_{\infty}^{2}}{K^{2}N}\mathbb{E}\Bigl[\bigl({\textstyle \sum}_{k\in\llbracket K\rrbracket}{\rm d}\mu^{\psi_{k}^{-1}}/{\rm d}\nu(Z)\bigr)^{2}\Bigr]\,.
\end{align*}
We conclude noting that we have $|A/B|\leq\|f\|_{\infty}$.
\end{proof}
For clarity we reproduce the very useful lemma \cite[Lemma 6]{thin2021neo},
correcting a couple of minor typos along the way.
\begin{lem}
\label{lem:NEO-IS-bounds}Let $A,B$ be two integrable random variables
satisfying $|A/B|\leq M$ almost surely for some $M>0$ and let $a=\mathbb{E}(A)$
and $b=\mathbb{E}(B)\neq0$. Then
\begin{gather*}
{\rm var}(A/B)\leq\mathbb{E}[|A/B-a/b|^{2}]\leq\frac{2}{b^{2}}\Bigl\{{\rm var}(A)+M^{2}{\rm var}(B)\Bigr\}\,,\\
|\mathbb{E}[A/B]-a/b|\leq\frac{\sqrt{{\rm var}(A/B){\rm var}(B)}}{b}\,.
\end{gather*}
\end{lem}

\subsubsection{Optimal weights \protect\label{subsec:Optimal-weights}}

As mentioned in Subsection~\ref{subsec:Straightforward-generalizations}
it is possible to consider the more general scenario where unequal
probability weights $\omega=\{\omega_{k}\in\mathbb{R},k\in\llbracket0,T\rrbracket\colon\sum_{k=0}^{T}\omega_{k}=1\}$
are ascribed to the elements of the snippets, yielding the estimator
\begin{align*}
\hat{\mu}_{\omega}(f) & :=\frac{\sum_{i=1}^{N}\sum_{k=0}^{T}\omega_{k}\frac{{\rm d}\mu^{\psi_{k}^{-1}}}{{\rm d}\nu}(Z^{i})f\circ\psi_{k}(Z^{i})}{\sum_{j=1}^{N}\sum_{l=0}^{T}\omega_{l}\frac{{\rm d}\mu^{\psi_{l}^{-1}}}{{\rm d}\nu}(Z^{j})}\,,
\end{align*}
and a natural question is that of the optimal choice of $\omega$.
Note that in the context of PISA the condition $\omega_{k}\geq0$
is not required, as suggested by the justification of the identity
(\ref{eq:PISA-fundamental-identity-theory}). However it should be
clear that this condition should be enforced in the context of integrator
snippet SMC, since the probabilistic interpretation is otherwise lost,
or if the expectation is known to be non-negative. Here we discuss
optimization of the variance upperbound provided by Lemma~\ref{lem:NEO-IS-bounds},
\begin{align*}
\frac{2\|f\|_{\infty}^{2}}{N}\Bigl\{{\rm var}\biggl({\textstyle \sum_{k\in\llbracket K\rrbracket}}\omega_{k}\frac{{\rm d}\mu^{\psi_{k}^{-1}}}{{\rm d}\nu}(Z)\frac{f\circ\psi_{k}}{\|f\|_{\infty}}(Z)\biggr)+ & {\rm var}\biggl({\textstyle \sum_{k\in\llbracket K\rrbracket}}\omega_{k}\frac{{\rm d}\mu^{\psi_{k}^{-1}}}{{\rm d}\nu}(Z)\biggr)\Bigr\}\\
\leq\frac{2\|f\|_{\infty}^{2}}{N} & \omega^{\top}\big(\Sigma(f;\psi)+\Sigma_{\psi}(1;\psi)\big)\omega
\end{align*}
where for $k,l\in\llbracket K\rrbracket$, 
\[
\Sigma_{kl}(f;\psi)=\mathbb{E}\biggl[\frac{{\rm d}\mu^{\psi_{k}^{-1}}}{{\rm d}\nu}(Z)\frac{f\circ\psi_{k}}{\|f\|_{\infty}}(Z)\frac{{\rm d}\mu^{\psi_{l}^{-1}}}{{\rm d}\nu}(Z)\frac{f\circ\psi_{l}}{\|f\|_{\infty}}(Z)\biggr]-\mu(f/\|f\|_{\infty})^{2}\,.
\]
It is a classical result that $\omega^{\top}\big(\Sigma(f;\psi)+\Sigma_{\psi}(1;\psi)\big)\omega\geq\lambda_{\min}\big(\Sigma(f;\psi)+\Sigma_{\psi}(1;\psi)\big)\omega^{\top}\omega$
and that minimum is reached for the eigenvector(s) $\omega_{{\rm min}}$
corresponding to the smallest eigenvalue $\lambda_{\min}\big(\Sigma(f;\psi)+\Sigma_{\psi}(1;\psi)\big)$
of $\Sigma(f;\psi)+\Sigma_{\psi}(1;\psi)$. If the constraint of non-negative
entries is to be enforced then a standard quadratic programming procedure
should be used. The same ideas can be applied to the function independent
upperbounds used in our definitions of efficiency.

\subsection{More on variance reduction and optimal flow}

We now focus on some properties of the estimator $\hat{\mu}(f)$ of
$\mu(f)$ in (\ref{eq:mu-hat}). To facilitate the analysis and later
developments we consider the scenario where, with $t\mapsto\psi_{t}(z)$
the flow solution of an ODE, assume the dominating measure $\upsilon\gg\mu$
and $\upsilon$ is invariant by the flow. We consider
\[
\bar{\mu}({\rm d}t,{\rm d}z)=\frac{1}{T}\mu^{\psi_{-t}}({\rm d}z)\mathbf{1}\{0\leq t\leq T\}{\rm d}t\,,
\]
and notice that similarly to the integrator scenario, for any $f\colon\mathsf{Z}\rightarrow\mathbb{R}$,
$\bar{f}(t,z):=f\circ\psi_{t}(z)$ we have $\mathbb{E}_{\bar{\mu}}\left(\bar{\mu}(\bar{f}\mid\check{Z})\right)=\bar{\mu}(\bar{f})=\mu(f)$.
where for any $z\in\mathsf{Z}$, where
\[
\bar{\mu}(\bar{f}\mid z)=\frac{1}{T}\int_{0}^{T}f\circ\psi_{t}(z)\frac{\mu\circ\psi_{t}(z)}{\int\mu\circ\psi_{u}(z){\rm d}u}{\rm d}t
\]
the following estimator of $\mu(f)$ is considered, with $\check{Z}_{i}\overset{{\rm iid}}{\sim}\bar{\mu}$
\[
\hat{\mu}(f)=\frac{1}{N}\sum_{i=1}^{N}\bar{\mu}\big(\bar{f}(\mathrm{T},\check{Z}_{i})\mid\check{Z}_{i}\big)\,.
\]
We now show that in the examples considered in this paper our approach
can be understood as implementing unbiased Riemann sum approximations
of line integrals along contours. Adopting a continuous time approach
can be justified as follows. For numerous integrators, the conditions
of the following proposition are satisfied; this is the case for the
leapfrog integrator of Hamilton's equation e.g. \cite[Theorem 3.4]{hairer1993sp,thin2021neo}
and \cite[Appendix 3.1, Theorem 9]{thin2021neo} for detailed results.
\begin{prop}
\label{prop:line-integrals-integrators} Let $\tau>0$, for any $z\in\mathsf{Z}$
let $[0,\tau]\ni t\mapsto\psi_{t}(z)$ be a flow and for $\epsilon>0$
let $\{\psi^{k}(z;\epsilon),0\leq k\epsilon\leq\tau\}$ be a discretization
of $t\mapsto\psi_{t}$ such that that for any $z\in\mathsf{Z}$ there
exists $C>0$ such that for any $(k,\epsilon)\in\mathbb{N}\times\mathbb{R}_{+}$
\[
|\psi^{k}(z;\epsilon)-\psi_{k\epsilon}(z)|\leq C\epsilon^{2}\,.
\]
Then for any continuous $g\colon\mathsf{Z}\rightarrow\mathbb{R}$
such that the Riemann integral
\[
I(g):=\int_{0}^{\tau}g\circ\psi_{t}(z){\rm {\rm d}}t\,,
\]
exists we have
\[
\lim_{T\rightarrow\infty}\frac{1}{n}\sum_{k=0}^{n-1}g\circ\psi_{k\tau/n}(z;\tau/n)=I(g)\,.
\]
\end{prop}

\begin{proof}
We have
\[
\lim_{n\rightarrow\infty}\frac{1}{n}\sum_{k=0}^{n-1}g\circ\psi_{k\tau/n}(z)=\int_{0}^{\tau}g\circ\psi_{t}(z){\rm {\rm d}}t
\]
\[
\lim_{n\rightarrow\infty}\frac{1}{n}\biggl|\sum_{k=0}^{n-1}g\circ\psi_{k\tau/n}(z)-g\circ\psi^{k}(z;\tau/n)\biggr|=0\,,
\]
and we can conclude.
\end{proof}
\begin{rem}
Naturally convergence is uniform in $z\in\mathsf{Z}$ with additional
assumptions and we note that in some scenarios dependence on $\tau$
can be characterized, allowing in principle $\tau$ to grow with $n$
.
\end{rem}

\subsubsection{Hamiltonian contour decomposition}

Assume $\mu$ has a probability density with respect to the Lebesgue
measure and let $\zeta\colon\mathsf{Z}\subset\mathbb{R}^{d}\rightarrow\mathbb{R}$
be Lipschitz continuous such that $\nabla\zeta(z)\neq0$ for all $z\in\mathsf{Z}$,
then the co-area formula states that
\[
\int_{\mathsf{Z}}f(z)\mu(z){\rm d}z=\int\big[\int_{\zeta^{-1}(s)}f(z)|\nabla\zeta(z)|^{-1}\mu(z)\mathcal{H}_{d-1}({\rm d}z)\big]{\rm d}s,
\]
where $\mathcal{H}_{d-1}$ is the $(d-1)$-dimensional Hausdorff measure,
used here to measure length along the contours $\zeta^{-1}(s)$. For
example in the HMC setup where $\mu(z)\propto\exp\big(-H(z)\big)$
and $z=(x,v)$ one may choose $\zeta(z)=H(z)=-\log\big(\mu(z)\big)$,
leading to a decomposition of the expectation $\mu(f)$ according
to equi-energy contours of $H(z)$
\[
\mu(f)=\int\exp(-s)\big[\int_{\zeta^{-1}(s)}f(z)|\nabla H(z)|^{-1}\mathcal{H}_{d-1}({\rm d}z)\big]{\rm d}s.
\]
We now show how the solution of Hamilton's equation could be used
as the basis for estimators of $\mu(f)$ mixing Riemanian sum-like
and Monte Carlo estimation techniques.

\paragraph{Favourable scenario where $d=2$. }

Let $s\in\mathbb{R}$ such that $H^{-1}(s)\neq\emptyset$ and assume
that for some $(x_{0},v_{0})\in H^{-1}(s)$ Hamilton's equations $\dot{x}=-\nabla_{v}H(x,v)$
and $\dot{v}=\nabla_{x}H(x,v)$ have solutions $t\mapsto(x_{t},v_{t})=\psi_{t}(z_{0})\in H^{-1}(s)$
with $\big(x_{\tau(s)},v_{\tau(s)})=(x_{0},v_{0})=z_{0}$ for some
$\tau(s)>0$, that is the contours $H^{-1}(s)$ can be parametrised
with the solutions of Hamilton's equation at corresponding level. 
\begin{prop}
\label{prop:decomp-Hamilton-2D-estimator} We have
\[
\int_{\mathsf{Z}}f(z)\mu({\rm d}z)=\int_{\mathsf{Z}}\left[[\tau\circ H(z_{0})]^{-1}\int_{0}^{\tau\circ H(z_{0})}f(z_{t}){\rm d}t\right]\mu({\rm d}z_{0})\,,
\]
implying that, assuming the integral along the path $t\mapsto z_{t}$
is tractable, 
\begin{equation}
[\tau\circ H(z_{0})]^{-1}\int_{0}^{\tau\circ H(z_{0})}f(z_{t}){\rm d}t\,\text{for}\,z_{0}\sim\mu\,,\label{eq:co-area-estimator-integral}
\end{equation}
is an unbiased estimator of $\mu(f)$.
\end{prop}

\begin{proof}
Since $H^{-1}(s)$ is rectifiable we have that \cite[Theorem 2.6.2]{burago2022course}
for $g\colon\mathsf{Z}\rightarrow\mathbb{R}$,
\[
\int_{H^{-1}(s)}g(z)\mathcal{H}_{d-1}({\rm d}z)=\int_{0}^{\tau(s)}g(z_{t})|\nabla H(z_{t})|{\rm d}t\,.
\]
Consequently
\begin{align}
\int_{\mathsf{Z}}f(z)\mu({\rm d}z) & =\int\tau(s)\exp(-s)\left[\tau(s)^{-1}\int_{0}^{\tau(s)}f(z_{s,t}){\rm d}t\right]{\rm d}s\,,\label{eq:decomposition-Ham-2D}
\end{align}
where the notation now reflects dependence on $s$ of the flow and
notice that since 
\[
H^{-1}(s)\ni z_{0}\mapsto[\tau\circ H(z_{0})]^{-1}\int_{0}^{\tau\circ H(z_{0})}f(z_{t})
\]
is constant, $z_{0,s}\in H^{-1}(s)$ can be arbitrary. since indeed
from (\ref{eq:decomposition-Ham-2D})
\begin{align*}
\int_{\mathsf{Z}}\mathbf{1}\{H(z)\in A\}\mu({\rm d}z)= & \int\tau(s)\exp(-s)\left[\tau(s)^{-1}\int_{0}^{\tau(s)}\mathbf{1}\{s\in A\}{\rm d}t\right]{\rm d}s\\
= & \int\tau(s)\exp(-s)\mathbf{1}\{s\in A\}{\rm d}s.
\end{align*}
Note that we can write, since for $s\in\mathbb{R}$, $H^{-1}(s)\ni z_{0}\mapsto[\tau\circ H(z_{0})]^{-1}\int_{0}^{\tau\circ H(z_{0})}f(z_{t})$
is constant,
\begin{align*}
\int_{\mathsf{Z}}\left[[\tau\circ H(z_{0})]^{-1}\int_{0}^{\tau\circ H(z_{0})}f(z_{t}){\rm d}t\right]\mu({\rm d}z_{0})\\
=\int\tau(s)\exp(-s) & \left\{ \tau(s)^{-1}\int_{0}^{\tau(s)}\left[\big(\tau\circ H(z_{s,0})\big)^{-1}\int_{0}^{\tau\circ H(z_{0})}f(z_{s,t}){\rm d}t\right]\right\} {\rm d}t{\rm d}s\\
=\int\tau(s)\exp(-s) & \left[\tau(s)^{-1}\int_{0}^{\tau(s)}f(z_{s,t}){\rm d}t\right]{\rm d}t{\rm d}s\,.
\end{align*}
\end{proof}
A remarkable point is that the strategy developed in this manuscript
provides a general methodology to implement numerically the ideas
underpinning the decomposition of Proposition~\ref{prop:decomp-Hamilton-2D-estimator}
by using the estimator $\check{\mu}(f)$ in (\ref{eq:mu-check}) and
invoking Proposition~\ref{prop:line-integrals-integrators}, assuming
$s\mapsto\tau(s)$ to be known. This point is valid outside of the
SMC framework and it is worth pointing out that $\check{\mu}(f)$
in (\ref{eq:mu-check}) is unbiased if the samples used are marginally
from $\bar{\mu}$.

Remark the promise of dimension free estimators if the one dimensional
line integrals in (\ref{eq:co-area-estimator-integral}) were easily
computed and sampling from the one-dimensional energy distribution
was routine -- however the scenario $d\geq3$ is more subtle.

\paragraph*{General scenario}

In the scenario where $d\geq3$ the co-area decomposition still holds,
but the solution to Hamilton's equation can in general not be used
to compute integrals over the hypersurface $H^{-1}(s)$. This would
require a form of ergodicity \cite{szasz-2000boltzmann,tupper2005ergodicity}
of the form, for $z_{0}\in H^{-1}(s)$,
\[
\lim_{\tau\rightarrow\infty}\frac{1}{\tau}\int_{0}^{\tau}f\circ\psi_{t}(z_{0}){\rm d}t=\bar{f}(z_{0})=\frac{\int_{\zeta^{-1}(s)}f(z)|\nabla H(z)|^{-1}\mathcal{H}_{d-1}({\rm d}z)}{\int_{\zeta^{-1}(s)}|\nabla H(z)|^{-1}\mathcal{H}_{d-1}({\rm d}z)}\,,
\]
where the limit always exists in the ${\rm L}^{2}(\mu)$ sense and
constitutes Von Neumann's mean ergodic theorem \cite{vonneumann:1932,szasz-2000boltzmann},
and the rightmost equality forms Boltzmann's conjecture. An interesting
property in the present context is that $\mathbb{E}_{\bar{\mu}}\bigl(\bar{f}(Z)\bigr)=\mathbb{E}_{\mu}\bigl(f(Z)\bigr)$
for $f\in\mathrm{L}^{2}(\mu)$ and one could replicate the variance
reduction properties developed earlier. Boltzmann's conjecture has
long been disproved, as numerous Hamiltonians can be shown not to
lead to ergodic systems, although some sub-classes do. However a weaker,
or local, form or ergodicity can hold on sets of a partition of $\mathsf{Z}$
\begin{example}[Double well potential]
\label{exa:double-well} Consider the scenario where $\mathsf{X}=\mathsf{V}=\mathbb{R}$,
$U(x)=(x^{2}-1)^{2}$ and kinetic energy $v^{2}$ \cite{stoltz2010free}.
Elementary manipulations show that satisfying Hamilton's equations
(\ref{eq:hamiltons-equations}) imposes $t\mapsto H(x_{t},\dot{x}_{t})=(x_{t}^{2}-1)^{2}+\dot{x}_{t}^{2}=s>0$
and therefore requires $t\mapsto x_{t}\in[-\sqrt{1+\sqrt{s}},-\sqrt{1-\sqrt{s}}]\cup[\sqrt{1-\sqrt{s}},\sqrt{1+\sqrt{s}}]$
-- importantly the intervals are not connected for $s<1$. Rearranging
terms any solution of (\ref{eq:hamiltons-equations}) must satisfy
\begin{align*}
\dot{x}_{t} & =\pm\sqrt{s-(x_{t}^{2}-1)^{2}}\,,
\end{align*}
that is the velocity is a function of the position in the double well,
maximal for $x_{t}^{2}=1$, vanishing as $x_{t}^{2}\rightarrow1\pm\sqrt{s}$
and a sign flip at the endpoints of the intervals. Therefore the system
is not ergodic, but ergodicity trivially occurs in each well.
\end{example}

In general this partitioning of $\mathsf{Z}$ can be intricate but
it should be clear that in principle this could reduce variability
of an estimator. In the toy Example~\ref{exa:double-well}, a purely
mathematical algorithm inspired by the discussions above would choose
the right or left well with probability $1/2$ and then integrate
deterministically, producing samples taking at most two values which
could be averaged to compute $\mu(f)$. We further note that in our
context a relevant result would be concerned with the limit, for any
$x_{0}\in\mathsf{X}$,
\[
\lim_{\tau\rightarrow\infty}\int\big[\frac{1}{\tau}\int_{0}^{\tau}f\circ\psi_{t}(x_{0},\rho\mathbf{e}){\rm d}t\big]\varpi_{\mathbf{e}}({\rm d}\mathbf{e})\,,
\]
where we have used a polar reparametrization $v=\rho\mathbf{e}$ $(\rho,\mathbf{e})\in\mathbb{R}_{+}\times\mathcal{S}_{d-1}$,
and whether a marginal version of Boltzmann's conjecture holds for
this limit. Example~\ref{exa:double-well} indicates that this is
not true in general but the cardinality of the partition of $\mathsf{X}$
may be reduced.

\subsubsection{Advantage of averaging and control variate interpretation \protect\label{subsec:control-variate-interpretation}}

Consider the scenario where $(\mathsf{T},\check{Z})\sim\bar{\mu}(t,{\rm d}z)=\frac{1}{\tau}\mu^{\psi_{-t}}({\rm d}z)\mathbf{1}\{0\leq t\leq\tau\}$,
$[0,\tau]\ni t\mapsto\psi_{t}$ for some $\tau>0$ is the flow solution
of an ODE of the form $\dot{z}_{t}=F\circ z_{t}$ for some field $F\colon\mathsf{Z}\rightarrow\mathsf{Z}$.
For $f\colon\mathsf{Z}\rightarrow\mathbb{R}$ we have

\[
{\rm var}_{\mu}\left(f(Z)\right)={\rm var}_{\bar{\mu}}\left(\mathbb{E}_{\bar{\mu}}\big(\bar{f}(\mathsf{T},\check{Z})\mid\check{Z}\big)\right)+\mathbb{E}_{\bar{\mu}}\big({\rm var}_{\bar{\mu}}(\bar{f}(\mathsf{T},\check{Z})\mid\check{Z})\big)\,.
\]
and we are interested in determining $t\mapsto\psi_{t}$ (or $F$)
in order to minimize ${\rm var}_{\bar{\mu}}\left(\mathbb{E}_{\bar{\mu}}\big(\bar{f}(\mathsf{T},\check{Z})\mid\check{Z}\big)\right)$
and maximize improvement over simple Monte Carlo. We recall that the
Jabobian determinant of the flow $t\mapsto\psi_{t}(z)$ is given by
\cite[p. 174108-5]{10.1063/1.4874000}
\[
J_{t}(z)=|\det\big(\nabla\otimes\psi_{t}(z)\big)|=\mathcal{J}\circ\psi_{t}(z)\;\text{with}\;\mathcal{J}(z):=\exp\left(\int_{0}^{t}(\nabla\cdot F)(z)\right)\,.
\]

\begin{lem}
\label{lem:var-barmu-autocor}Let $t\mapsto\psi_{t}$ be a flow solution
of $\dot{z}_{t}=F\circ z_{t}$ and assume that with $\upsilon$ the
Lebesgue measure, $\upsilon\gg\mu$. Then
\begin{align*}
\mathbb{E}_{\bar{\mu}}\left(\mathbb{E}_{\bar{\mu}}\big(\bar{f}(\mathsf{T},\check{Z})\mid\check{Z}\big)^{2}\right) & =2\int_{0}^{\tau}\left\{ \int_{0}^{\tau-u}[\bar{\mu}\circ\psi_{-t}(z)]^{-1}{\rm d}t\right\} \Bigl\{\int\mu({\rm d}z)f(z)f\circ\psi_{u}(z)\mu\circ\psi_{u}(z)\mathcal{J}\circ\psi_{u}(z)\Bigr\}{\rm d}u\\
 & \leq2\left\{ \int_{0}^{\tau}[\bar{\mu}\circ\psi_{-t}(z)]^{-1}{\rm d}t\right\} \int_{0}^{\tau}\Bigl\{\int\mu({\rm d}z)f(z)f\circ\psi_{u}(z)\mu\circ\psi_{u}(z)\mathcal{J}\circ\psi_{u}(z)\Bigr\}{\rm d}u\,,
\end{align*}
with
\begin{align*}
\bar{\mu}\circ\psi_{-t}(z) & =\int_{-t}^{\tau-t}\mu\circ\psi_{u}(z)\,\mathcal{J}\circ\psi_{u}(z){\rm d}u\,.
\end{align*}
In particular for $t\mapsto\psi_{t}$ the flow solution of Hamilton's
equations associated with $\mu$ we have
\begin{align*}
\mathbb{E}_{\bar{\mu}}\left(\mathbb{E}_{\bar{\mu}}\big(\bar{f}(\mathsf{T},\check{Z})\mid\check{Z}\big)^{2}\right) & =2\int\big(1-t/\tau\big)\int\mu({\rm d}z)f(z)f\circ\psi_{t}(z){\rm d}t\,\cdot
\end{align*}
\end{lem}

\begin{proof}
Using Fubini's theorem we have
\begin{align*}
\int\left\{ \int_{0}^{\tau}f\circ\psi_{t}(z)\frac{{\rm d}\mu^{\psi_{-t}}}{{\rm d}\bar{\mu}}(z){\rm d}t\right\} ^{2}\bar{\mu}({\rm d}z) & =\\
=2\int_{0}^{\tau}\int_{0}^{\tau}\mathbf{1}\{t'\geq t\}\int\bar{\mu}({\rm d}z) & f\circ\psi_{t}(z)\frac{{\rm d}\mu^{\psi_{-t}}}{{\rm d}\bar{\mu}}(z)f\circ\psi_{t'-t+t}(z)\frac{{\rm d}\mu^{\psi_{-t'}}}{{\rm d}\bar{\mu}}\circ\psi_{-t}\circ\psi_{t}(z){\rm d}t'{\rm d}t\\
=2\int_{0}^{\tau}\int_{0}^{\tau}\mathbf{1}\{t'\geq t\}\int\mu({\rm d}z) & f(z)f\circ\psi_{t'-t}(z)\frac{{\rm d}\mu^{\psi_{-t'}}}{{\rm d}\bar{\mu}}\circ\psi_{-t}(z){\rm d}t'{\rm d}t\,.
\end{align*}
Further, we have
\begin{align*}
\frac{{\rm d}\mu^{\psi_{-t'}}}{{\rm d}\bar{\mu}}(z) & =\frac{\mu\circ\psi_{t'}(z)\,\mathcal{J}\circ\psi_{t'}(z)}{\bar{\mu}(z)}\,,
\end{align*}
with $\bar{\mu}(z)=\int_{0}^{\tau}\mu\circ\psi_{u}(z)\,\mathcal{J}\circ\psi_{u}(z){\rm d}u$.
It is straightforward that
\begin{align*}
\bar{\mu}\circ\psi_{-t}(z) & =\int_{0}^{\tau}\mu\circ\psi_{u-t}(z)\,\mathcal{J}\circ\psi_{u-t}(z){\rm d}u\\
 & =\int_{-t}^{\tau-t}\mu\circ\psi_{u'}(z)\,\mathcal{J}\circ\psi_{u'}(z){\rm d}u'\,.
\end{align*}
Consequently
\begin{align*}
\int\left\{ \int_{0}^{\tau}f\circ\psi_{t}(z)\frac{{\rm d}\mu^{\psi_{-t}}}{{\rm d}\bar{\mu}}(z){\rm d}t\right\} ^{2}\bar{\mu}({\rm d}z)=\\
=2\int_{0}^{\tau}\int_{0}^{\tau} & \frac{\mathbf{1}\{\tau-t\geq u\geq0\}}{\bar{\mu}\circ\psi_{-t}(z)}\int\mu({\rm d}z)f(z)f\circ\psi_{u}(z)\mu\circ\psi_{u}(z)\mathcal{J}\circ\psi_{u}(z){\rm d}u{\rm d}t\\
=2\int_{0}^{\tau} & \left\{ \int_{0}^{\tau-u}[\bar{\mu}\circ\psi_{-t}(z)]^{-1}{\rm d}t\right\} \int\mu({\rm d}z)f(z)f\circ\psi_{u}(z)\mu\circ\psi_{u}(z)\mathcal{J}\circ\psi_{u}(z){\rm d}u
\end{align*}
For the second statement we have $\bar{\mu}\circ\psi_{-t}(z)=\tau\,\mu(z)$
and
\begin{align*}
\int\left\{ \int_{0}^{\tau}f\circ\psi_{t}(z)\frac{{\rm d}\mu^{\psi_{-t}}}{{\rm d}\bar{\mu}}(z){\rm d}t\right\} ^{2}\bar{\mu}({\rm d}z)=\\
=2\int_{0}^{\tau} & \bigl(1-u/\tau\bigr)\int\mu({\rm d}z)f(z)f\circ\psi_{u}(z){\rm d}t''
\end{align*}
which is akin to the integration correlation time encountered in MCMC. 
\end{proof}
This is somewhat reminiscent of what is advocated in the literature
in the context of HMC or randomized HMC where the integration time
$t$ (or $T$ when using an integrator) is randomized \cite{neal2011mcmc,hoffman2021adaptive}. 
\begin{example}
\label{exa:hamilton-different-scales}Consider the Gaussian scenario
where $\mathsf{X}=\mathsf{V}=\mathbb{R}^{d}$, $\pi(x)=\mathcal{N}(x;0,\Sigma)$
with diagonal covariance matrix such that for $i\in\llbracket d\rrbracket$,
$\Sigma_{ii}=\sigma_{i}^{2}$ and $\varpi(v)=\mathcal{N}(v;0,{\rm Id})$.
Using the reparametrization $(x_{0}(i),v_{0}(i))=(a_{i}\sigma_{i}\sin(\phi_{i}),a_{i}\cos\phi_{i})$
the solution of Hamilton's equations is given for $i\in\llbracket d\rrbracket$
by
\begin{align*}
x_{t}(i) & =a_{i}\sigma\sin(t/\sigma_{i}+\phi_{i})\\
v_{t}(i) & =a_{i}\cos(t/\sigma_{i}+\phi_{i})\,.
\end{align*}
We have for each component $\mathbb{E}_{\mu}[X_{0}(i)X_{t}(i)]=\sigma_{i}^{2}\cos(t/\sigma_{i})$,
which clearly does not vanish as $t$, or $\tau$, increase: this
is a particularly negative result for standard HMC where the final
state of the computed integrator is used. Worse, as noted early on
in \cite{neal2011mcmc} this implies that for heterogeneous values
$\{\sigma_{i}^{2},i\in\llbracket d\rrbracket\}$ no integration time
may be suitable simultaneously for all coordinates. This motivated
the introduction of random integration times \cite{neal2011mcmc}
which leads to the average correlation
\[
\mathbb{E}_{\mu}\left\{ \frac{1}{\tau}\int_{0}^{\tau}X_{0}(i)X_{t}(i){\rm d}t\right\} =\frac{\sin(\tau/\sigma_{i})}{\tau/\sigma_{i}}\,,
\]
where it is assumed that $\tau$ is independent of the initial state.
This should be contrasted with the fixed integration time scenario
since as $\tau$ increases this vanishes and is even negative for
some values of $\tau$ (the minimum is here reached for about $\tau_{i}=4.5\sigma_{i}$
for a given component). 
\end{example}

The example therefore illustrates that our approach implements this
averaging feature, and therefore shares its benefits, within the context
of an iterative algorithm. The example also highlights a control variate
technique intepretation. More specifically in the discrete time scenarios
$\{f\circ\psi^{k}(z),k\in\llbracket P\rrbracket\}$ can be interpreted
as control variates, but can induce both positive and negative correlations.

\subsubsection{Towards an optimal flow?}

In this section we are looking to determine a flow for some $\tau>0$
$[0,\tau]\ni t\mapsto\psi_{t}$ of an ODE of the form $\dot{z}_{t}=F\circ z_{t}$
for some field $F\colon\mathsf{Z}\rightarrow\mathsf{Z}$ which defines
as above the probability model
\[
\bar{\mu}({\rm d}t,{\rm d}z)=\frac{1}{\tau}\mu^{\psi_{-t}}({\rm d}z)\mathbf{1}\{0\leq t\leq\tau\}{\rm d}t\,,
\]
which has the property that for any $f\colon\mathsf{Z}\rightarrow\mathbb{R}$,
defining $\bar{f}(t,z):=f\circ\psi_{t}(z)$, then $\bar{\mu}(\bar{f})=\mu(f)$.
This suggests the use of a Rao-Blackwellized estimator inspired by
$\mathbb{E}_{\bar{\mu}}\left(\bar{\mu}(\bar{f}\mid\check{Z})\right)$.
Assuming $\mathsf{Z}=\mathbb{R}^{d}\times\mathbb{R}^{d}$ and that
$\mu$ has a density with respect to the Lebesgue measure, then for
$z\in\mathsf{Z}$
\begin{equation}
\bar{\mu}(\bar{f}\mid z)=\frac{1}{\tau}\int_{0}^{\tau}f\circ\psi_{t}(z)\frac{\mu\circ\psi_{t}(z)\,\mathcal{J}\circ\psi_{t}(z)}{\int\mu\circ\psi_{u}(z)\,\mathcal{J}\circ\psi_{u}{\rm d}u}{\rm d}t\label{eq:bar-mu-estimator-optim-section}
\end{equation}

In the light of Lemma~\ref{lem:var-barmu-autocor} we aim to find
for any $z\in\mathsf{Z}$ the flow solutions $t\mapsto\psi_{t}(z)$
of ODEs $\dot{z}_{t}=F_{z}(z_{t})$ such that the function $t\mapsto\int\mu({\rm d}z)f(z)f\circ\psi_{t}(z)\mu\circ\psi_{t}(z)\,\mathcal{J}\circ\psi_{t}(z)$
decreases as fast as possible. This is motivated by the fact that
the integral on $[0,\tau]$ of this mapping appears in the variance
upper bound for (\ref{eq:bar-mu-estimator-optim-section}) in Lemma~\ref{lem:var-barmu-autocor},
which we want to minimize. Note that we also expect this mapping to
be smooth under general conditions not detailed in this preliminary
work. For smooth enough flow and $f$ we have, with $g:=f\times\mu\times\mathcal{J}$,
\begin{align*}
\frac{{\rm d}}{{\rm d}t}\mathbb{E}_{\mu}\left[f(Z)\frac{g\circ\psi_{t}(Z)}{\bar{\mu}(Z)}\right] & =\mathbb{E}_{\mu}\left[\frac{f(Z)}{\bar{\mu}(Z)}\langle\nabla g\circ\psi_{t}(Z),\dot{\psi}_{t}(Z)\rangle\right]
\end{align*}
Pointwise, the steepest descent direction is given by
\[
\dot{\psi}_{t}(z)=-C_{t}(z)\frac{f(z)\,\nabla g\circ\psi_{t}(z)}{|f(z)\,\nabla g\circ\psi_{t}(z)|}
\]
for a positive function $C_{t}(z)$ to be determined optimally. In
this scenario we therefore have
\[
\frac{{\rm d}}{{\rm d}t}\mathbb{E}_{\mu}\left[f(Z)\frac{g\circ\psi_{t}(Z)}{\bar{\mu}(Z)}\right]=-\mathbb{E}_{\mu}\left[\frac{C_{t}(z)}{\bar{\mu}(Z)}|f(Z)\,\nabla g\circ\psi_{t}(Z)|\right]
\]
and by Cauchy-Schwartz the (positive) expectation is maximized for
\[
C_{t}(z)=C_{t}\frac{|f(z)\,\nabla g\circ\psi_{t}(z)|}{\bar{\mu}(z)}\,.
\]
and the trajectories we are interested in must be such that, for some
$C>0$,
\[
\dot{\psi}_{t}(z)=-\frac{C}{\bar{\mu}(z)}f(z)\,\nabla g\circ\psi_{t}(z)=-\frac{C}{\bar{\mu}(z)}f(z)F\circ\psi_{t}(z)\,.
\]
Note that for any $z$ the term $|f(z)/\bar{\mu}(z)|$ is only a change
of speed and that the trajectory of $\mathbb{R}_{+}\ni t\mapsto\psi_{t}(z)$
is independent of this factor, despite the remarkable fact that $\bar{\mu}(z)$
depends on this flow. The result seems fairly natural.

\section{MCMC with integrator snippets\protect\label{sec:MCMC-with-integrator}}

We restrict this discussion to integrator snippet based algorithms,
but more general Markov snippet algorithms could be considered. Consider
again the target distribution
\[
\bar{\mu}({\rm d}z)=\sum_{k=0}^{T}\bar{\mu}(k,{\rm d}z)\,,
\]
with
\[
\bar{\mu}(k,{\rm d}z)=\frac{1}{T+1}\mu_{k}({\rm d}z)\,,
\]
where $\mu_{k}({\rm d}z):=\mu^{\psi^{-k}}({\rm d}z)$ for $k\in\llbracket0,T\rrbracket$. 

\subsection{Windows of states}

Assume that we are in the context of Example~\ref{exa:verlet-HMC},
dropping $n$ for simplicity, the HMC algorithm using the integrator
$\psi^{s}$ is as follows
\begin{equation}
\bar{P}(z,{\rm d}z')=\alpha(z)\delta_{\psi^{s}(z)}({\rm d}z)+\bar{\alpha}(z)\delta_{\sigma(z)}({\rm d}z')\label{eq:pre-windows-of-states-kernel}
\end{equation}
with here
\[
\alpha(z)=\min\left\{ 1,\frac{{\rm d}\bar{\mu}^{\sigma\circ\psi^{s}}}{{\rm d}\bar{\mu}}(z)\right\} =\min\left\{ 1,\frac{\bar{\mu}\circ\psi^{w}(z)}{\bar{\mu}(z)}\right\} 
\]
where the last equality holds when $\upsilon\gg\bar{\mu}$, $\upsilon^{\psi^{s}}=\upsilon$
and we let $\bar{\mu}(z)={\rm d}\bar{\mu}/{\rm d}\upsilon(z)$. In
other works the snippet $\mathsf{z}=\big(z,\psi(z),\psi^{2}(z),\ldots,\psi^{T}(z)\big)$
is shifted along the orbit $\{\psi^{k}(z),k\in\mathbb{Z}\}$ by $\psi^{s}$.
Naturally this needs to be combined with updates of the velocity to
lead to a viable ergodic MCMC algorithm. This can be achieved with
the following kernel, with $\psi^{k}(z)=\big(\psi_{x}^{k}(z),\psi_{v}^{k}(z)\big)$
and $A\in\mathscr{Z}$,
\[
\bar{Q}(z,A):=\sum_{k,l=0}^{T}\bar{\mu}(k\mid z)\frac{1}{T+1}\int\mathbf{1}\{\psi^{-l}\big(\psi_{x}^{k}(z),v'\big)\in A\}\varpi({\rm d}v')\,,
\]
described algorithmically in Alg.~\ref{alg:bar-Q}. Indeed for any
$(l,v')\in\llbracket0,T\rrbracket\times\mathsf{V}$, using identity
(\ref{eq:unfolding-HMC-setup}),
\[
\int\sum_{k=0}^{T}\bar{\mu}(k\mid z)\int\mathbf{1}\{\psi^{-l}\big(\psi_{x}^{k}(z),v'\big)\in A\}\bar{\mu}({\rm d}z)=\int\sum_{k=0}^{T}\int\mathbf{1}\{\psi^{-l}\big(x,v'\big)\in A\}\mu({\rm d}z)\,,
\]
and hence
\begin{align*}
\bar{\mu}\bar{Q}(A) & =\frac{1}{T+1}\sum_{l=0}^{T}\int\mathbf{1}\{\psi^{-l}\big(x,v'\big)\in A\}\varpi({\rm d}v')\mu({\rm d}z)\\
 & =\frac{1}{T+1}\sum_{l=0}^{T}\int\mathbf{1}\{\psi^{-l}\big(x,v\big)\in A\}\mu({\rm d}z)\\
 & =\int\mathbf{1}\{z\in A\}\frac{1}{T+1}\sum_{l=0}^{T}\mu^{\psi^{-l}}({\rm d}z)\\
 & =\bar{\mu}(A)\,.
\end{align*}

Again samples from this MCMC targetting $\bar{\mu}$ can be used to
estimate expectations with respect to $\mu$ using the identity (\ref{eq:unfolding-HMC-setup}).
This is closely related to the ``windows of states'' approach of
\cite{neal1994improved,neal2011mcmc,QIN2001827}, where a window of
states is what we call a Hamiltonian snippet in the present manuscript.
Indeed the windows of states approach corresponds to the Markov update
\begin{equation}
P(z,A)=\sum_{k,l=0}^{T}\frac{1}{T}\int\bar{P}\big(\psi^{-k}(z),{\rm d}z'\big)\bar{\mu}(l\mid z')\mathbf{1}\{\psi^{l}(z')\in A\}\,,\label{eq:windows-of-states-kernel}
\end{equation}
which, we show below, leaves $\mu$ invariant. Indeed, note that for
$k,l\in\llbracket0,T\rrbracket$ and $A\in\mathscr{Z}$,
\begin{align*}
\int\mu({\rm d}z)\int\bar{P}\big(\psi^{-k}(z),{\rm d}z'\big)\bar{\mu}(l\mid z')\mathbf{1}\{\psi^{l}(z')\in A\} & =\int_{A}\mu^{\psi^{-k}}({\rm d}z)\int\bar{P}\big(z,{\rm d}z'\big)\bar{\mu}(l\mid z')\mathbf{1}\{\psi^{l}(z')\in A\}\,,
\end{align*}
and therefore
\begin{align*}
\int\mu({\rm d}z)\int P(z,{\rm d}z')\mathbf{1}\{z'\in A\}= & \int\bar{\mu}({\rm d}z)\bar{P}(z,{\rm d}z')\sum_{l=0}^{T}\bar{\mu}(l\mid z')\mathbf{1}\{\psi^{l}(z')\in A\}\\
= & \int\bar{\mu}({\rm d}z')\sum_{l=0}^{T}\bar{\mu}(l\mid z')\mathbf{1}\{\psi^{l}(z')\in A\}\\
= & \int\mu({\rm d}z)\mathbf{1}\{z\in A\}\,,
\end{align*}
where we obtain the last line from (\ref{eq:unfolding-HMC-setup}).
$P$ is not $\mu$-reversible in general, making theoretical comparisons
challenging.

\begin{algorithm}
Given $z$

Sample $k\sim\bar{\mu}(\cdot\mid z)$ and $l\sim\mathcal{U}(\llbracket0,T\rrbracket)$.

Compute $\psi^{k}(z)=\big(\psi_{x}^{k}(z),\psi_{v}^{k}(z)\big)$.

Sample $v'\sim\varpi(\cdot)$

Return $\psi^{-l}(\psi_{x}^{k}(z),v')$

\caption{Kernel $\bar{Q}$}
\label{alg:bar-Q}
\end{algorithm}

\subsection{Multinomial HMC}

Multinomial HMC is another related MCMC update suggested in \cite{betancourt2017conceptual}
which can be rephrased in terms of our framework as follows. For $(z,A)\in\mathsf{Z}\times\mathscr{Z}$
define the kernel
\[
P(z,A)=\frac{1}{T+1}\sum_{k,l\in\llbracket0,T\rrbracket}\bar{\mu}\big(k\mid\psi^{-l}(z)\big)\mathbf{1}\big\{\psi^{k-l}(z)\in A\big\}\,.
\]
Introducing velocity refreshment before returning $\psi^{k-l}(z)$,
can be reframed in terms $\bar{M}\colon\mathsf{Z}\times\mathscr{Z}$
as defined in (\ref{eq:unfolding-HMC-setup}) leading the the Markov
kernel
\[
P'(z,A)=\frac{1}{T+1}\sum_{l\in\llbracket0,T\rrbracket}\bar{M}\big(\psi^{-l}(z),A\big)\,.
\]
We focus below on the version with no refreshment, $P$. Establishing
$\mu$-invariance is direct since with $z\sim\mu$ and $l\sim\mathcal{U}\big(\llbracket0,T\rrbracket\big)$
then $\psi^{-l}(z)\sim\bar{\mu}$ and with $k\sim\bar{\mu}\big(\cdot\mid\psi^{-l}(z)\big)$
we know that $\psi^{k-l}(z)\sim\mu$ from (\ref{eq:unfolding-HMC-setup}).
More formally
\begin{align*}
\mu P(A) & =\frac{1}{T+1}\sum_{k\in\llbracket0,T\rrbracket}\int\mu({\rm d}z)\sum_{l\in\llbracket0,T\rrbracket}\bar{\mu}\big(k\mid\psi^{-l}(z)\big)\mathbf{1}\big\{\psi^{k}\circ\psi^{-l}(z)\in A\big\}\\
 & =\sum_{k\in\llbracket0,T\rrbracket}\int\bar{\mu}({\rm d}z)\bar{\mu}\big(k\mid z\big)\mathbf{1}\big\{\psi^{k}(z)\in A\big\}\\
 & =\mu(A)\,.
\end{align*}
This Markov kernel is in fact reversibility, that is for $f,g\colon\mathsf{Z}\rightarrow\mathbb{R}$
\begin{align*}
\sum_{k,l\in\llbracket0,T\rrbracket}\int f(z)\mu({\rm d}z)\bar{\mu}\big(k\mid\psi^{-l}(z)\big)g\circ\psi^{k-l}(z)\\
=\frac{1}{T+1}\sum_{k,l\in\llbracket0,T\rrbracket}\int & f\circ\psi^{l}(z)\mu^{\psi^{-l}}({\rm d}z)\frac{{\rm d}\mu^{\psi^{-k}}}{{\rm d}\bar{\mu}}\big(z\big)g\circ\psi^{k}(z)\\
=\frac{1}{T+1}\sum_{k,l\in\llbracket0,T\rrbracket}\int & f\circ\psi^{l}(z)\frac{{\rm d}\mu^{\psi^{-l}}}{{\rm d}\bar{\mu}}(z)\mu^{\psi^{-k}}\big({\rm d}z\big)g\circ\psi^{k}(z)\\
=\frac{1}{T+1}\sum_{k,l\in\llbracket0,T\rrbracket}\int & f\circ\psi^{l-k}(z)\frac{{\rm d}\mu^{\psi^{-l}}}{{\rm d}\bar{\mu}}\circ\psi^{-k}(z)\mu\big({\rm d}z\big)g(z)\\
=\sum_{k,l\in\llbracket0,T\rrbracket}\int & f\circ\psi^{l-k}(z)\mu({\rm d}z)\bar{\mu}\big(l\mid\psi^{-k}(z)\big)g(z)
\end{align*}
A standard SMC implementation relying on this kernel and corresponding
near optimal backward kernel would lead to an algorithm different
from ours. Indeed in this scenario, for $n\in\llbracket P\rrbracket$,
for particles $\{\dot{z}_{n}^{(i)},i\in\llbracket N\rrbracket\}$
the weights are of the standard form $\{{\rm d}\mu_{n+1}/{\rm d}\mu_{n}(\dot{z}_{n+1}^{(i)}),i\in\llbracket N\rrbracket\}$,
which should be contrasted with the weights in Alg.~(\ref{eq:def-Pnk-intro})
$\big\{{\rm d}\bar{\mu}_{n+1}/{\rm d}\mu_{n}\big(z_{n}^{(i)}\big),i\in\llbracket N\rrbracket\big\}$
where resampling ``looks ahead'' and uses the information provided
by the snippets stemming from $\big\{ z_{n}^{(i)},i\in\llbracket N\rrbracket\big\}$.
Note that there is no indexing mistake here and that the additional
dots have been added for coherence: in Alg.~(\ref{eq:def-Pnk-intro})
the output of iteration $n+1$ are $\{\check{z}_{n+1}^{(i)},i\in\llbracket N\rrbracket\}$
and $\big\{ z_{n}^{(i)},i\in\llbracket N\rrbracket\big\}$ intermediate
states to which resampling is applied. One expects robustness to arise
from this difference.

\bibliographystyle{plain}
\bibliography{bib-waste-recycling}

\end{document}